\documentclass[prb,preprint,a4paper]{revtex4}
\usepackage[dvips]{graphicx}
\usepackage{amssymb}
\usepackage{amsmath}
\usepackage{amsthm}

\newcommand{\qbin}[3]{\genfrac{[}{]}{0pt}{0}{#1}{#2}_{#3}}
\newcommand{\pqbin}[4]{\genfrac{[}{]}{0pt}{0}{#1}{#2}_{#3,#4}}
\newcommand{\qpoch}[3]{\left(#1;\,#2\right)_{#3}}
\newcommand{\floor}[1]{\left\lfloor #1 \right\rfloor}
\newcommand{\pqpr}[4]{\mathbb{P}_{#3,#4}\left(#1,\,#2\right)}
\newcommand{\pr}[1]{\mathbb{P}\left(#1\right)}
\newcommand{\pqsum}[3]{\Psi_{#2,#3}\left(#1\right)}
\newcommand{\pqratio}[5]{R_{#4,#5}\left(#1,\,#2,\,#3\right)}
\newcommand{\qratio}[4]{R_{#4}\left(#1,\,#2,\,#3\right)}
\newcommand{\ratio}[3]{R\left(#1,\,#2,\,#3\right)}
\newcommand{\dd}{\mathrm{d}}
\newcommand{\gammaf}[1]{\Gamma\!\left(\scriptstyle #1\right)}
\newcommand{\besselk}[2]{K_{#1}\left(\scriptstyle #2\right)}
\newcommand{\besseli}[2]{I_{#1}\left(\scriptstyle #2\right)}
\newcommand{\dilog}[1]{\mathrm{Li}_2\left(\scriptstyle #1\right)}
\newcommand{\lambert}[1]{\mathrm{W}\left(#1\right)}
\newcommand{\Lambert}[2]{\mathrm{\Omega}_{#2}\left(#1\right)}
\newcommand{\erf}[1]{\mathrm{erf}\left(#1 \right)}
\newcommand{\fc}[1]{\mathrm{\Phi}\!\left(#1\right)}
\newcommand{\mean}[1]{\left\langle #1 \right\rangle}
\newcommand{\abs}[1]{\left\lvert #1 \right\rvert}
\newcommand{\var}[1]{\mathrm{Var}\left(#1\right)}
\newcommand{\xmom}[1]{\varrho_{#1}}
\newcommand{\kmom}[1]{\sigma_{#1}}
\newcommand{\hyper}[3]{{}_1\mathrm{F}_1\!\left(\scriptstyle #1,\,#2,\,#3\right)}
\newcommand{\hyperu}[3]{\mathrm{U}\!\left(\scriptstyle #1,\,#2,\,#3\right)}

\newcommand{\bigo}[1]{\mathcal{O}\left(#1\right)}
\newcommand{\fz}{\mathrm{Z}}
\newcommand{\fZ}{\mathcal{Z}}
\newcommand{\fF}{\mathcal{F}}
\newcommand{\amu}{\bar{\mu}}
\newcommand{\achi}{\bar{\chi}}

\theoremstyle{plain}
\newtheorem{theorem}{Theorem}[section]
\newtheorem{lemma}[theorem]{Lemma}

\DeclareMathOperator{\atanh}{atanh}

\newenvironment{myalgorithm}[1]%
               {\\[2ex]\noindent\textbf{Algorithm} #1\par}%
               {\par\noindent\ignorespacesafterend}

\begin{document}
\title{On the $p,q$-binomial distribution\\
  and the Ising model} 

\date{\today}
\author{P. H. Lundow} 
\email{phl@kth.se} 
\author{A. Rosengren}
\email{roseng@kth.se} 
\affiliation{
  Condensed Matter Theory, Department of Theoretical Physics,\\
  AlbaNova University Center, KTH, SE-106 91 Stockholm, Sweden }

\begin{abstract}
  A completely new approach to the Ising model in 1 to 5 dimensions is
  developed.  We employ $p,q$-binomial coefficients, a generalisation
  of the binomial coefficients, to describe the magnetisation
  distributions of the Ising model. For the complete graph this
  distribution corresponds exactly to the limit case $p=q$. We take
  our investigation to the simple $d$-dimensional lattices for
  $d=1,2,3,4,5$ and fit $p,q$-binomial distributions to our data, some
  of which are exact but most are sampled.  For $d=1$ and $d=5$ the
  magnetisation distributions are remarkably well-fitted by
  $p,q$-binomial distributions. For $d=4$ we are only slightly less
  successful, while for $d=2,3$ we see some deviations (with
  exceptions!)  between the $p,q$-binomial and the Ising
  distribution. We begin the paper by giving results on the behaviour
  of the $p,q$-distribution and its moment growth exponents given a
  certain parameterization of $p,q$.  Since the moment exponents are
  known for the Ising model (or at least approximately for $d=3$) we
  can predict how $p,q$ should behave and compare this to our measured
  $p,q$. The results speak in favour of the $p,q$-binomial
  distribution's correctness regarding their general behaviour in
  comparison to the Ising model. The full extent to which they
  correctly model the Ising distribution is not settled though. 
\end{abstract}

\maketitle

\section{Introduction}
Choose a graph, e.g. a square lattice, on $n$ vertices and compute its
Ising partition function $\fZ$, keeping track of its terms according
to their magnetisation, so that $\fZ=\fZ_0 + \fZ_1+\cdots +\fZ_n$. The
quotient $\fZ_k/\fZ$ is the probability of having $k$ negative spins,
or magnetisation $M=n-2\,k$. The controlling parameter of the
partition function is the temperature, that is, for any given
temperature the partition function provides us with a distribution of
magnetisations. At infinite temperature this is simply the binomial
distribution. At zero temperature, on the other hand, we receive a
distribution with two peaks, one at $k=0$ and one at $k=n$, both with
50\% of the probability mass. The distribution is always
symmetrical. What happens between these two extreme temperatures?
Though there are exceptions to the rule, for most graphs the
distribution begins its life at high temperatures as a unimodal
distribution with the peak at the middle $k=n/2$. As we lower the
temperature the distribution gets increasingly wider until we reach a
temperature where the distribution changes from unimodal to
bimodal. Near this temperature, slightly above and slightly below, the
distribution is particularly wide.  Lowering the temperature even
further the distribution develops two sharp peaks, both essentially
gaussian, and the peaks move outwards.

This article tries to model the distributions using $p,q$-binomial
coefficients.  In one case, the complete graph, they model exactly the
distributions and for $d$-dimensional lattices the dimension seems to
determine how well they fit to the Ising distributions. Especially in
the case of $d=1$ and $d=5$ the Ising distributions are particularly
well-fitted by the $p,q$-binomial coefficients.

The paper begins in section~\ref{sec:definitions} by providing the
necessary basic tools, such as the Pochhammer and $q$-Pochhammer
symbol, $q$-binomial coefficients and finally the $p,q$-binomial
coefficients. Some nice results on their properties are also stated,
even though they are of no direct use to us in the rest of the
paper. They are merely intended to give the reader a feel of how
$p,q$-binomials behave. In section~\ref{sec:pqdist} we define the
$p,q$-binomial distribution and give an algorithm for finding values
of $p$ and $q$ when the distribution is given as input.  The problem
is to determine an optimal choice for $p$ and $q$. As it turns out we
only have to focus on the value of the probability and the location of
the distribution's peaks, at least for a bimodal distribution. The
unimodal distribution always has its peak at the middle so in this
case we instead take the quotient between the two middle probabilities
as controlling parameter. This quotient is unfortunately rather
sensitive to noise, making it difficult to determine the parameters
$p$ and $q$ for sampled data.

Section~\ref{sec:pequalq} gives detailed results for the special case
when $p=q$. This case corresponds exactly to the complete graph and is
the only case where we can give asymptotically exact expressions for
the sum of the coefficients.  In section~\ref{sec:pochhammer} we
provide some useful tools for working with $p,q$-binomial coefficients
in the case when $p\ne q$. After this build-up of tools we are, at
long last, ready to give some general results on the distribution of
$p,q$-binomial coefficients in section~\ref{sec:isocurves}.  Using the
parameterization $p=1+y/n$ and $q=1+z/n$ we find the asymptotic value
of $y$, given $z$, where the distribution is flat in the middle. We
also allow for a small change in $y$, using a higher order parameter
$a$, so that we can follow properly how the distribution changes from
unimodal to bimodal.  However, the computations that we rely on
involve some rather complicated series expansions that were made using
Mathematica. These are much too long to fit into this paper. We have
prepared a simplified Mathematica notebook that performs all the
necessary computations. The interested reader can obtain it by
contacting the first author.

Section~\ref{sec:slope} looks into the case of moving $y$ and $z$
along a line with any given slope as opposed to the previous section
where $z$ stays fixed.  In section~\ref{sec:moments} we give exact
scaling formulae for the moments of the distributions depending on the
parameters $a$ and $z$.  For a given moment of these distributions we
always obtain the same exponent on $n$, regardless of $a$ and $z$. In
section~\ref{sec:exponents} we try to remedy this by letting the
previously fixed parameter $z$ depend ever so slightly (at most
logarithmically) on $n$. This is based on the assumption that the
previous formulae in section~\ref{sec:isocurves} still hold.  However,
we can now adjust the exponent of $n$ though this comes at the cost of
an extremely slow convergence.

Section~\ref{sec:ising} defines the Ising model, laying the ground for
studying distributions of magnetisations, the intended application of
our endeavour. In section~\ref{sec:lattices} we apply our tools to the
$d$-dimensional lattice graphs for $d=1,2,3,4,5$ fitting
$p,q$-distributions to simulated distributions and comparing them.

A condensed reading, more suitable to the reader who is pressed for
time, should include a look at \eqref{qbin}, \eqref{pqbin},
\eqref{pqsum}, \eqref{pqpr}, \eqref{pqratio}, \eqref{qubound},
\eqref{qlbound}, \eqref{pqratio2} for the necessary definitions and
results concerning the basics. After that, the most important results
are stated in equations \eqref{pqratio3}, \eqref{pqratio5},
\eqref{kmomm}, \eqref{logmom2} and \eqref{loglogmom2}.  After looking
up the basic definitions regarding the Ising model the reader can skip
to \eqref{zkn}. In Section~\ref{sec:lattices} the reader can now pick
and choose his favourite lattice and look at the pictures.

\section{Definitions, notations, the very basics}\label{sec:definitions}
The $q$-binomial coefficient
\begin{equation}\label{qbin}
  \qbin{n}{k}{q} = 
  \prod_{i=1}^k \frac{1 - q^{n-i+1}}{1 - q^i}, 
  \quad \textrm{$q \ne 1$,  $0\le k\le n$}
\end{equation} 
is a natural extension of the standard binomial coefficient
\begin{equation}
  \binom{n}{k} = \frac{n!}{k!\,(n-k)!},\quad 0\le k\le n
\end{equation}
The Pochhammer symbol (or shifted factorial) is defined as
\begin{equation} 
  (a)_n = \prod_{i=0}^{n-1} (a+i)
\end{equation}
so that $(1)_n=n!$.  Its $q$-deformed relative, the $q$-Pochhammer
symbol, is defined as
\begin{equation} 
  \qpoch{a}{q}{n} = 
  \prod_{i=0}^{n-1} \left(1-a\,q^i\right), \quad n\ge 0
\end{equation}
The $q$-binomial coefficient can then be expressed as
\begin{equation}\label{qbin2}
  \qbin{n}{k}{q} =
  \frac{\qpoch{q}{q}{n}}{\qpoch{q}{q}{k}\,\qpoch{q}{q}{n-k}} =
  \frac{\qpoch{q^{n-k+1}}{q}{k}}{\qpoch{q}{q}{k}}
\end{equation}
The $q$-numbers are defined for any real number $a$ as
\begin{equation}\label{qnum}
  [a]_q = \frac{1-q^a}{1-q}, \quad q\ne 1
\end{equation} 
and it is easy to show that 
\begin{equation}
  \lim_{q\to 1} [a]_q = a 
\end{equation}
Note that for integers $n\ge 1$ we have
\begin{equation}\label{qnum1}
  [n]_q = \frac{1-q^n}{1-q} = 1 + q + \cdots + q^{n-1}, \quad q\ne 1
\end{equation}
so that $[n]_q \to n$ when $q\to 1$.
To continue, the $q$-number factorials are defined as
\begin{equation}
  [n]_q! = \prod_{k=1}^n [k]_q
\end{equation}
and then obviously 
\begin{equation} 
  \lim_{q\to 1} [n]_q! = n!
\end{equation} 
The $q$-binomial can now be defined in an alternative way as
\begin{equation}
  \qbin{n}{k}{q} =
  \frac{[n]_q!}{[k]_q!\,[n-k]_q!}
\end{equation}
and then it follows that 
\begin{equation} 
  \lim_{q\to 1} \qbin{n}{k}{q} = \binom{n}{k}
\end{equation} 
Quite analogously it is easy to verify that
\begin{equation}
  \lim_{q\to 1}\frac{\qpoch{q^a}{q}{n}}{(1-q)^n} = (a)_n 
\end{equation} 
Finally, note also that $\qbin{n}{k}{q}$ can be viewed as a formal
polynomial in $q$ of degree $k\,(n-k)$ where the coefficient of $q^j$
counts the number of $k$-subsets of $\{1,\ldots,n\}$ with element sum
$j+k\,(k+1)/2$. It is thus a polynomial with positive coefficients.

We have so far only stated what belongs to the standard repertoire on
the subject.  The $q$-binomials and the $q$-Pochhammer function have
many interesting properties and we point the interested reader to the
books \cite{gasper:04}, \cite{andrews:76}, \cite{andrews:99} and
especially the charming little book \cite{andrews:86}. For more on the
standard binomial coefficient and Pochhammer function we recommend
\cite{graham:94} which contains a wealth of useful information.

A natural extension of the $q$-binomial coefficient, the
$p,q$-binomial coefficient, was defined in \cite{corcino:08} as
\begin{equation}\label{pqbin}
  \pqbin{n}{k}{p}{q} = 
  \prod_{i=1}^k \frac{p^{n-i+1} - q^{n-i+1}}{p^i - q^i}, 
  \quad\textrm{$p \ne q$,\, $0\le k\le n$}
\end{equation}
Clearly, in the case $p=1$ this reduces
to a $q$-binomial coefficient. Also, note that $p$ and $q$ are
interchangable so that
\begin{equation}
  \pqbin{n}{k}{p}{q} = \pqbin{n}{k}{q}{p}
\end{equation}
Just as the standard binomial coefficients, their $p,q$-analogues are
also symmetric
\begin{equation}
  \pqbin{n}{k}{p}{q} = \pqbin{n}{n-k}{p}{q}
\end{equation}

It is an easy exercise to show the following identity and we leave this
to the reader.
\begin{equation}\label{pqvq}
  \pqbin{n}{k}{p}{q} = 
  p^{k\,(n-k)}\,\qbin{n}{k}{q/p} = 
  q^{k\,(n-k)}\,\qbin{n}{k}{p/q}
\end{equation}
As a corollary it follows that
\begin{equation}\label{pqlimit}
  \lim_{p,q\to r} \pqbin{n}{k}{p}{q} = r^{k\,(n-k)}\,\binom{n}{k} 
\end{equation}

Any identity involving $q$-binomial coefficients can then be extended
to a $p,q$-binomial identity by first replacing $q$ with $q/p$ and
then use the identity \eqref{pqvq}.  For example, the binomial theorem
\begin{equation}
  \left( 1 + x \right)^n = \sum_{k=0}^n \binom{n}{k}\,x^k
\end{equation}
has the $q$-analogue 
\begin{equation}\label{eq5}
  \prod_{\ell=0}^{n-1} \left(1+x\,q^{\ell}\right) = 
  \sum_{k=0}^n \qbin{n}{k}{q}\,q^{\binom{k}{2}}\,x^k
\end{equation}
Now replace $q$ with $q/p$ and then use the identity
\begin{equation}
  \qbin{n}{k}{q/p} = p^{-k\,(n-k)}\,\pqbin{n}{k}{p}{q}
\end{equation}
from \eqref{pqvq} above.  This gives us
\begin{equation}
  \prod_{\ell=0}^{n-1} \left(1+x\,(q/p)^{\ell}\right) = 
  \sum_{k=0}^n \pqbin{n}{k}{p}{q}\,p^{-k\,(n-k)}\,(q/p)^{\binom{k}{2}}\,x^k
\end{equation}
After multiplying both sides with $p^{\binom{n}{2}}$ this simplifies
into
\begin{equation}\label{eq10}
  \prod_{\ell=0}^{n-1} \left(p^{\ell}+x\,q^{\ell}\right) = 
  \sum_{k=0}^n \pqbin{n}{k}{p}{q}\,p^{\binom{n-k}{2}}\,q^{\binom{k}{2}}\,x^k
\end{equation}
which was also shown in \cite{corcino:08} using a recursion technique.

As another application we consider the Chu-Vandermonde identity
\begin{equation}\label{vdm}
  \binom{m+n}{k} = \sum_{\ell=0}^k \binom{m}{k-l}\,\binom{n}{\ell}
\end{equation}
which is a fairly direct consequence of the binomial theorem applied
to the product $(1+x)^m\,(1+x)^n$. The $q$-Vandermonde identity, see
e.g. \cite{gasper:04} and (though misprinted) \cite{corcino:08}, can
be stated as
\begin{equation}\label{qvdm}
  \qbin{m+n}{k}{q} = 
  \sum_{\ell=0}^k \qbin{m}{k-\ell}{q}\,\qbin{n}{\ell}{q}\,q^{\ell\,(m-k+\ell)}
\end{equation}
Using \eqref{pqvq} above we can now obtain a $p,q$-analog of this.
\begin{theorem}
  \begin{equation}\label{pqvdm}
    \pqbin{m+n}{k}{p}{q} = 
    \sum_{\ell=0}^k \pqbin{m}{k-\ell}{p}{q}\,\pqbin{n}{\ell}{p}{q}\,
    p^{(n-\ell)\,(k-\ell)}\,q^{\ell\,(m-k+\ell)}
  \end{equation}
\end{theorem}
\begin{proof}
  In the $q$-Vandermonde identity \eqref{qvdm}, replace $q$ with $q/p$
  and multiply both sides with $p^{k\,(m+n-k)}$. Using \eqref{pqvq}
  the left hand side is now a pure $p,q$-binomial coefficient. The
  $\ell$th term of the right hand side is
  \begin{gather*}
    \qbin{m}{k-\ell}{q/p}\,\qbin{n}{\ell}{q/p}\,
    p^{k\,(m+n-k)}\,(q/p)^{\ell\,(m-k+\ell)} = \\
    \qbin{m}{k-\ell}{q/p}\,\qbin{n}{\ell}{q/p}\,
    p^{k\,(m+n-k)-\ell\,(m-k+\ell)}\,q^{\ell\,(m-k+\ell)} = \\
    p^{(k-\ell)\,(m-(k-\ell))}\,\qbin{m}{k-\ell}{q/p}\,
    p^{\ell\,(n-\ell)}\,\qbin{n}{\ell}{q/p}\,
    p^{(n-\ell)\,(k-\ell)}\,q^{\ell\,(m-k+\ell)} = \\
    \pqbin{m}{k-\ell}{p}{q}\,\pqbin{n}{\ell}{p}{q}\,
    p^{(n-\ell)\,(k-\ell)}\,q^{\ell\,(m-k+\ell)}
  \end{gather*}
  and the theorem follows.
\end{proof}

The $q$-binomial coefficients have been shown to form a log-concave
(and thus unimodal) sequence for $q\ge 0$, see e.g. \cite{butler:90}
and \cite{krattenthaler:89}. However, for the $p,q$-binomial
coefficients this does not always hold. Rewriting them as a product
like in \eqref{pqvq} we have in fact a product of two sequences; that
of $p^{k\,(n-k)}$ and $\qbin{n}{k}{q/p}$ for $k=0,\ldots,n$. The first
sequence is log-concave for $p\ge 1$ and log-convex for $p\le 1$. It
is well-known that the element-wise product of two log-concave
positive sequences is also log-concave. So, if $p\ge 1$ and $q \ge 0$
then the sequence of $\pqbin{n}{k}{p}{q}$ is log-concave.

We conjecture that for $p,q>0$ the sequence can be either unimodal,
with the maximum at $k=\floor{n/2}$, or bimodal, with the maxima at
$k$ and $n-k$ for some $0\le k\le n/2$, but not trimodal etc. We will
assume this to be true in this paper but a formal proof is still
lacking. Note that if we allow negative values of $p$ the sequence can
have a local maximum at every alternate index $k$.

We write $f(n)\sim g(n)$ to denote that $f(n)/g(n)\to 1$ as
$n\to\infty$.  Analogously $f(n)\propto g(n)$ denotes that
$f(n)/g(n)\to A$, for some non-zero real number $A$, as $n\to\infty$.

\section{The $p,q$-binomial distribution}\label{sec:pqdist}
Now we are ready to introduce the notation
\begin{equation}\label{pqsum}
  \pqsum{n}{p}{q} = \sum_{k=0}^n \pqbin{n}{k}{p}{q}
\end{equation}
and define the $p,q$-binomial probability function
\begin{equation}\label{pqpr}
  \pqpr{n}{k}{p}{q} = 
  \frac{
    \pqbin{n}{k}{p}{q}
  }{
    \pqsum{n}{p}{q}
  }
\end{equation}
The reader should here observe that the sum of the coefficients has,
to the best of our knowledge, no simpler expression in the general
case. Neither do the sum of the $q$-binomial coefficients have a
simpler expression that we are aware of.  Compare this with the case
of the standard binomial coefficients for which the sum is simply
$2^n$.

Having made our assumption of unimodality/bimodality we can now set up
a simple computational scheme to find values of $p$ given $q$.  First
we need to define a highly useful quantity; the ratio between two
coefficients
\begin{equation}\label{pqratio}
  \pqratio{n}{k}{\ell}{p}{q} = 
  \frac{\pqbin{n}{k-\ell}{p}{q}}{\pqbin{n}{k}{p}{q}},
  \quad 0\le\ell\le k\le n/2
\end{equation}
In the special case when $\ell=1$ we are looking at two
consecutive coefficients. The ratio then becomes
\begin{equation}\label{pqratio1}
  \begin{split}
    \pqratio{n}{k}{1}{p}{q} = 
    \frac{\pqbin{n}{k-1}{p}{q}}{\pqbin{n}{k}{p}{q}} &= \\
    \frac{p^k - q^k}{p^{n-k+1} - q^{n-k+1}} &= \\
    p^{-(n-2\,k+1)}\,\frac{1-\left(q/p\right)^k}{1-\left(q/p\right)^{n-k+1}}&=\\
    q^{-(n-2\,k+1)}\,\frac{1-\left(p/q\right)^k}{1-\left(p/q\right)^{n-k+1}}
  \end{split}
\end{equation}
\begin{lemma}
  Let $k,n$ be positive integers such that $1\le k \le n-k$.
  If $x>1$ then
  \[0 < \frac{1-x^k}{1-x^{n-k+1}} < \frac{k}{n-k+1}\]
  and if $0<x<1$  then
  \[\frac{k}{n-k+1} < \frac{1-x^k}{1-x^{n-k+1}} < 1 \]
\end{lemma}
\begin{proof}
  Obviously we have
  \[0 < \frac{1-x^k}{1-x^{n-k+1}} < 1 \]
  for all $x>0$ so we proceed to the other inequalities instead.
  Define 
  \[
  y = \frac{x^k+x^{k+1}+\cdots+x^{n-k}}{1+x+\cdots+x^{k-1}}
  \]
  so that
  \[
  \frac{1}{1+y} =
  \frac{1+x+\cdots+x^{k-1}}{1+x+\cdots+x^{n-k}} = 
  \frac{1-x^k}{1-x^{n-k+1}}
  \]
  If $x>1$ then
  \[
  y \ge \frac{x^k+x^k+\cdots+x^k}{x^{k-1}+x^{k-1}+\cdots+x^{k-1}} =
  \frac{(n-2\,k+1)\,x^k}{k\,x^{k-1}} = 
  \frac{n-2\,k+1}{k}\,x >
  \frac{n-2\,k+1}{k}
  \]
  and then
  \[\frac{1}{1+y} < \frac{k}{n-k+1}\]
  If $0<x<1$, then completely analogously
  \[
  y \le \frac{n-2\,k+1}{k}\,x < \frac{n-2\,k+1}{k}
  \]
  and thus
  \[\frac{1}{1+y} > \frac{k}{n-k+1}\]
  and the lemma follows.
\end{proof}

\begin{theorem}
  Let $k$ and $n$ be positive integers such that $k \le n-k$.
  For $p>q>0$ we have
  \[
  0 < 
  \frac{\pqbin{n}{k-1}{p}{q}}{\pqbin{n}{k}{p}{q}} < 
  q^{-(n-2\,k+1)}\,\frac{k}{n-k+1}
  \]
  while for $0<p<q$ we have
  \[
  q^{-(n-2\,k+1)}\,\frac{k}{n-k+1} <
  \frac{\pqbin{n}{k-1}{p}{q}}{\pqbin{n}{k}{p}{q}} < 
  q^{-(n-2\,k+1)}
  \]
\end{theorem}

Given $q$ and a coefficient ratio $r$ at coefficient $k$ we can now
find the correct $p$ through a simple iteration scheme: If
\[0<r<q^{-(n-2\,k+1)}\frac{k}{n-k+1}\]
then $p>q>0$ and we use
\begin{equation}\label{pqiter}
  p \leftarrow \left( q^k - r\,q^{n-k+1} + r\,p^{n-k+1}\right)^{1/k} 
\end{equation}
which is obtained from setting \eqref{pqratio1} to $r$.  As start
value of $p$ we may use a number slightly larger than $q$.  If
\[q^{-(n-2\,k+1)}\frac{k}{n-k+1} < r < q^{-(n-2\,k+1)} \]
then $0<p<q$ and we use
\begin{equation}
  p \leftarrow 
  \left(\frac{p^k}{r}+q^{n-k+1}-\frac{q^k}{r}\right)^{\frac{1}{n-k+1}} 
\end{equation}
and use $0$ as the starting value for $p$.  To prove that these
iteration schemes actually converge one would have to show that their
derivatives with respect to $p$ is at most $1$ using the start value.
Since we have no such proof we will leave it at that and just claim
that they are practical.

Given a $p,q$-binomial distribution, or a distribution that we wish to
approximate by a $p,q$-binomial distribution, can we find the pair
$p,q>0$ that generated it such that the distribution of $p,q$-binomial
coefficients have the correct probability and ratio $r$ at coefficient
$k$?  An iteration method again solves this problem practically under
the assumption $p>q$. Suppose this input distribution has the
probabilitites $\pr{0},\pr{1}\,\ldots,\pr{n}$. Let $k$, $\pr{k}$,
$r=\pr{k-1}/\pr{k}$ and an $\epsilon$ be given as input parameters.
\begin{myalgorithm}{$p,q$-Find}
\begin{enumerate}
  \item Assign $q_{\min}\leftarrow 0$ and $q_{\max} \leftarrow
    \left(\frac{k}{r\,(n-k+1)}\right)^{\frac{1}{n-2\,k+1}}$.
  \item $q\leftarrow(q_{\min} + q_{\max})/2$
  \item Compute the corresponding $p$ as in the method above.
  \item If $\pqpr{n}{k}{p}{q}<\pr{k}$ then $q_{\min}\leftarrow q$,
    otherwise $q_{\max}\leftarrow q$.
  \item If $q_{\max}-q_{\min} < \epsilon$ then exit loop, otherwise
    jump to step 2.
\end{enumerate}
\end{myalgorithm}
This method seems to work best when $\pr{k}$ is one of the maximum
probabilities.  However, if the distribution is unimodal so that the
maximum probability is at $k=n/2$, then the scheme will depend heavily
on the quality of $r$. On the other hand, if the distribution is
bimodal then this problem goes away and we may simply set $r=1$,
unless $n$ is too small.

It is implied, though we do not have a proof, that increasing $q$
while keeping $k$ and $r$ fixed also increases the probability
$\pqpr{n}{k}{p}{q}$. It actually increases until $p=q$ which then
constitutes an interesting limit case, which we will deal with in
section~\ref{sec:pequalq}.

\section{The special case $p=q$}\label{sec:pequalq}
We will extend the definition of the $p,q$-binomial coefficients in
\eqref{pqbin} to include also the limiting case when $p=q$ as in
\eqref{pqlimit}. Thus we will define
\begin{equation}\label{pqbin2}
  \pqbin{n}{k}{q}{q} = q^{k\,(n-k)}\,\binom{n}{k}
\end{equation}
Let us look particularly at the point $q$ where the coefficient ratio
is $1$ at $k=n/2$, i.e. $\pqratio{n}{n/2}{1}{q}{q} = 1$. Also, henceforth
we will assume that $n$ is even to simplify some calculations.
\begin{lemma}\label{lem:pqstar}
  For $q=\frac{n}{n+2}$ we have $\pqratio{n}{n/2}{1}{q}{q} = 1$.
\end{lemma}
\begin{proof}
  In the case when $p=q$ the ratio is 
  \[\pqratio{n}{k}{1}{q}{q} = q^{-(n-2\,k+1)}\,\frac{k}{n-k+1} \]
  and for $k=n/2$ we have
  \[\pqratio{n}{n/2}{1}{q}{q} = q^{-1}\,\frac{n}{n+2}\]
  Setting this to $1$ gives the lemma.
\end{proof}
What is the sum of the coefficients at this point? To answer this we
compare the middle coefficient with a coefficient situated at some
carefully chosen distance from the middle.  How big is the middle
coefficient? Note first that
\begin{equation}
  \pqbin{n}{n/2}{q}{q} = q^{\frac{n^2}{4}}\,\binom{n}{n/2}
\end{equation}
\begin{lemma}\label{lem:pqmid0}
  For $q=\frac{n}{n+2}$ we have 
  \[
  \pqbin{n}{n/2}{q}{q} \sim
  \sqrt{\frac{2\,e}{\pi\,n}}\,\left(\frac{2}{\sqrt{e}}\right)^n
  \]
\end{lemma}
meaning that the quotient between the left- and right-hand side goes
to $1$ as $n\to\infty$. The proof follows from an easy application of
the identity
\begin{equation}\label{expid}
  \left(1+\frac{x}{n}\right)^n = e^x\,\left(1-\frac{x^2}{2\,n}+
  \frac{x^3}{3\,n^2}+\frac{x^4}{8\,n^2}+\cdots\right)
\end{equation}
and we leave it to the reader. A somewhat more involved application of
\eqref{expid} is the following lemma
\begin{lemma}\label{lem:pqkvot0}
  Let $x$ be some real number.
  For $q=\frac{n}{n+2}$ we have
  \[
  \pqratio{n}{n/2}{x\,n^{3/4}}{q}{q} \sim
  \exp\left(-\frac{4}{3}\,x^4\right)
  \]
\end{lemma}
This allows us to give the exact order of the sum.
\begin{theorem}\label{thm:pqsum0}
  For $q=\frac{n}{n+2}$ we have
  \[
  \pqsum{n}{q}{q} \sim 
  \frac{\gammaf{1/4}\,3^{1/4}\,n^{1/4}}{\sqrt{\pi}}\,
  \left(\frac{2}{\sqrt{e}}\right)^{n-1}
  \]
\end{theorem}
\begin{proof}
  The calculations goes as follows though we leave out some details.
  \begin{gather*}
    \pqsum{n}{q}{q} = 
    \pqbin{n}{n/2}{q}{q}\,\sum_{k=-n/2}^{n/2} 
    \frac{
      \pqbin{n}{n/2+k}{q}{q}
    }{
      \pqbin{n}{n/2}{q}{q}
    } \sim \\
    n^{3/4}\, \pqbin{n}{n/2}{q}{q}\, 
    \int\limits_{-\infty}^{+\infty}\exp\left(-\frac{4}{3}\,x^4\right)\,\dd x 
    \sim \\
    n^{3/4}\,\sqrt{\frac{2\,e}{\pi\,n}}\,\left(\frac{2}{\sqrt{e}}\right)^n\,
    \frac{3^{1/4}\,\gammaf{1/4}}{2\,\sqrt{2}} = \\
    \frac{\gammaf{1/4}\,3^{1/4}\,n^{1/4}}{\sqrt{\pi}}\,
    \left(\frac{2}{\sqrt{e}}\right)^{n-1}
  \end{gather*}
  where the factor $n^{3/4}$ in front of the integral comes from the
  change of variables $k=x\,n^{3/4}$.
\end{proof}
This is the sum at just a single point where the distribution becomes
flat in the middle region. We can do even better if we allow ourselves
to move around in the vicinity of this point.
\begin{lemma}\label{lem:pqmid1}
  Let
  \[ q = \frac{n}{n+2} + \frac{a}{n^{3/2}} \]
  for some real number $a$.
  Then
  \[
  \pqbin{n}{n/2}{q}{q} \sim 
  \sqrt{\frac{2\,e}{\pi\,n}}\,\left(\frac{2}{\sqrt{e}}\right)^n\,
  \exp\left(\frac{a\,\sqrt{n}}{4}\right)
  \]
\end{lemma}
This can be verified using \eqref{expid} as can the following lemma.
\begin{lemma}\label{lem:pqkvot1}
  Let $a$ and $x$ be real numbers.
  For $q=\frac{n}{n+2}+\frac{a}{n^{3/2}}$ we have
  \[
  \pqratio{n}{n/2}{x\,n^{3/4}}{q}{q} \sim
  \exp\left(-a\,x^2-\frac{4}{3}\,x^4\right)
  \]
\end{lemma}
We now have the resources to estimate the sum of the coefficients for
a whole spectrum of values of $q$ near $n/(n+2)$.  The next theorem
can be shown using the same technique as Theorem~\ref{thm:pqsum0},
though the result gets slightly more complicated due to the integral on
the right hand side in the previous lemma.
\begin{theorem}\label{thm:pqsum1}
  Let $q=\frac{n}{n+2}+\frac{a}{n^{3/2}}$.
  For $a>0$ the asymptotic order of $\pqsum{n}{q}{q}$ is
  \[
    \frac{n^{1/4}}{4}\,
    \sqrt{\frac{6\,a\,e}{\pi}}\,
    \left(\frac{2}{\sqrt{e}}\right)^n\,
    \exp\left(\frac{a\,\sqrt{n}}{4}+\frac{3\,a^2}{32}\right)\,
    \besselk{1/4}{\frac{3\,a^2}{32}}
  \]
  For $a<0$ the asymptotic order of $\pqsum{n}{q}{q}$ is
  \[
    \frac{n^{1/4}}{4}\,
    \sqrt{-3\,a\,e\,\pi}\,
    \left(\frac{2}{\sqrt{e}}\right)^n\,
    \exp\left(\frac{a\,\sqrt{n}}{4}+\frac{3\,a^2}{32}\right)\,
    \left(\besseli{1/4}{\frac{3\,a^2}{32}} + 
    \besseli{-1/4}{\frac{3\,a^2}{32}}\right)
  \]
\end{theorem}
Here $\besseli{\alpha}{x}$ and $\besselk{\alpha}{x}$ denote the
modified Bessel functions of the first and second kind respectively.
The $m$th moment is simpler to express using an integral formulation.
\begin{gather}
  \sum_{k=-n/2}^{n/2}\abs{k}^m \,
  \pqbin{n}{\frac{n}{2}+k}{q}{q} \sim \\
  n^{\frac{3\,m+3}{4}}\,\pqbin{n}{n/2}{q}{q}\,
  \int\limits_{-\infty}^{+\infty} \abs{x}^m, 
  \exp\left(-a\,x^2-\frac{4}{3}\,x^4\right)\,
  \dd x
\end{gather}
and the asymptotic behaviour of the middle coefficient is given by
lemma~\ref{lem:pqmid1}. 

The same technique allows us to repeat this for points farther away
from the critical point $q=n/(n+2)$. If we increase $q$ by $a/n$ then
the coefficients get sharply concentrated in the middle like that of
standard binomial coefficients. Again \eqref{expid} to the rescue.
\begin{lemma}\label{lem:pqmid2}
  Let $q=\frac{n+a}{n+2}$ where $a>0$. Then 
  \[
  \pqbin{n}{n/2}{q}{q} \sim 
  \sqrt{\frac{2\,e}{\pi\,n}}\,\left(\frac{2}{\sqrt{e}}\right)^n\,
  \exp\left(\frac{a\,n}{4}-\frac{a^2}{8}\right)
  \]
\end{lemma}
Note that above, when the coefficients had a rather wide distribution,
we examined their behaviour at $x\,n^{3/4}$ from the middle. Under our
current assumption of $q$ the distribution gets more sharply
concentrated around the middle (basically they become gaussian), thus
we study their behaviour at $x\,\sqrt{n}$ from the middle.
\begin{lemma}\label{lem:pqkvot2}
  Let $q=\frac{n+a}{n+2}$ and $a>0$. Then 
  \[
  \pqratio{n}{n/2}{x\,\sqrt{n}}{q}{q} \sim
  \exp\left(-a\,x^2\right)
  \]
\end{lemma}
\begin{theorem}\label{thm:pqsum2}
  Let $q=\frac{n+a}{n+2}$ where $a>0$.
  Then
  \begin{gather*}
    \pqsum{n}{q}{q} \sim \sqrt{n}\,\pqbin{n}{n/2}{q}{q}\,
    \int\limits_{-\infty}^{+\infty}
    \exp\left(-a\,x^2\right)\,\dd x \sim \\
    \sqrt{\frac{2\,e}{a}}\,\left(\frac{2}{\sqrt{e}}\right)^n\,
    \exp\left(\frac{a\,n}{4}-\frac{a^2}{8}\right)
  \end{gather*}
\end{theorem}
If we decrease $a$ below zero the sequence becomes sharply bimodal,
with all its mass concentrated around two peaks. We can of course
connect the position of the peaks with the parameter $a$.  Suppose
that we want one of the peaks to have its maximum located at $k$ and
$k-1$, that is, let the ratio here be $1$.
\begin{lemma}\label{lem:pqmag0}
  Given a number $\mu$ such $0<\abs{\mu}<1$ let
  \[k=\frac{n}{2}\left(1+\mu\right)\]
  If
  \[
  a = 2\,\left(1-\frac{\atanh{\mu}}{\mu}\right)
  \]
  then
  \[
  \lim_{n\to\infty}\pqratio{n}{k}{1}{q}{q} = 1
  \]
\end{lemma}
\begin{proof}
  A simple calculation shows that the limit of the ratio is 
  \[
  \lim_{n\to\infty}\pqratio{n}{k}{1}{q}{q} = 
  \frac{1+\mu}{1-\mu}\,\exp\left(\mu\,(a-2)\right)
  \]
  Setting the limit to 1 and solving the equation gives the lemma.
\end{proof}
We continue as before and estimate the growth rate of the peak
coefficient.  The result (and the proof) is somewhat more complicated
but follows from an application of \eqref{expid}.
\begin{lemma}\label{lem:pqmid3}
  Let $\mu$ and $a$ be defined as in lemma~\ref{lem:pqmag0} and
  set $q=\frac{n+a}{n+2}$. Then
  \begin{gather*}
    \pqbin{n}{\frac{n}{2}\,\left(1+\mu\right)}{q}{q} \sim
    \frac{\scriptstyle
      \sqrt{2}\,\exp\left\{
      \frac{n}{2}\,
      \left(
      \log\frac{4}{1-\mu^2} -
      \frac{1+\mu^2}{\mu}\,\atanh{\mu}
      \right) + 
      \frac{
        1-\mu^2
      }{
        2\,\mu^2
      }\,
      \left(2\,\mu-\atanh{\mu} \right)\,\atanh{\mu}
      \right\}
    }{
      \sqrt{\pi\,n\,\left(1-\mu^2\right)}
    }
  \end{gather*}
\end{lemma}
Again we take $x\,\sqrt{n}$ steps away from the peak and find the
shape of the distribution.
\begin{lemma}\label{lem:pqkvot3}
  Let $\mu$ and $a$ be defined as in lemma~\ref{lem:pqmag0} and
  set $q=\frac{n+a}{n+2}$. Then
   \[
   \pqratio{n}{\frac{n}{2}\,\left(1+\mu\right)}{x\,\sqrt{n}}{q}{q} \sim
   \exp\left\{
   2\,x^2\,\left(
   \frac{1}{\mu^2-1}+\frac{\atanh{\mu}}{\mu}
   \right)
   \right\}
   \]
\end{lemma}
Finally we get the sum by multiplying the integral with the peak
coefficient and $2\,\sqrt{n}$, where the factor 2 is due to that we
have two peaks.
\begin{theorem}\label{thm:pqsum3}
  Let $\mu$ and $a$ be defined as in lemma~\ref{lem:pqmag0} and set
  $q=\frac{n+a}{n+2}$. Then
  \begin{gather*}
    \pqsum{n}{q}{q} \sim
    2\,\sqrt{n}\,\pqbin{n}{\frac{n}{2}\,\left(1+\mu\right)}{q}{q}
    \,\int\limits_{-\infty}^{+\infty}
    \exp\left\{
   2\,x^2\,\left(
   \frac{1}{\mu^2-1}+\frac{\atanh{\mu}}{\mu}
   \right)
   \right\}\,\dd x \sim \\
    \sqrt{2\,\pi\,n}\,\pqbin{n}{\frac{n}{2}\,\left(1+\mu\right)}{q}{q}\,
    \sqrt{
      \frac{
        \mu\,\left(1-\mu^2\right)
      }{
        \mu + \left(\mu^2-1\right)\,\atanh{\mu}
      }
    }
  \end{gather*}
  where the growth rate of the peak coefficient is that of
  lemma~\ref{lem:pqmid3}.
\end{theorem}

\section{Pochhammer bounds}\label{sec:pochhammer}
In this section we set up bounds for the $q$-Pochhammer function and
use them for giving bounds of quotients between $q$-binomial
coefficients.  Basically we mimic the upper bound in
\cite{kirousis:01} and \cite{kaporis:07} but we extend it to obtain a
lower bound as well.  They are very useful bounds so we will do this
in some detail though everything is based on standard elementary
methods.  First we need the integral estimate of a sum.  Let $f(x)$ be
a continuous, positive, decreasing function on the interval $m\le x\le
n+1$ where $m$ and $n$ are integers. Then
\begin{equation}
  \int\limits_{m}^{n+1}f(x)\,\dd x 
  \le \sum_{k=m}^nf(k) \le 
  f(m) + \int\limits_m^nf(x)\,\dd x
\end{equation}
Recall that the dilogarithm is defined as 
\begin{equation}
  \dilog{x} = 
  \sum_{n=1}^{\infty} \frac{x^n}{n^2} = 
  -\int\limits_0^x \frac{\log(1-t)}{t}\,\dd t
\end{equation}
Now let $0<a<1$ and $0<q<1$ and note that
\begin{equation}
  -\log\qpoch{a}{q}{n} = \sum_{k=0}^{n-1} -\log\left(1-a\,q^k\right)
\end{equation}
Note also that $-\log\left(1-a\,q^x\right)$ is a positive and
decreasing function for $x\ge 0$.  Take the series expansion
\begin{equation}
  -\log(1-x) = \sum_{k=1}^{\infty} \frac{x^k}{k}
\end{equation}
so that
\begin{equation}
  -\log(1-a\,q^x) = \sum_{k=1}^{\infty} \frac{\left(a\,q^x\right)^k}{k}
\end{equation}
Integration gives
\begin{gather}
  \int\limits_{0}^u -\log(1-a\,q^x) \,\dd x = 
  \sum_{k=1}^{\infty} \int\limits_0^u \frac{\left(a\,q^x\right)^k}{k}\,\dd x = 
  \sum_{k=1}^{\infty} \frac{a^k}{k}\, 
  \left[\frac{q^{k\,x}}{k\,\log q}\right]_0^u = \\
  \frac{1}{\log q}\,
  \left(
  \sum_{k=1}^{\infty} \frac{\left(a\,q^u\right)^k}{k^2} - 
  \sum_{k=1}^{\infty} \frac{a^k}{k^2}
  \right) = 
  \frac{\dilog{a\,q^u} - \dilog{a}}{\log q}
\end{gather}
Together with the integral estimates above we have
\begin{equation}
  \frac{\dilog{a\,q^n}-\dilog{a}}{\log q} \le
  -\log\qpoch{a}{q}{n} \le 
  -\log(1-a) + \frac{\dilog{a\,q^{n-1}}-\dilog{a}}{\log q}
\end{equation}
Reversing the signs and taking exponentials we finally obtain
\begin{gather}
  \qpoch{a}{q}{n} \ge
  (1-a)\,\exp\left(\frac{\dilog{a}-\dilog{a\,q^{n-1}}}{\log q}\right) \\
  \qpoch{a}{q}{n} \le
  \exp\left(\frac{\dilog{a}-\dilog{a\,q^n}}{\log q}\right)
\end{gather}

Now we turn to the $q$-binomial coefficients. Let $0\le\ell\le k\le
n/2$ and use \eqref{qbin2} to note that the ratio between coefficient
$k-\ell$ and coefficient $k$ is
\begin{equation}
  \qratio{n}{k}{\ell}{q} = \frac{\qbin{n}{k-\ell}{q}}{\qbin{n}{k}{q}} = 
  \frac{\qpoch{q^{k-\ell+1}}{q}{\ell}}{\qpoch{q^{n-k+1}}{q}{\ell}}
\end{equation}
Using the bounds for the $q$-Pochhammer function we can now bound the
ratio.  For the upper bound of the ratio we take the quotient of the
upper bound and the lower bound. The ratio $\qratio{n}{k}{\ell}{q}$
then has the upper bound
\begin{equation}\label{qubound}
  \qratio{n}{k}{\ell}{q} \le 
  \frac{
    \exp\left(
    \frac{
      \dilog{q^{k-\ell+1}}+\dilog{q^{n-k+\ell}}-\dilog{q^{k+1}}-\dilog{q^{n-k+1}}
    }{
      \log q
    }\right)
  }{
    1-q^{n-k+1}
  }
\end{equation}
and, quite analogously, it has the lower bound
\begin{equation}\label{qlbound}
  \textstyle
  \qratio{n}{k}{\ell}{q} \ge 
  \left(1-q^{k-\ell+1}\right)\,
  \exp\left(
  \frac{
    \dilog{q^{k-\ell+1}}+\dilog{q^{n-k+\ell+1}}-\dilog{q^{k}}-\dilog{q^{n-k+1}}
  }{
    \log q
  }\right)
\end{equation}
Using \eqref{pqvq} we can now obtain bounds for the quotients of
$p,q$-binomial coefficients. Note simply that
\begin{equation}\label{pqratio2}
  \pqratio{n}{k}{\ell}{p}{q} = 
  \frac{\qratio{n}{k}{\ell}{q/p}}{p^{\ell\,(n-2\,k+\ell)}}
\end{equation}
and use the bounds from \eqref{qubound} and \eqref{qlbound}.

\section{Controlling the $p,q$-binomial distribution}\label{sec:isocurves}
Choosing a $k$ for a given $n$ such that $\pqratio{n}{k}{1}{p}{q}=1$
defines a set of pairs $p,\,q$; an isocurve.  If we choose a value of
$q$, then what value should $p$ have to result in a distribution of
coefficients which has a peak at $k$ (and $n-k$), i.e. with
$\pqratio{n}{k}{1}{p}{q}=1$?  The iterative method \eqref{pqiter}
above produces the correct $p$ for any given $q$ but reveals no
information on $p$. To obtain this we need to parameterize $q$
properly and one way to do this is to set $q=1+z/n$ for some $z\le
0$. This parameterization was also used in \cite{kirousis:01} and
\cite{kaporis:07} for computing giving upper bounds on $q$-binomial
coefficients.

\subsection{Wide and flat distributions}
Let us begin with the particular distribution which has
its peak in the middle, i.e. $k=n/2$. Due to symmetry we are
interested only in the case $p>q$.  Recall that in the special case
$p=q$ the distribution has its peak at $n/2$ when
\begin{equation}
  q=\frac{n}{n+2} = 1 - \frac{2}{n} + \frac{4}{n^2} - \frac{8}{n^3} +\cdots
\end{equation}
which would correspond to $z=-2$. Let us also work under the
assumption that $p$ has the expansion
\begin{equation}
  p = 1 + \frac{y_1}{n} + \frac{y_2}{n^2} +  \cdots
\end{equation}
and determine what values of $y_1,y_2,\ldots$ we should have. First,
according to \eqref{pqratio1}, we should have
\begin{equation}
  \pqratio{n}{n/2}{1}{p}{q} = 
  \frac{p^{n/2} - q^{n/2}}{p^{n/2+1} - q^{n/2+1}} = 1
\end{equation}
which we rewrite as
\begin{equation}
  p^{n/2+1} - p^{n/2} = q^{n/2+1} - q^{n/2}
\end{equation}
Setting $p=1+y_1/n+y_2/n^2+\cdots$ and $q=1+z/n$ and performing a
series expansion using \eqref{expid}, we find the right hand side to
be
\begin{equation}
  \frac{z\,e^{z/2}}{n} - \frac{z^3\,e^{z/2}}{4\,n^2} +
  \frac{z^4\,(3\,z+16)\,e^{z/2}}{96\,n^3} + \cdots
\end{equation}
and the first two terms of the left hand side are
\begin{equation}
  \frac{y_1\,e^{y_1/2}}{n} - 
  \frac{\left(y_1^3-4\,y_2 - 2\,y_1\,y_2\right)\,e^{y_1/2}}{4\,n^2} + \cdots
\end{equation}
We solve this term by term. First
\begin{equation}
  y_1\,e^{y_1/2} = z\,e^{z/2}
\end{equation}
has the solution $y_1=2\,w$ where
\begin{equation}\label{wdef}
  w = \lambert{\frac{z}{2}\,e^{z/2}}
\end{equation}
Here $s=\lambert{x}$ is the Lambert function solving
$s\,e^s=x$. Note that for $z\ge -2$ we have $y_1=z$ but
this is not the case for $z<-2$ where we have $y_1>z$.
The second coefficient is
\begin{equation}\label{secondy}
  y_2 = \frac{w\,\left(4\,w^2-z^2\right)}{2\,(w+1)}
\end{equation}
with $w$ as before. We could go on and solve for $y_3,y_4,\ldots$ but,
for the case in hand we actually only need $y_1$. We will henceforth
drop the subscript and refer to it as simply $y$.

What shape does the distribution have at this particular $p$ and $q$?
What we are seeking is an expression for
$\pqratio{n}{n/2}{\ell}{p}{q}$ and this is where we start using
\eqref{qubound}, \eqref{qlbound} and \eqref{pqratio2}. We define
$\ell=x\,n^{3/4}$, just as we did in the case of $p=q$. First we need
an expression for $q/p$ and we let
\begin{equation}
  r = \frac{q}{p} = 
  \frac{1+\frac{z}{n}}{1+\frac{2\,w}{n}+\cdots} = 
  1 + \frac{z-2\,w}{n} + \cdots
\end{equation}
where $w$ is defined as in \eqref{wdef}. We are now ready to compute
the limits of the upper bound \eqref{qubound} and lower bound
\eqref{qlbound}.  It turns out that the limits of these bounds
coincide and we receive
\begin{equation}\label{pqratio3}
  \pqratio{n}{n/2}{x\,n^{3/4}}{p}{q} \sim
  \exp\left(\frac{w\,z\,(2\,w+z)\,x^4}{6}\right)
\end{equation}
Note that the special case $z=-2$ corresponds to $p=q$ and gives the
coefficient $-4/3$ of $x^4$.  The calculations were performed with
Mathematica and are much to unwieldy to fit in this paper. We have
prepared a Mathematica notebook that performs the calculations step by
step, using some practical transformation rules. What we compute is
actually the limit of the logarithm of the upper and lower bound. The
steps are as follows; compute series expansions of the different
powers of $r=q/p$, use them inside the dilogarithms and then compute
their series expansions, add the dilogarithms and the series
expansions of the logarithms of the remaining factors. Some
transformations of this expression helps Mathematica to take the limit
that gives the result.  

Figure~\ref{fig:pq1} demonstrates how the asymptotic ratio is achieved
with increasing $n$. It shows $\pqratio{n}{n/2}{x\,n^{3/4}}{p}{q}$ for
$z=-9$ at $y=2\,w=-0.10539\ldots$ and the asymptotic ratio is given by
\eqref{pqratio3}, that is, $e^{-0.719718\,x^4}$. The red curve is the
asymptote and the blue curves are for finite $n$ where the curves for
larger $n$ are closer to the asymptote.
\begin{figure}[!ht]
  \begin{center}
    \includegraphics[width=0.99\textwidth]{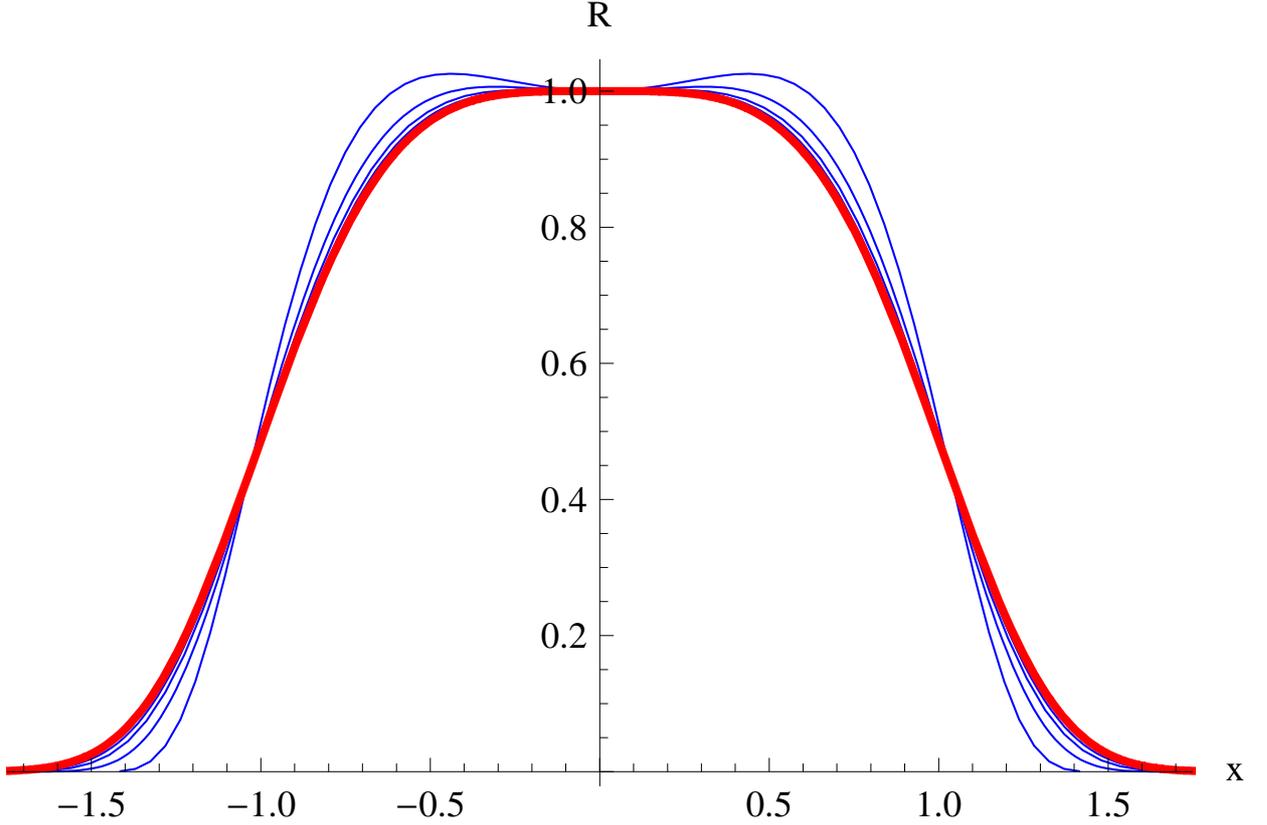}
  \end{center}
  \caption{$\pqratio{n}{n/2}{x\,n^{3/4}}{p}{q}$ (blue curves) versus
    $x$ at $z=-9$ and $y=2\,w=-0.105391$ for $n=2^6$, $2^8$,
    $2^{10}$ and $2^{12}$. The red curve is the asymptote
    $e^{-0.719718\,x^4}$.}
    \label{fig:pq1}
\end{figure}

\subsection{Wide and double-peaked distributions}
Changing $p$ only slightly, say on the order of $1/n^{3/2}$, allows us
to move around in the region where the distribution is wide.  The end
result is that with $q=1+z/n$ and $p=1+2\,w/n+a/n^{3/2}$ we get
\begin{equation}\label{pqratio4}
  \pqratio{n}{n/2}{x\,n^{3/4}}{p}{q} \sim
  \exp\left(
  \frac{w\,z\,(2\,w+z)\,x^4}{6}+
  \frac{a\,z\,(1+w)\,x^2}{2\,w-z}
  \right)
\end{equation}
The ratio expression gets nicer if we change the coefficient of $1/n^{3/2}$.
We suggest the following parametrization instead;
with $q=1+z/n$ let 
\begin{equation}\label{pexpr1}
  p = 1 + \frac{2\,w}{n} + 
  \frac{
    a\,w\,\left(z^2-4\,w^2\right)
  }{
    3\,(1+w)\,n^{3/2}
  }
\end{equation}
and we receive
\begin{equation}\label{pqratio5}
  \pqratio{n}{n/2}{x\,n^{3/4}}{p}{q} \sim
  \exp\left(
  \frac{w\,z\,(2\,w+z)}{6}\,\left(x^4 - 2\,a\,x^2\right)
  \right)
\end{equation}
This puts the maximum at $x=\pm \sqrt{a}$ for $a>0$ and at $x=0$ for
$a\le 0$.

Figure~\ref{fig:pq2} works like figure~\ref{fig:pq1} but here with the
parameter $a$ set to $1$. It shows
$\pqratio{n}{n/2}{x\,n^{3/4}}{p}{q}$ for $z=-9$ with
$p=1-0.105391/n-1.50171/n^{3/2}$ and the asymptotic ratio is given by
\eqref{pqratio5}, that is, $e^{-0.719718\,(x^4-2\,x^2)}$. The red
curve is the asymptote and the blue curves are for finite $n$ where
the larger $n$ are closer to the asymptote. Had we set $a<0$ the
distributions would still be wide but with a single peak in the middle.
\begin{figure}[!ht]
  \begin{center}
    \includegraphics[width=0.99\textwidth]{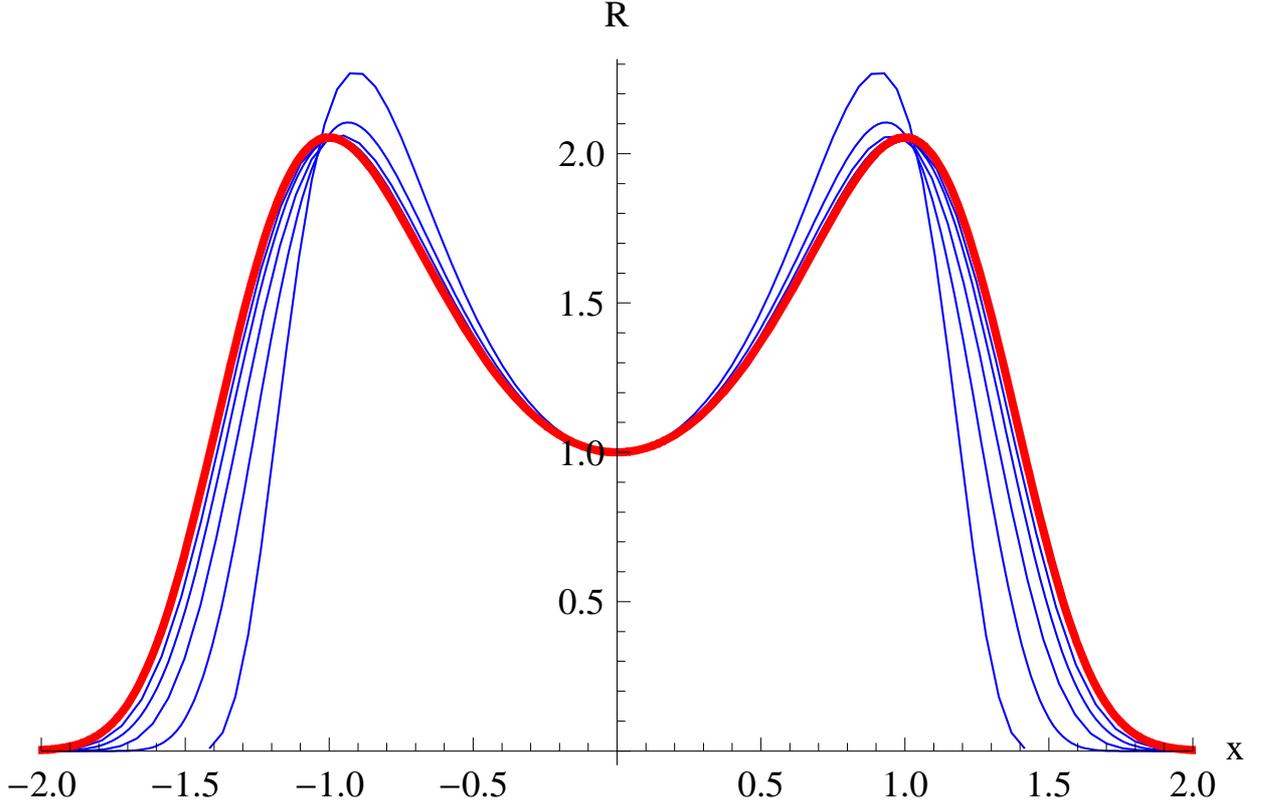}
  \end{center}
  \caption{$\pqratio{n}{n/2}{x\,n^{3/4}}{p}{q}$ (blue curves) versus
    $x$ at $z=-9$ and $a=1$ so that $q=1-9/n$ and
    $p=1-0.105391/n-1.50171/n^{3/2}$ for $n=2^6$, $2^8$, $2^{10}$,
    $2^{12}$ and $2^{14}$. The red curve is the asymptote
    $e^{-0.719718\,(x^4-2\,x^2)}$.}
    \label{fig:pq2}
\end{figure}

\subsection{Peakish distributions}
If we instead choose a middle ratio of $1+a/n$ then we receive a
sharply peaked distribution.  Say that we want
\begin{equation}\label{pqratio8}
  \pqratio{n}{n/2}{1}{p}{q} = 
  \frac{p^{n/2} - q^{n/2}}{p^{n/2+1} - q^{n/2+1}} = 1 + \frac{a}{n}
\end{equation}
for $p=1+y_1/n+\cdots$ and $q=1+z/n$ as before.  Note that $a=-2$
corresponds to a binomial distribution, so we are usually interested
in the case $-2<a<0$. After expanding the equation we receive as
before, from the first term of the left and right hand side, the
equation
\begin{equation}
  (a+y_1)\,e^{y_1/2} = (a+z)\,e^{z/2}
\end{equation}
which has the solution $y_1 = 2\,w-a$ where
\begin{equation}\label{wexpr8}
  w = \lambert{\frac{a+z}{2}\,e^{(a+z)/2}}
\end{equation}
The second equation, which is slightly longer, has the solution
\begin{equation}
  y_2 = \frac{
    (a-2\,w+z)\,\left(a^2\,(w+2)-2\,a\,w^2-w\,z\,(2\,w+z)\right)
  }{
    2\,(w+1)\,(a+z)
  }
\end{equation}
but we do not really need it at this moment. This gives a
distribution with a width of the order $\sqrt{n}$. Computing the limit
ratio gives us
\begin{equation}\label{pqratio6}
  \pqratio{n}{n/2}{x\,\sqrt{n}}{p}{q} \sim e^{a\,x^2}
\end{equation}
thus giving us an essentially gaussian distribution. Note that it does
not depend on $z$ in its current form. Of course, we expect other
terms to depend on $z$ but these vanish when $n\to\infty$.

\subsection{Two separate peaks}
Suppose that we want the peaks located outside the middle. Defining
$k=\frac{n}{2}\,(1+\mu)$ for $0<\abs{\mu}<1$ means that we move the
peaks out from the middle and that we move to another isocurve.  We
keep $q=1+z/n$ and $p=1+y_1/n+y_2/n^2+\cdots$ and solve
$\pqratio{n}{k}{1}{p}{q} = 1$, i.e.
\begin{equation}
  p^{n-k+1}-p^k = q^{n-k+1} - q^k
\end{equation}
With $k=\frac{n}{2}\,(1+\mu)$ we use \eqref{expid} on both sides and
find the equation
\begin{equation}
  \exp\left(\frac{y_1}{2}\,(1-\mu)\right) - 
  \exp\left(\frac{y_1}{2}\,(1+\mu)\right) =
  \exp\left(\frac{z}{2}\,(1-\mu)\right) - 
  \exp\left(\frac{z}{2}\,(1+\mu)\right)
\end{equation}
which we rewrite as 
\begin{equation}
  e^{y_1/2}\,\sinh{\frac{\mu\,y_1}{2}} = e^{z/2}\,\sinh{\frac{\mu\,z}{2}}
\end{equation}
At this point it would be appropriate to define the function
$s=\Lambert{x}{\mu}$, for $0<\mu<1$, as the maximum solution to the
equation
\begin{equation}
  x = e^s\,\sinh\left(s\,\mu\right)
\end{equation}
Note here that the function $e^x\,\sinh(\mu\,x)$ has a minimum at
$(-\atanh\mu)/\mu$ for $0<\mu<1$ and, due to symmetry, a maximum at
the same point for $-1<\mu<0$. The function $\Omega$ returns a
value in the interval
\begin{equation}
  -\frac{\atanh{\mu}}{\mu}<s<0
\end{equation}
The solution sought in our equation is thus $y_1=2\,w$ where
\begin{equation}\label{Wdef}
  w = \Lambert{
    e^{z/2}\,\sinh\frac{\mu\,z}{2}
  }{
    \mu
  }
\end{equation}
We should mention that $\Lambert{x}{\mu}$ is a natural extension of
$\lambert{x}$. In fact, if we let
$w=\lambert{x\,e^x}$ then
\begin{equation}
  \Lambert{e^x\,\sinh\mu\,x}{\mu} = 
  w + \frac{w\,(x^2-w^2)}{6\,(1+w)}\,\mu^2 +\cdots
\end{equation}
giving a good approximation for small values of $\mu$. The reader may
recall \eqref{secondy} above for comparison.

Having computed $p=1+y_1/n$ we compute the upper and lower bounds of
the ratio as before.  Unfortunately, the ratio has the rather ghastly
expression
\begin{multline}\label{pqratio7}
  \pqratio{n}{\frac{n}{2}\,(1+\mu)}{x\,\sqrt{n}}{p}{q} \sim
  \\ \exp\left\{\frac{x^2}{2}\,\left(2\,w+z-\frac{e^{\mu w+\frac{\mu
        z}{2}} \left(e^{2 w}-e^z\right) (z-2 w)}{-e^{2 \mu
      w+w+\frac{z}{2}}+e^{\mu w+2 w+\frac{\mu z}{2}}+e^{\mu
      w+\frac{\mu z}{2}+z}-e^{w+\mu z+\frac{z}{2}}}\right)\right\}
\end{multline}
where $w$ is defined by \eqref{Wdef}.

In figure~\ref{fig:pq3} we show how the finite cases approach their
asymptote for $\mu=1/3$ and $z=-9$. This gives 
$p=1-0.153208/n$ and the asymptotic ratio is given by
\eqref{pqratio7}, that is, $e^{-0.103551\,x^2}$. The red
curve is the asymptote and the blue curves are for finite $n$ where
the larger $n$ are closer to the asymptote.
\begin{figure}[!ht]
  \begin{center}
    \includegraphics[width=0.99\textwidth]{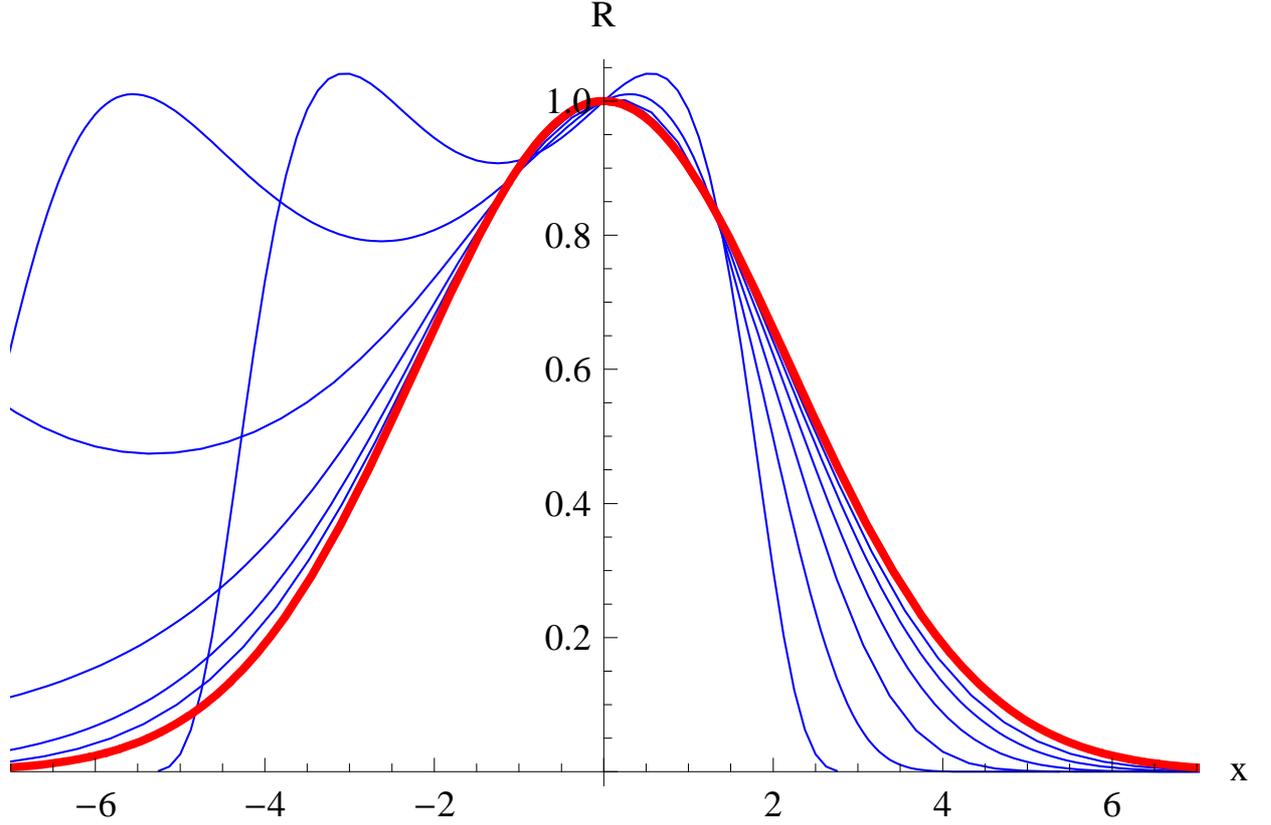}
  \end{center}
  \caption{$\pqratio{n}{n/2\,(1+\mu)}{x\,\sqrt{n}}{p}{q}$ (blue
    curves) versus $x$ at $\mu=1/3$ and $z=-9$ so that $q=1-9/n$ and
    $p=1-0.153208/n$ for $n=2^6$, $2^8$, $2^{10}$, $2^{12}$, $2^{14}$
    and $2^{16}$. The red curve is the asymptote
    $e^{-0.103551\,x^2}$.}
    \label{fig:pq3}
\end{figure}

\section{Moving along a diagonal}\label{sec:slope}
With $q=1+z/n$ and using the expressions given by \eqref{pexpr1} and
\eqref{pqratio5}, we can move around in the region where
$\pqratio{n}{n/2}{1}{p}{q}$ is very close to $1$. It is understood
here that a change in the parameter $a$ only changes $p$, i.e. we only
move in the $p$-direction and keep $q$ fixed. Let us instead say that
we want to move, ever so slightly, in both the $p$- and
$q$-direction. If the movement is on the order $1/n^{3/2}$ we can
translate it to a different parameter $a$ in the expression given by
\eqref{pexpr1}.

Given a $z$ in $q=1+z/n$ the starting point for $p$ is $p=1+y/n$ with
$y=2\,w$ and $w$ given by \eqref{wdef}. Let us think of the starting
coordinate as $(y_0,\,z_0)$. We wish to move from this point
$s/\sqrt{n}$ steps in the direction $t$, where the limit $t\to\infty$
corresponds to the case where we only move in the direction
$y$ (or $p$). The new point is now
\begin{equation}
  (y_1,\,z_1) = 
  \left(
  y_0+\frac{s\,t}{\sqrt{1+t^2}\,\sqrt{n}},\,
  z_0+\frac{s}{\sqrt{1+t^2}\,\sqrt{n}}
  \right)
\end{equation}
and note here that $\sqrt{(y_1-y_0)^2+(z_1-z_0)^2} =
\abs{s}/\sqrt{n}$. 

Having moved in the $z$-direction from $z_0$ to $z_1$ we now compute
the point $y_2$ so that $(y_2,\,z_1)$ stays on the isocurve with
$\pqratio{n}{n/2}{1}{p}{q}=1$. A small change in $z$ requires only a
small change in $y$ to preserve this property.

With $w_0=\lambert{z_0\,e^{z_0/2}/2}$ we have the series expansion
\begin{equation}
  \lambert{\frac{(z_0+x)\,\exp\left((z_0+x)/2\right)}{2}} = 
  w_0 + \frac{x\,w_0\,(z_0+2)}{2\,z_0\,(w_0+1)} + \cdots
\end{equation}
giving an approximation for small $x$. We want to set 
\begin{equation}
  x=\frac{s}{\sqrt{1+t^2}\,\sqrt{n}}
\end{equation}
and thus we have
\begin{equation}
  y_2 = y_0 + \frac{s\,w_0\,(z_0+2)}{z_0\,(w_0+1)\,\sqrt{1+t^2}\,\sqrt{n}}
\end{equation}
The difference $y_1 - y_2$ is what we want:
\begin{gather}
  y_1 - y_2 = \frac{s\,t}{\sqrt{1+t^2}\,\sqrt{n}} - 
  \frac{s\,w_0\,(z_0+2)}{z_0\,(w_0+1)\,\sqrt{1+t^2}\,\sqrt{n}} = \\
  \frac{s}{\sqrt{1+t^2}\,\sqrt{n}}\,
  \left(t - \frac{w_0\,(z_0+2)}{z_0\,(w_0+1)}\right)
\end{gather}
Now solve 
\begin{equation}
  y_1 - y_2 = \frac{a\,w_0\,(z_0^2-4\,w_0^2)}{3\,(w_0+1)\,\sqrt{n}}
\end{equation}
where the right hand side is extracted from \eqref{pexpr1}.
This gives 
\begin{equation}
  a = \frac{3\,s}{\sqrt{1+t^2}}\,
  \frac{t\,z_0\,(w_0+1)-w_0\,(z_0+2)}{z_0\,w_0\,(z_0^2-4\,w_0^2)}
\end{equation}
which then gives the ratio in \eqref{pqratio5}.

As an example, we take the case $p=q$ which corresponds to $t=1$ and
$z_0\to -2^-$. 
Taking the limit of the expression for $a$ gives
\begin{equation}
  \lim_{z\to -2^-} a = -\frac{3\,s}{8\,\sqrt{2}}
\end{equation}
which corresponds to the coefficient ratio
\begin{equation}
  \pqratio{n}{n/2}{x\,n^{3/4}}{p}{q} \sim 
  \exp\left(-\frac{4}{3}\,(x^4-2\,a\,x^2)\right) =
  \exp\left(-\frac{4}{3}\,x^4 - \frac{s}{\sqrt{2}}\,x^2\right)
\end{equation}
Compare this with lemma~\ref{lem:pqkvot1}; the parameter $a$ in the
lemma corresponds to taking $s=a\,\sqrt{2}$ steps in direction $t=1$
and this gives the same coefficient ratio as in the lemma.

\section{Moments}\label{sec:moments}
Once we have the ratios \eqref{pqratio3}, \eqref{pqratio5},
\eqref{pqratio6}, \eqref{pqratio7} it is an easy task to compute
moments of the distributions. Let us do this for the most interesting
case of \eqref{pqratio5}. First, to make the notation somewhat
simpler, denote $\phi = \fc{z} = -w\,z\,(2\,w+z)/6$, i.e. $\phi>0$,
with $w$ as in \eqref{wdef}, so that
\begin{equation}
  \pqratio{n}{n/2}{x\,n^{3/4}}{p}{q} \sim e^{-\phi\,(x^4-2\,a\,x^2)}
\end{equation}
where $a$ is defined by \eqref{pexpr1}. We use the notation
\begin{equation}
  \kmom{m} = \mean{\abs{k-\frac{n}{2}}^m}
\end{equation}
for the $m$th moment of the probability distribution of
$k=0,\ldots,n$.  Define also
\begin{equation}
  \xmom{m} = \int\limits_{-\infty}^{+\infty} \abs{x}^m \,
  \exp\left(-\phi\,(x^4-2\,a\,x^2)\right)\,\dd x,\quad m\ge 0
\end{equation}
so that the $m$th  moment becomes
$\mean{\abs{x}^m}=\xmom{m}/\xmom{0}$.  For $m=0$ we have
\begin{equation}
  \xmom{0} = \,\left\{\begin{array}{ll}
  \Big.\frac{\pi\,\sqrt{a}}{2}\,\exp\left(\phi\,a^2/2\right)\,
  \left(
  \besseli{1/4}{\phi\,a^2/2} + \besseli{-1/4}{\phi\,a^2/2}
  \right),& \textrm{for $a>0$} \\
  \Big.\frac{\sqrt{-a}}{\sqrt{2}}\,\exp\left(\phi\,a^2/2\right)\,
  \besselk{1/4}{\phi\,a^2/2},& \textrm{for $a<0$} \\
  \Big.\frac{\gammaf{1/4}}{2\,\phi^{1/4}},&\textrm{for $a=0$}
  \end{array}\right.
\end{equation}
so that 
\begin{equation}
  1 = \sum_{k=-n/2}^{n/2} \pqpr{n}{n/2+k}{p}{q} \sim
  n^{3/4}\,\pqpr{n}{n/2}{p}{q}\,\xmom{0}
\end{equation}
Next, for $m=1$
\begin{equation}
  \xmom{1} = \frac{\sqrt{\pi}}{2\,\sqrt{\phi}}\,
  \exp\left(\phi\,a^2\right)\,
  \left(1+\erf{a\,\sqrt{\phi}}\right)
\end{equation}
where $-1<\erf{x}<1$ is the error function.
and thus
\begin{equation}
  \kmom{1} = \mean{\abs{k-\frac{n}{2}}} \sim
  n^{3/4}\,\frac{\xmom{1}}{\xmom{0}}
\end{equation}
In general we have for $m\ge 0$ that 
\begin{equation}\label{kmomm}
  \kmom{m} = \mean{\abs{k-\frac{n}{2}}^m} \sim
  n^{3\,m/4}\,\frac{\xmom{m}}{\xmom{0}}
\end{equation}
where $\xmom{m}$ is given by
\begin{equation}\label{pqmom}
  \left\{\begin{array}{ll}
  \Big.\!\!\!\frac{\phi^{-(m+1)/4}}{2}\,\left(
  \gammaf{\frac{m+1}{4}}\,\hyper{\frac{m+1}{4}}{\frac{1}{2}}{\phi\,a^2} +
  2\,a\,\sqrt{\phi}\,\gammaf{\frac{m+3}{4}}\,
  \hyper{\frac{m+3}{4}}{\frac{3}{2}}{\phi\,a^2}\right),
  &a>0 \\
  \Big.\!\!\!\frac{\phi^{-(m+1)/4}}{2^{(m+1)/2}}\,\gammaf{\frac{m+1}{2}}\,
  \hyperu{\frac{m+1}{4}}{\frac{1}{2}}{\phi\,a^2},
  &a<0 \\
  \Big.\!\!\!\frac{\phi^{-(m+1)/4}}{2}\,\gammaf{\frac{m+1}{4}},
  &a=0 \\
  \end{array}\right.
\end{equation}
Here $\hyper{a}{b}{c}$ and $\hyperu{a}{b}{c}$ denote the confluent
hypergeometric functions of the first and second kind respectively.
If we want to compute cumulant ratios we first need moment ratios
which of course is easy now.  For example, in the case of $a=0$ we
have
\begin{align}
  \frac{\kmom{2}}{\kmom{1}^2} \sim &
  \frac{\xmom{0}\,\xmom{2}}{\xmom{1}^2} = \sqrt{2} = 1.4142\ldots \\
  \frac{\kmom{4}}{\kmom{2}^2} \sim &
  \frac{\xmom{0}\,\xmom{4}}{\xmom{2}^2} = \frac{\gammaf{1/4}^4}{8\,\pi^2}
  = 2.1884\ldots \\
  \frac{\kmom{6}}{\kmom{2}^3} \sim &
  \frac{\xmom{0}^2\,\xmom{6}}{\xmom{2}^3} = \frac{3\,\gammaf{1/4}^4}{8\,\pi^2}
  = 6.5653\ldots
\end{align}
These moment ratios are the same as obtained in the 5-dimensional
Ising model, see e.g. \cite{brezin:85}.

\section{Fine-tuning the exponents}\label{sec:exponents}
Note \eqref{pqratio5} and, as before, keep $\fc{z} =
-w\,z\,(2\,w+z)/6$ where $w$ is defined by \eqref{wdef}. Recall that
the first and second absolute moments obtained from \eqref{pqmom} for
$a=0$ are
\begin{align}\label{kmom12}
  \kmom{1} & \sim n^{3/4}\,\frac{\xmom{1}}{\xmom{0}} = 
  n^{3/4}\,\frac{\sqrt{\pi}}{\gammaf{1/4}}\,\frac{1}{\fc{z}^{1/4}} 
  \propto 
  \frac{
    n^{3/4}
  }{
    \fc{z}^{1/4}
  } \\
  \kmom{2} & \sim n^{3/2}\,\frac{\xmom{2}}{\xmom{0}} =
  n^{3/2}\,\frac{\pi\,\sqrt{2}}{\gammaf{1/4}^2}\,\frac{1}{\sqrt{\fc{z}}}
  \propto 
  \frac{
    n^{3/2}
  }{
    \sqrt{\fc{z}}
  }
\end{align}
The argument $z$ is allowed to depend on $n$ but probably not to a
high order. At this point it is not clear how $z$ may depend on $n$
for the calculations leading to \eqref{pqratio5} and $\fc{z}$ to
work. We will assume, for the moment (see the end of this section),
that the expressions for the moments of \eqref{pqmom} are valid when
$z=\bigo{\log n}$.

The series expansion of $\lambert{x}$ is
\begin{equation}
  w=\lambert{x} = x - x^2 + \frac{3\,x^3}{2} - \frac{8\,x^4}{3} + \cdots
\end{equation}
With $x=e^{z/2}\,z/2$, and note that $z$ is negative, we have
\begin{equation}
  w = \lambert{x} = \lambert{\frac{z\,e^{z/2}}{2}} =
  \frac{z\,e^{z/2}}{2} - \frac{z^2\,e^z}{4} +
  \frac{3\,z^3\,e^{3\,z/2}}{16} + \cdots
\end{equation}
so that 
\begin{equation}\label{fcseries}
  \fc{z} = \frac{-z\,w\,(2\,w+z)}{6} = 
  -\frac{z^3\,e^{z/2}}{12} - \frac{z^3\,e^z}{12} + \frac{z^4\,e^z}{24} + \cdots
\end{equation}
Set $z=\lambda_0 + \lambda_1\,\log n+\lambda_2\,\log\log n +
\lambda_3\,\log\log\log n$ with $\lambda_1, \lambda_2\le 0$ and focus on
the first term of \eqref{fcseries}.
\begin{multline}\label{fczlogn}
  \fc{z} \sim 
  \frac{-1}{12}\,\left( \lambda_0 +
  \lambda_1\,\log n 
  +\lambda_2\,\log\log n + \lambda_3\,\log\log\log
  n\right)^3\\ 
  e^{\lambda_0/2}\, n^{\lambda_1/2}\,
  \log^{\lambda_2/2}n\, \left(\log\log n\right)^{\lambda_3/2}
\end{multline}
We are interested in two special cases. First choose $\lambda_1<0$,
$\lambda_2=-6$ and $\lambda_3=0$. This gives
\begin{equation}
  \fc{z} \sim \frac{(-\lambda_1)^3}{12}\,e^{\lambda_0/2}\,n^{\lambda_1/2}
\end{equation}
Combining this with \eqref{kmom12} we receive
\begin{equation}\label{logmom1}
  \kmom{1} \sim 
  \frac{3^{1/4}\,\sqrt{2\,\pi}}{\gammaf{1/4}}\,
  \frac{n^{3/4-\lambda_1/8}}{(-\lambda_1)^{3/4}\,e^{\lambda_0/8}}
\end{equation}
and
\begin{equation}\label{logmom2}
  \kmom{2} \sim 
  \frac{2\,\sqrt{6}\,\pi}{\gammaf{1/4}^2}\,
  \frac{n^{3/2-\lambda_1/4}}{(-\lambda_1)^{3/2}\,e^{\lambda_0/4}}
\end{equation}
Had we let $\lambda_2=0$, instead of $\lambda_2=-6$, then we would
have ended up with a factor $\log^{3/4}n$ in the denominator of
\eqref{logmom1} and a factor $\log^{3/2}n$ in the denominator of
\eqref{logmom2}. For the second case we choose $\lambda_1=0$,
$\lambda_2<0$ and $\lambda_3=-6$. We get
\begin{equation}
  \fc{z} \sim \frac{(-\lambda_2)^3}{12}\,e^{\lambda_0/2}\,\log^{\lambda_2/2}n
\end{equation}
This together with \eqref{kmom12} gives us
\begin{equation}\label{loglogmom1}
  \kmom{1} \sim 
  \frac{3^{1/4}\,\sqrt{2\,\pi}}{\gammaf{1/4}}\,
  \frac{n^{3/4}\,\log^{-\lambda_2/8} n}{(-\lambda_2)^{3/4}\,e^{\lambda_0/8}}
\end{equation}
and
\begin{equation}\label{loglogmom2}
  \kmom{2} \sim 
  \frac{2\,\sqrt{6}\,\pi}{\gammaf{1/4}^2}\,
  \frac{n^{3/2}\,\log^{-\lambda_2/4}n}{(-\lambda_2)^{3/2}\,e^{\lambda_0/4}}
\end{equation}
These expressions obviously converge extremely slowly and are probably
not of any use for $n$ that might occur in practical situations.

We have managed to verify \eqref{logmom1} and \eqref{logmom2} by using
the method described in section~\ref{sec:isocurves}, again using
Mathematica, only in the special case $z=-\log n$, i.e. $\lambda_1=-1$,
$\lambda_0=\lambda_2=\lambda_3=0$. We could then confirm that
\begin{equation}\label{pqratio9}
  \pqratio{n}{n/2}{\frac{x\,n^{7/8}}{\log^{3/4}n}}{p}{q} \sim
  \exp\left(\frac{-x^4}{12}\right)
\end{equation}
Computing the moments of this distribution produces the same result as
setting $\lambda_1=-1$ and $\lambda_0=\lambda_2=\lambda_3=0$ in
\eqref{fczlogn} and then computing the moments in the same the way we
obtained \eqref{logmom1} and \eqref{logmom2}.  A more general
computation seems not to be within reach with our current set of tools
though. To conclude this section we note that the exponent of $n$ in
\eqref{pqratio9} is $3/4-\lambda_1/8$. For this exponent to stay less
than one we thus need $\lambda_1>-2$, giving us a bound on $z$.

\section{The Ising model}\label{sec:ising}
A state $\tau$ on a graph $G$ is a function from the set of vertices
to $\{\pm 1\}$. There are thus $2^n$ states for a graph on $n$
vertices.  We define the energy of a state as $E(\tau) =
\sum_{ij}\tau_i\tau_j$ where the sum is taken over all edges $ij$ of
$G$. The magnetisation is defined as $M(\tau) = \sum_i\tau_i$ with the
sum taken over all vertices $i$ of $G$. Note that $-n\le M\le n$ and
it only takes every alternate value,
i.e. $M\in\{-n,-n+2,-n+4,\ldots,n-4,n-2,n\}$. We will often need to
refer to it in terms of how many negative spins the state has. If $k$
spins are negative then $M=n-2\,k$.

The partition function of the Ising model is defined for any graph $G$
as
\begin{equation}
  \fz(G;\, x,\, y) = \sum_{\tau}x^{E(\tau)}\,y^{M(\tau)} =
  \sum_{E, M} a(E,M)\,x^E\,y^M
\end{equation}
The coefficients $a(E,M)$ then are defined as the number of states
with energy $E$ and magnetisation $M$. Denote the number of
states at energy $E$ by $a(E)=\sum_M a(E,M)$. Note that the number of
states at magnetisation $M$ is just $\binom{n}{k}$, where $k=(n-M)/2$
is the number of negative spins. Let also $\fz_k$ denote the terms of
$\fz$ with magnetisation $M=n-2\,k$ for a graph on $n$ vertices, so
that $\fz=\fz_0+\fz_2+\cdots+\fz_n$, i.e. $\fz_k$ are the terms
corresponding to $k$ negative spins.

If we evaluate the partition function in $x=e^K$ and $y=e^H$ with $K$
the dimensionless coupling, or inverse temperature $J/k_BT$, and
$H=h/k_BT$ as the dimensionless external magnetic field, we obtain the
physical partition function denoted $\fZ=\fZ(G;\,K,\,H) =
\fz(G;\,e^K,\,e^H)$, though we are usually interested only in the case
when $H=0$ (or $y=1$).  Analogously, we write
$\fZ=\fZ_0+\cdots+\fZ_n$. The dimensionless and normalised free energy
is defined as $\fF=\left(\log\fZ\right)/n$.  From the derivatives of
the free energy we can now obtain other physical quantities such as
the internal energy $\partial \fF/\partial K$ and the specific heat
$\partial^2 \fF/\partial K^2$ though we shall not be needing the
latter for this investigation.

We assume the Boltzmann distribution on the states so that (with
$H=0$) the probability for state $\tau$ is
\begin{equation}
  \pr{\tau} = \frac{e^{K\,E(\tau)}}{\fZ}
\end{equation}
We have then especially that the probability for energy $E$ is
\begin{equation}
  \pr{E} = \frac{a(E)\,e^{K\,E}}{\fZ}
\end{equation}
and the probability for magnetisation $M$ is
\begin{equation}
  \pr{M} = \frac{1}{\fZ}\,\sum_E a(E,M)\,e^{K\,E} = \frac{\fZ_k}{\fZ}
\end{equation}
where $M=n-2\,k$.

Denote by $K^*$ the coupling where
$\fZ_{n/2-1}=\fZ_{n/2}=\fZ_{n/2+1}$. This coupling will correspond to
the relation $\pqratio{n}{n/2}{1}{p}{q}=1$ for some choice of $p,q$.
We define the (spontaneous) normalised magnetisation
$\amu=\mean{\abs{M}}/n$ and the (spontaneous) susceptibility
$\achi=\var{\abs{M}}/n=\left(\mean{M^2}-\mean{\abs{M}}^2\right)/n$.
The pure susceptibility is simply
$\chi=\var{M}/n=\mean{M^2}=4\,\kmom{2}/n$.  Since $M=n-2\,k$ we thus
have $\amu=2\,\kmom{1}/n$ and $\achi=4\,(\kmom{2}-\kmom{1}^2)/n$.
Recall the traditional finite-size scaling laws which claim that in
the critical region, i.e. near $K_c$, $\amu\propto L^{-\beta/\nu}$ and
$\achi\propto L^{\gamma/\nu}$. Being near $K_c$ means that
$\abs{K-K_c}\propto L^{-1/\nu}$ and we especially expect $K^*$ to
belong to this region. Though the high- and low-temperature exponents
may or may not be equal for three dimensions, see \cite{cubeart} for
an in-depth numerical investigation of this matter, the details of
these exponents are not important for our present investigation. What
matters is that there are exponents that guide the growth of e.g. the
susceptibility near $K_c$.

\subsection{The complete graph}\label{sec:complete}
For a complete graph, denoted $K_n$ on $n$ vertices and $\binom{n}{2}$
edges the partition function is easy to compute. Suppose $k$ of the
vertices are assigned spin $-1$ and the other $n-k$ have spin $+1$.
The magnetisation is obviously $M=n-2\,k$ and the energy is
\begin{equation}
  E = 
  \binom{k}{2} + \binom{n-k}{2} - k\,(n-k) = 
  \binom{n}{2} - 2\,k\,(n-k)
\end{equation}
The partition function is then
\begin{equation}
  \fz(K_n;\, x,\,y) = 
  x^{\binom{n}{2}}\,y^{n}\,
  \sum_{k=0}^n \binom{n}{k}\,
  \left(\frac{1}{x^2}\right)^{k\,(n-k)}\,
  \left(\frac{1}{y^2}\right)^{k}
\end{equation}
and with $y=1$ we have
\begin{equation}
  \fz(K_n;\, x,\,1) = 
  x^{\binom{n}{2}}\,
  \sum_{k=0}^n \pqbin{n}{k}{q}{q} =
  x^{\binom{n}{2}}\,\pqsum{n}{q}{q}
\end{equation}
where $q=1/x^2$. 
Thus we have
\begin{equation}\label{zkn}
  \fZ(K_n;\,K,\,0) = 
  \exp\left\{K\,\binom{n}{2}\right\}\,\pqsum{n}{q}{q}
\end{equation}
where $q=\exp(-2\,K)$.
Obviously we have
\begin{equation}
  \pr{M=n-2\,k} = \pqpr{n}{k}{q}{q}
\end{equation}
Since we have defined the critical temperature as the $K=K^*$ where
the middle ratio is $1$, i.e. $\pr{M=-2}=\pr{M=0}=\pr{M=+2}$ then this
corresponds to the point where $\pqratio{n}{n/2}{1}{q}{q}=1$, which
takes place at $q=n/(n+2)$ as we saw in lemma~\ref{lem:pqstar}.  Thus
$K^*=\frac{1}{2}\,\log\left(1+\frac{2}{n}\right)$ for $K_n$.

In short, the partition function and the magnetisation distribution
for $K_n$ can be expressed in terms of $p,q$-binomial coefficients.
Does this hold for all graphs? No. In fact, it seems to only be true
for $K_n$.  However, it does seem to hold \emph{asymptotically} as the
order of the graphs increase, for some interesting classes of graphs.
The precise formulation of such a statement remains and falls outside
this paper.

\subsection{The average graph}
Let us compute the sum of all partition functions taken over all
graphs on $n$ vertices.
\begin{gather}
  \bar{\fz}_n(x,y) = \sum_{G\subseteq K_n} \fz(G;\,x,\,y) = \\
  \sum_{i=0}^n \binom{n}{i}\,y^{n-2\,i}\,
  \sum_{j=0}^{\binom{i}{2}} \binom{\binom{i}{2}}{j}\,x^j\,
  \sum_{k=0}^{\binom{n-i}{2}} \binom{\binom{n-i}{2}}{k}\,x^k\,
  \sum_{\ell=0}^{i\,(n-i)} \binom{i\,(n-i)}{\ell}\,x^{-\ell} = \\
  \sum_{i=0}^n \binom{n}{i}\,y^{n-2\,i}\,
  \left(1+x\right)^{\binom{i}{2}}\,
  \left(1+x\right)^{\binom{n-i}{2}}\,
  \left(1+\frac{1}{x}\right)^{\binom{i\,(n-i)}{2}} = \\
  \left(1+x\right)^{\binom{n}{2}}\,y^n\,\sum_{i=0}^n 
  \binom{n}{i}\,\left(\frac{1}{y^2}\right)^{i}\,
  \left(\frac{1}{x}\right)^{i\,(n-i)}
\end{gather}
and for $y=1$ we have
\begin{gather}
  \bar{\fz}_n(x,1) = 
  \left(1+x\right)^{\binom{n}{2}}\,\sum_{i=0}^n 
  \binom{n}{i}\,\left(\frac{1}{x}\right)^{i\,(n-i)} = \\
  \left(1+x\right)^{\binom{n}{2}}\,\sum_{i=0}^n \pqbin{n}{i}{q}{q} =
  \left(1+x\right)^{\binom{n}{2}}\,\pqsum{n}{q}{q}
\end{gather}
where $q=1/x$ so that $K^*=\log(1+2/n)$. Again we have $\pr{M=n-2\,k}
= \pqpr{n}{k}{q}{q}$. The mean magnetisation distribution can then be
modelled by $p,q$-binomial coefficients, though with $p=q$, just as
for the complete graph.

\subsection{The complete bipartite graph}
What about $K_{u,v}$, i.e. the complete bipartite graph on
$n=u+v$ vertices? Now the partition function is
\begin{gather}
  \fz(K_{u,v};\,x,\,y) = \\ \sum_{i=0}^{u} \sum_{j=0}^{v}
  \binom{u}{i}\,\binom{v}{j}\,y^{u+v-2\,i-2\,j}\,
  x^{i\,j+(u-i)\,(v-j)-i\,(v-j)-j\,(u-i)} =
  \\ x^{u\,v}\,y^{u+v}\,\sum_{i=0}^{u}
  \sum_{j=0}^{v}\binom{u}{i}\,\binom{v}{j}\,
  \left(\frac{1}{y^2}\right)^{i+j}\,
  \left(\frac{1}{x^2}\right)^{i\,(v-j)+j\,(u-i)}
\end{gather}
which for $y=1$ gives us
\begin{gather}
  \fz(K_{u,v};\,x,\,1) = \\
  x^{u\,v}\,\sum_{i=0}^{u} 
  \sum_{j=0}^{v}\binom{u}{i}\,\binom{v}{j}\,
  \left(\frac{1}{x^2}\right)^{i\,(v-j)+j\,(u-i)} = \\
  \sum_{k=0}^{u+v}\sum_{\ell} \binom{u}{\ell}\,\binom{v}{k-\ell}\,
  x^{(u-2\,\ell)\,(v-2\,(k-\ell))}
\end{gather}
which defines the partial sums for $y=1$ as
\begin{equation}
  \fz_k = 
  \sum_{\ell} \binom{u}{\ell}\,\binom{v}{k-\ell}\,
  x^{(u-2\,\ell)\,(v-2\,(k-\ell))}
\end{equation}
Data suggests that
\begin{equation}
  \pr{M=n-2\,k} = \frac{\fZ_k}{\fZ} \approx \pqpr{n}{k}{p}{q}
\end{equation}
given an appropriate choice of $p$ and $q$ and for a rather wide range
of temperatures. It does however not seem to hold if $u$ differ from
$v$.  In the left panel of figure~\ref{fig:knn1} we show a sample of
magnetisation distributions together with fitted $p,q$-distributions
for a $K_{32,32}$. To find the appropriate $p$ and $q$ we used the
method described in the $p,q$-find algorithm in
section~\ref{sec:pqdist}.  The fit is excellent. The right panel of
figure~\ref{fig:knn1} shows $y=n\,(p-1)$ versus $z=n\,(q-1)$ for a
range of temperatures and for complete bipartite graphs of different
sizes. High temperatures, $K=0$ gives $p=q=1$, i.e. $y=z=0$ in the
upper right corner. As the temperature decreases, i.e. with increasing
$K$, we move along the curves.  The points are where the distribution
is exactly flat, i.e. the inverse temperature $K^*$ where the middle
probabilities are equal.  We have no exact closed form expression for
$K^*$ but one can show that the series expansion of this inverse
temperature for $K_{n/2,n/2}$, i.e. a total of $n$ vertices is
\begin{equation}
  K^* = \frac{2}{n} - \frac{3}{n^2} + \frac{11}{3\,n^3} -
  \frac{101}{24\,n^4} + \frac{3827}{480\,n^5} + \cdots
\end{equation}
The calculations behind this are rather long and were done with
Mathematica.  Compare this with the expansion in the previous
subsection for the complete graph on $n$ vertices, $K_n$, which begins
$K^*=1/n-1/n^2+4/3 n^3+\cdots$.  In the lower left corner the
distribution has $\pr{M=n}=\pr{M=n-2}$, i.e. at $K=\log{n}/n$.
\begin{figure}
  \begin{minipage}{0.5\textwidth}
    \begin{center}
      \includegraphics[width=0.99\textwidth]{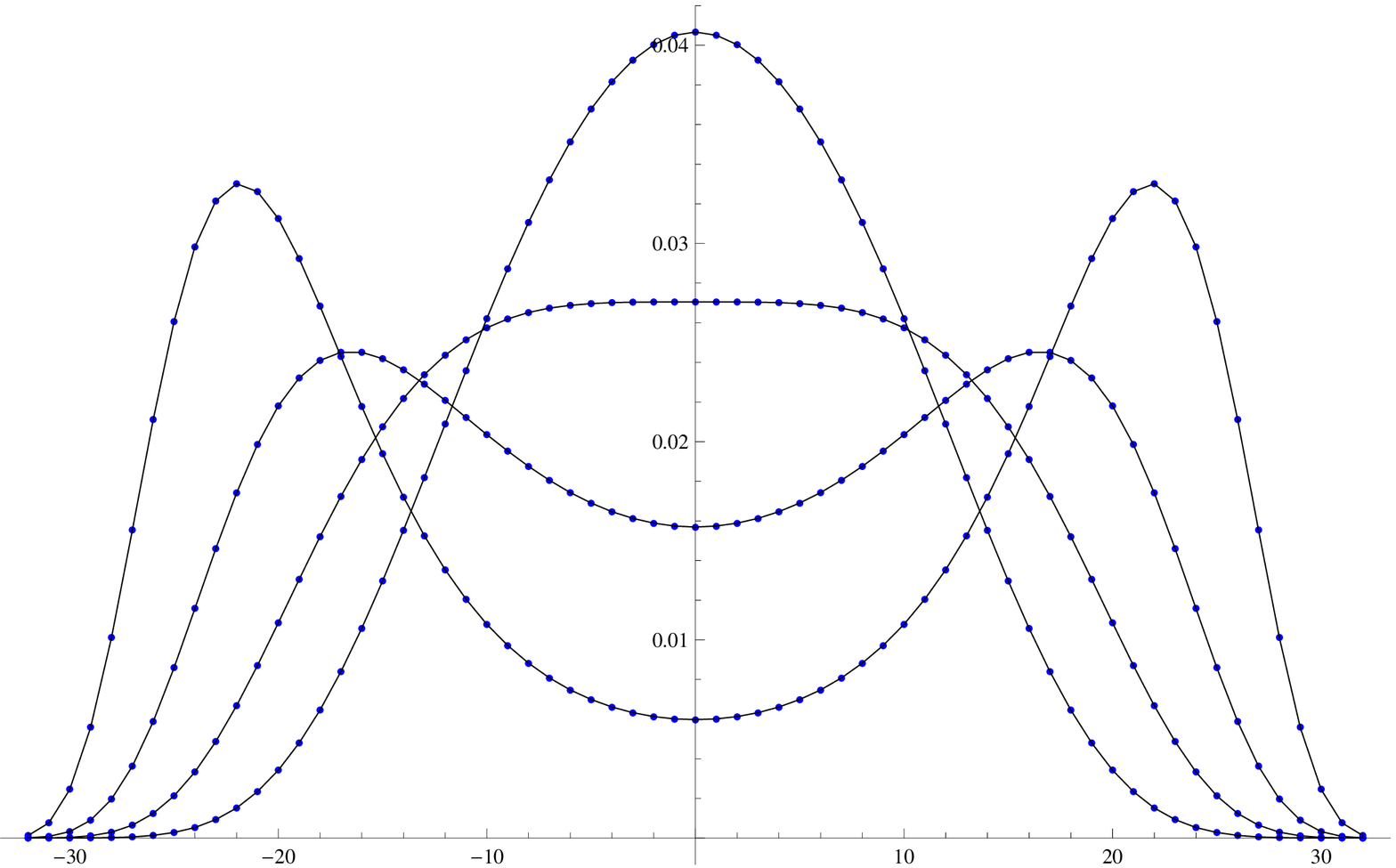}
    \end{center}
  \end{minipage}%
  \begin{minipage}{0.5\textwidth}
    \begin{center}
      \includegraphics[width=0.99\textwidth]{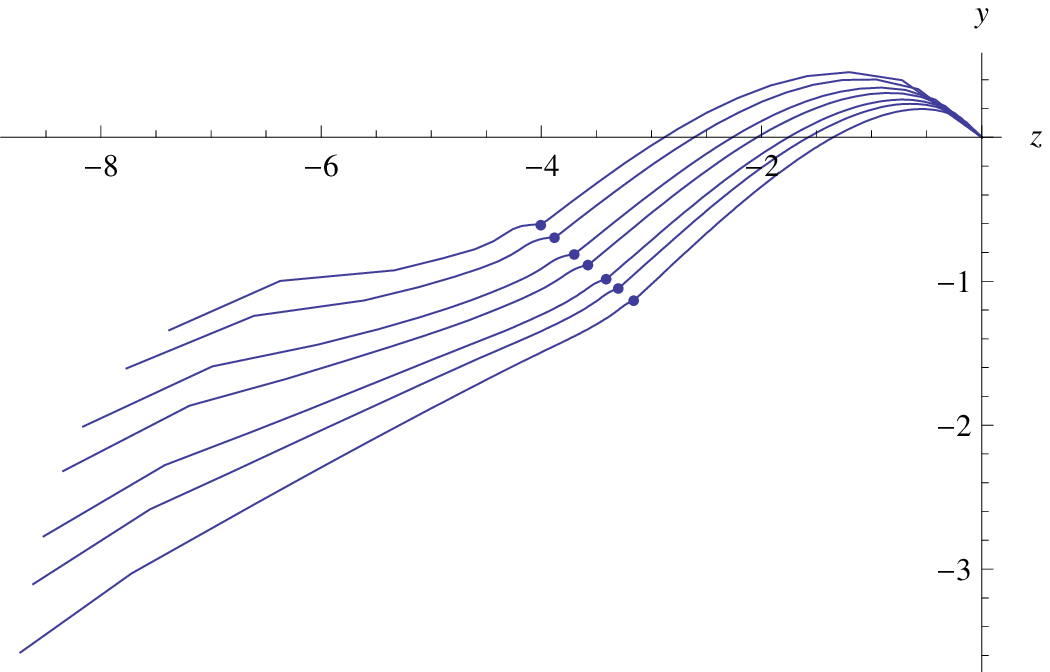}
    \end{center}
  \end{minipage}
  \caption{Left: magnetisation distributions for $K_{32,32}$ (points)
    and fitted $\pqpr{64}{k}{p}{q}$ (lines) vs $k-32$ for four
    different temperatures. Right: $y=n\,(p-1)$ versus $z=n\,(q-1)$
    for $K_{n/2,n/2}$, with $n=24, 32, 48, 64, 96, 128, 192$
    (downwards). Points represent $K^*$. Low temperatures in lower
    left corner.}
    \label{fig:knn1}
\end{figure}
The points in the right panel of figure~\ref{fig:knn1} should approach
$z=y=-2$. One can show that the distribution of a balanced bipartite
graph $K_{n/2,n/2}$ has the shape
\begin{equation}
  \ratio{n}{n/2}{x\,n^{3/4}} \sim
  \exp\left(-\frac{4}{3}\,x^4\right)
\end{equation}
at $K^*$, just like the complete graph on $n$ vertices. Since
$\fc{-2}=4/3$ we then assume that $z$, and thus also $y$, will approach
$-2$.

\subsection{The free energy}
If the magnetisations were indeed an exact $p,q$-binomial distribution
then we could also express the free energy as
\begin{equation}\label{free}
  \fF(G; K) = \frac{K\,m}{n} + \frac{\log\pqsum{n}{p}{q}}{n}
\end{equation}
for a graph on $n$ vertices and $m$ edges and it is here implied that
$p$ and $q$ depend on $K$. Why this expression?  Note that
$\fZ_0=a(m,n)\,e^{K\,m}=e^{K\,m}$ and thus we have
\begin{equation}
  \pqpr{n}{0}{p}{q} = \frac{\pqbin{n}{0}{p}{q}}{\pqsum{n}{p}{q}} =
  \frac{1}{\pqsum{n}{p}{q}} = \frac{\fZ_0}{\fZ} = \frac{e^{K\,m}}{\fZ}
\end{equation}
from which the result follows. Compare with \eqref{zkn} where this
relation holds exactly. Actually we expect \eqref{free} to be a good
approximation for $K$ near 0, where the distributions are close to
binomial, and for very high $K$ where all the probability mass is
concentrated on the extreme magnetisations. In the left plot of
figure~\ref{fig:free} we show the exactly computed free energy (red
curve) for a complete bipartite graph on $16+16$ vertices together
with the $p,q$-approximation \eqref{free} (points).  The fit is indeed
very good for the whole temperature range. Taking a derivative of the
points with respect to $K$ produces a good approximation to the
internal energy as the right plot shows.
\begin{figure}[!ht]
  \begin{minipage}{0.5\textwidth}
    \begin{center}
      \includegraphics[width=0.99\textwidth]{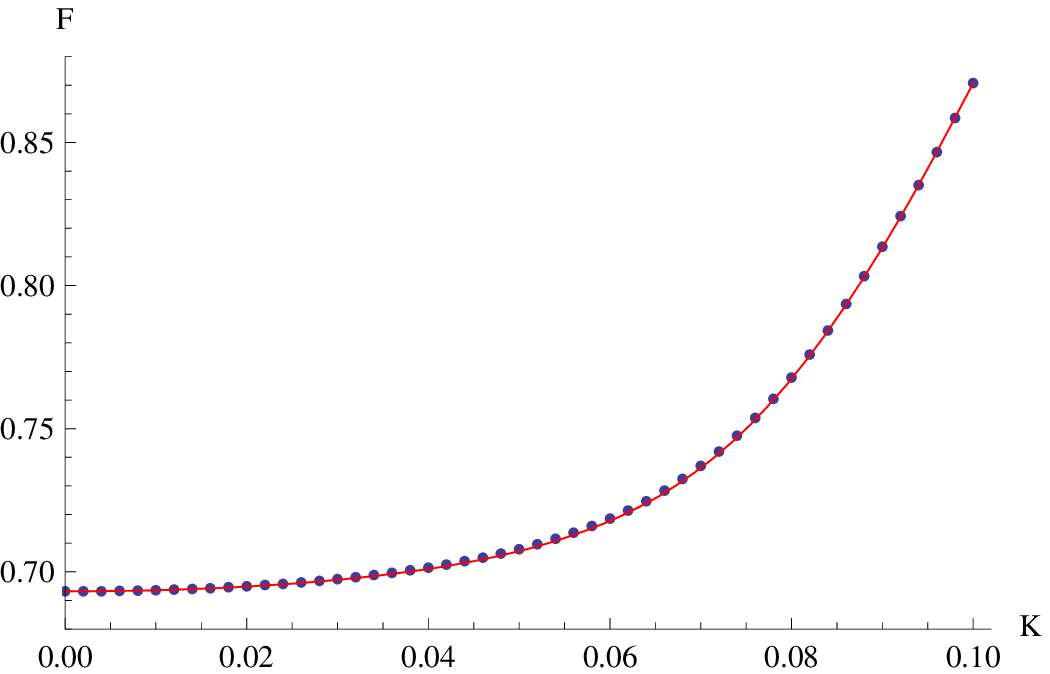}
    \end{center}
  \end{minipage}%
  \begin{minipage}{0.5\textwidth}
    \begin{center}
      \includegraphics[width=0.99\textwidth]{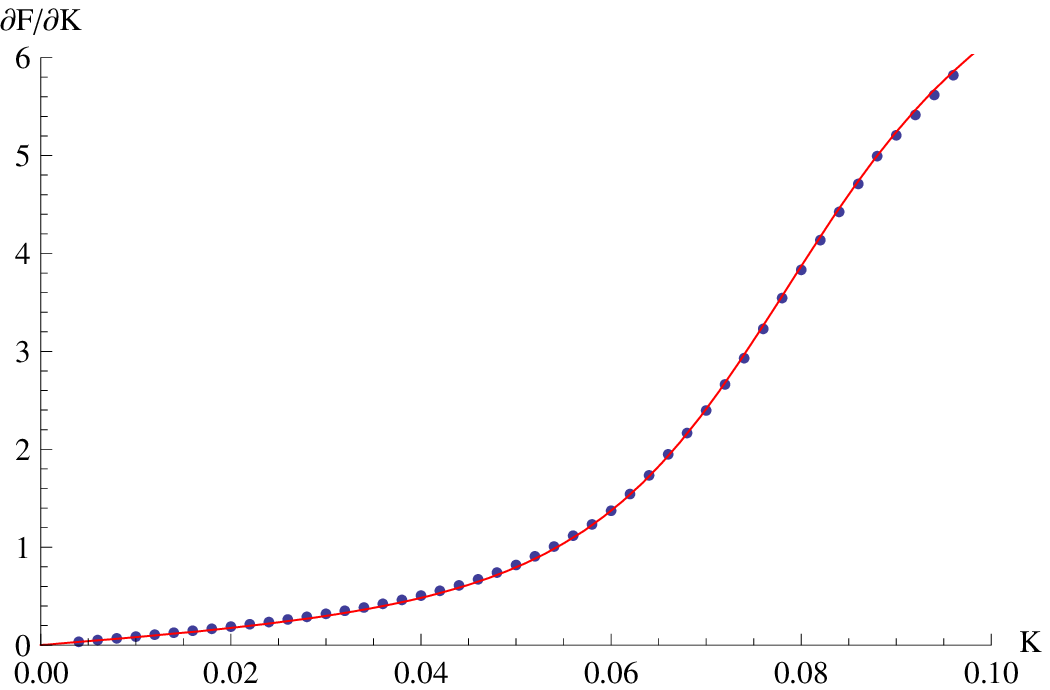}
    \end{center}
  \end{minipage}
  \caption{Left: free energy (red curve) compared to the formula
    \eqref{free} (points) for $K_{16,16}$. Right: internal energy 
    (red curve) compared to the derivative of the points produced 
    by \eqref{free} (points) for $K_{16,16}$.}
    \label{fig:free}
\end{figure}

\section{Lattices}\label{sec:lattices}
The plots in figure~\ref{fig:knn1} are very representative for several
graphs of interest. We intend to focus on the graphs that
traditionally are studied in statistical physics; lattice graphs. We
will take a look at the simple lattices in 1, 2, 3, 4 and 5
dimensions.  Just to be clear, a 1D-lattice is a cycle $C_n$ on $n$
vertices and it is $2$-regular, i.e. $2$ neighbours for each
vertex. The 2-dimensional $L\times L$-lattice is the cartesian product
of two cycles on $L$ vertices. The product thus has $n=L^2$ vertices
and it is $4$-regular. The $d$-dimensional $L\times L\times
\cdots\times L$-lattice is a product of $d$ cycles on $L$ vertices,
thus having a total of $n=L^d$ vertices. It is obviously
$2\,d$-regular.  Assuming finite-size scaling to hold then for a
$d$-dimensional lattice we have
\begin{equation}
  \kmom{1}=n\,\amu/2 \propto 
  n\,L^{-\beta/\nu} = n\,\left(n^{1/d}\right)^{-\beta/\nu} =
  n^{1-\beta/d\,\nu}
\end{equation}
and correspondingly for the second moment
\begin{equation}
  \kmom{2}=n\,\chi/4 \propto n\,L^{\gamma/\nu} =
  n\,\left(n^{1/d}\right)^{\gamma/\nu} = n^{1+\gamma/d\,\nu}
\end{equation}
Note also that for an $r$-regular triangle-free graphs we have
$\pr{M=n}=\pr{M=n-2}$ when $K=\frac{\log n}{2\,r}$.

\subsection{1D-lattices}
For 1D-lattices we can compute the coefficients $a(E,M)$ exactly.  It
is an exercise to show that the number of states with $k$ negative
spins and $\ell$ negative spin products (over the edges) is
\begin{equation}
  a(E,M) = \binom{k}{k - \ell/2}\,\binom{n-k-1}{n-k - \ell/2} + 
  \binom{k-1}{k - \ell/2}\, \binom{n-k}{n-k - \ell/2}
\end{equation}
where $M=n-2\,k$ and $E=n-2\,\ell$.  The distribution of
magnetisations do not behave in a way representative for lattices of
higher dimension.  However, for extremely low temperatures the
probabilities $\pr{M=-n}=\pr{M=n}$ will dominate the other
probabilities. The two outermost probabilities, $\pr{M=n}$ and
$\pr{M=n-2}$, are equal when $K=\frac{\log n}{4}$. For the 1D-lattice
the distribution is here sharply unimodal, while for higher dimensions
the distribution is bimodal and has its peaks at the extreme
magnetisations.  For $K$ larger than $(\log n)/4$ the distribution
actually has three peaks, i.e. a local maximum at $M=0$. For
1D-lattices the $p,q$-approximation of the distribution thus breaks
down beyond this $K$ since it can not model a local maximum in the
middle as well as peaks at the ends; they are at most
\emph{bi}modal. For $K$ less than this point the $p,q$-distribution is
a very good approximation. Figure~\ref{fig:1d} demonstrates this
clearly; for the flattest distribution (low temperature) the fitted
$p,q$-distribution starts to deviate from the actual distribution. In
figure~\ref{fig:1d0} we plot $y=n\,(p-1)$ and $z=n\,(q-1)$ versus $K$
for a range of different $n$. Clearly there is some limit curve here,
though we have not established what the limit function is.
\begin{figure}[!ht]
  \begin{center}
    \includegraphics[width=0.99\textwidth]{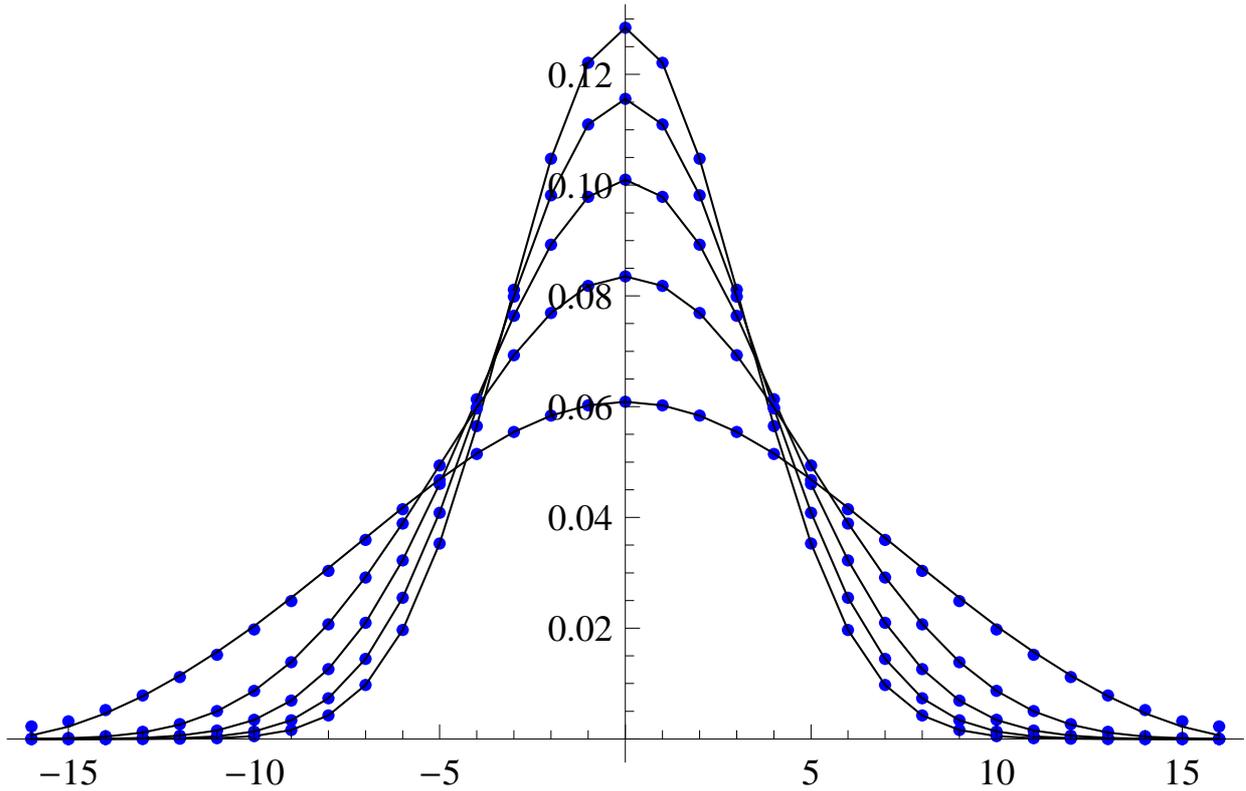}
  \end{center}
  \caption{Magnetisation distributions for $C_{32}$ (points) and
    fitted $\pqpr{n}{k}{p}{q}$ (lines) vs $k-n/2$ for several
    temperatures.}
    \label{fig:1d}
\end{figure}
\begin{figure}[!ht]
  \begin{minipage}{0.5\textwidth}
    \begin{center}
      \includegraphics[width=0.99\textwidth]{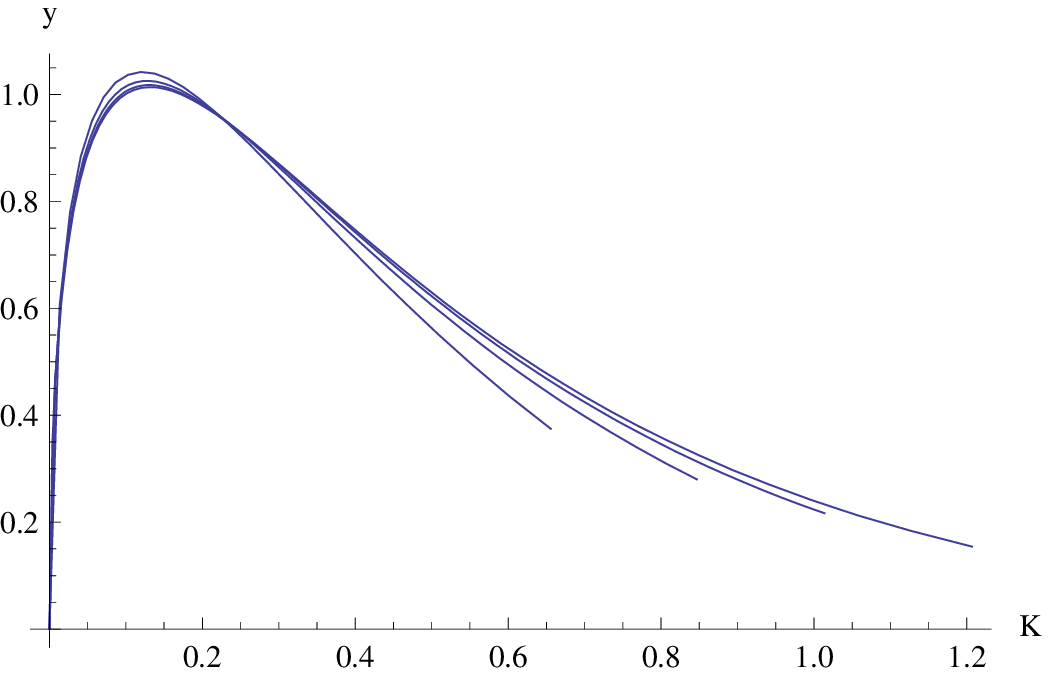}
    \end{center}
  \end{minipage}%
  \begin{minipage}{0.5\textwidth}
    \begin{center}
      \includegraphics[width=0.99\textwidth]{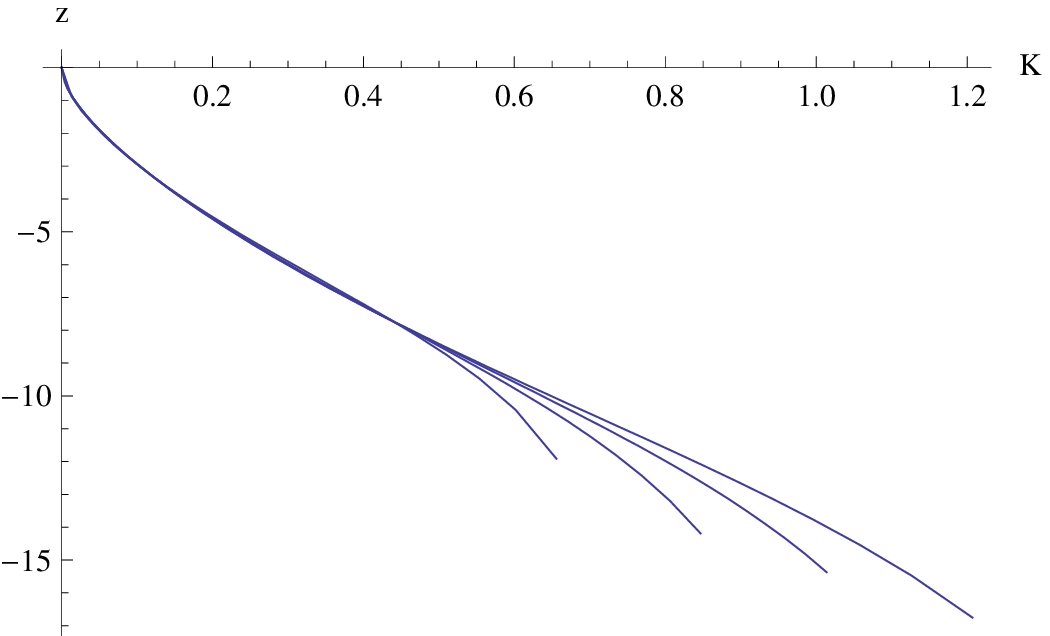}
    \end{center}
  \end{minipage}
  \caption{Left: $y=n\,(p-1)$ versus $K$ for $C_n$. Right:
    $z=n\,(q-1)$ versus $K$ for $C_n$.  Both plots are for
    $n=16,32,64,128$ (larger cycles stretch farther to the right).}
    \label{fig:1d0}
\end{figure}
In figure~\ref{fig:1d1} we see $y=n\,(p-1)$ versus $z=n\,(q-1)$ for
different $n$. The right plot of
figure~\ref{fig:1d1} shows the value at $K$ that gives the maximum
value of $y$.  The fitted straight line gives the limit $0.1333$, very
close to $2/15$.
\begin{figure}[!ht]
  \begin{minipage}{0.5\textwidth}
    \begin{center}
      \includegraphics[width=0.99\textwidth]{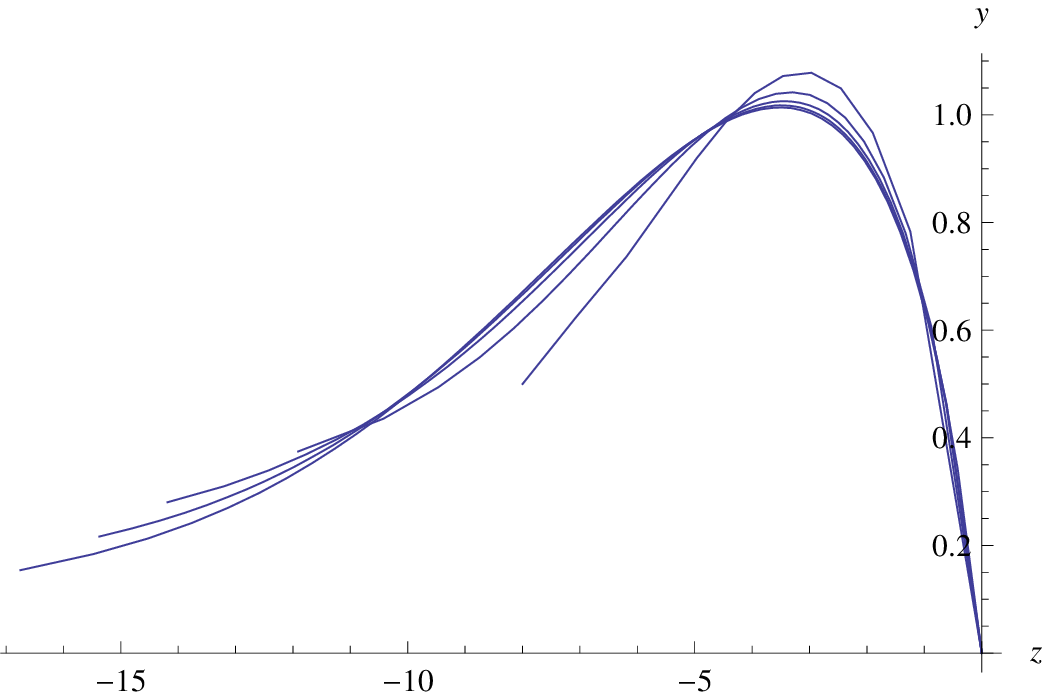}
    \end{center}
  \end{minipage}%
  \begin{minipage}{0.5\textwidth}
    \begin{center}
      \includegraphics[width=0.99\textwidth]{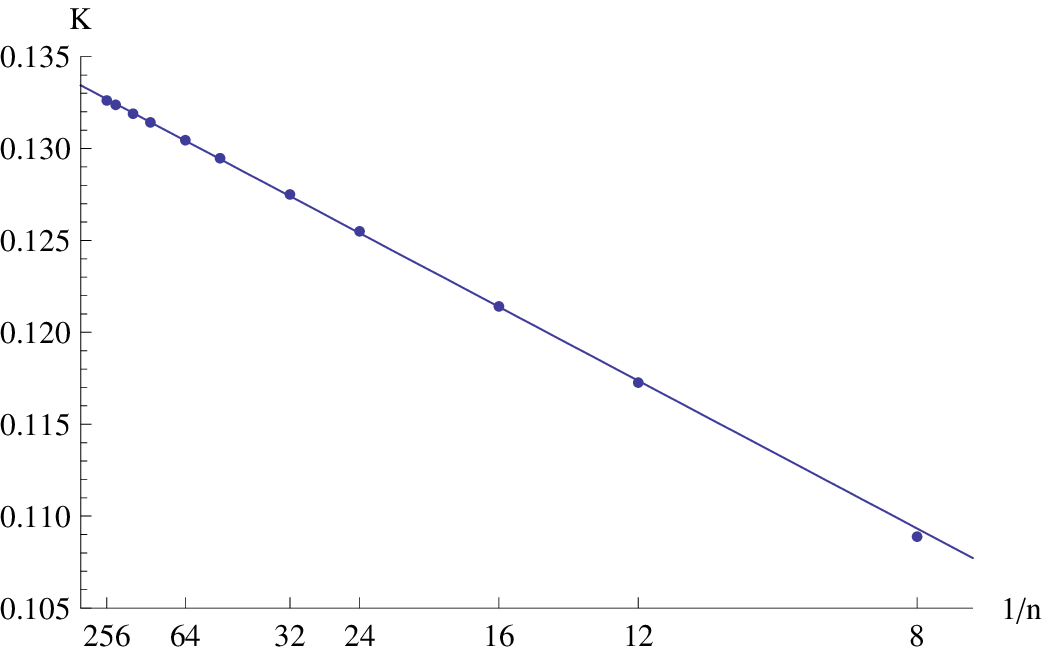}
    \end{center}
  \end{minipage}
  \caption{Left: $y=n\,(p-1)$ versus $z=n\,(q-1)$ for $C_n$, with
    $n=8, 16, 32, 64, 128$ with larger $n$ extending farther to the
    left. Right: $K$ giving the maximum $y$ vs $1/n$ for $C_n$,
    $n=8$,\,12, 16, 24, 32, 48, 64, 96, 128, 192, 256.}
    \label{fig:1d1}
\end{figure}
What about the values of $y$ and $z$? Indeed they converge beautifully
as figure~\ref{fig:1d2} indicates. The limit for $y$ is about $1.010$
and $z$ approaches a value of $-3.537$. 
\begin{figure}[!ht]
  \begin{minipage}{0.5\textwidth}
    \begin{center}
      \includegraphics[width=0.99\textwidth]{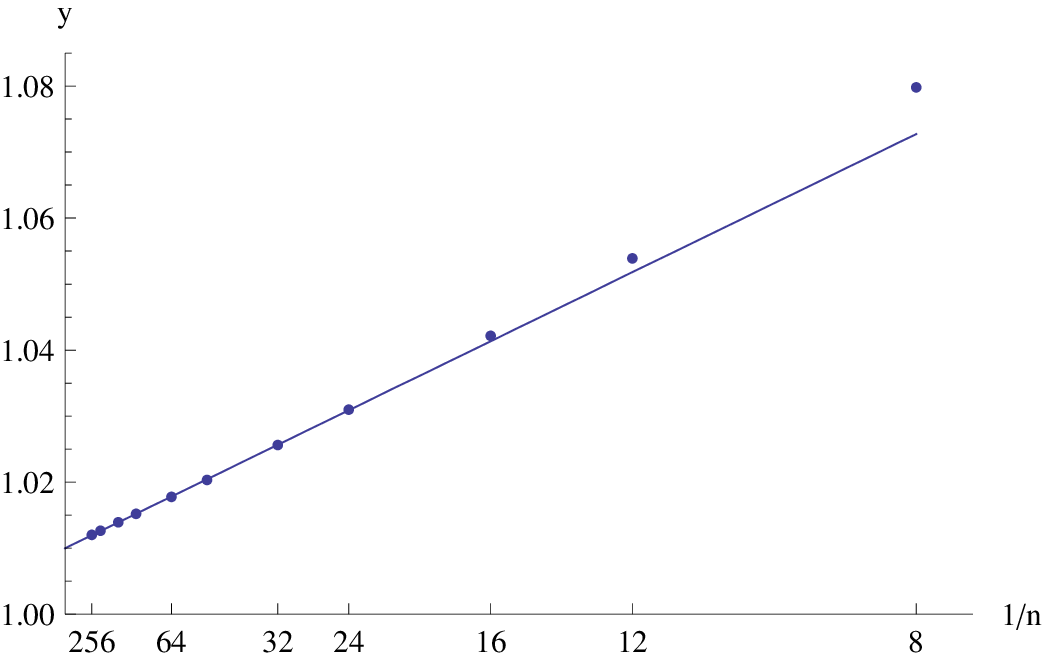}
    \end{center}
  \end{minipage}%
  \begin{minipage}{0.5\textwidth}
    \begin{center}
      \includegraphics[width=0.99\textwidth]{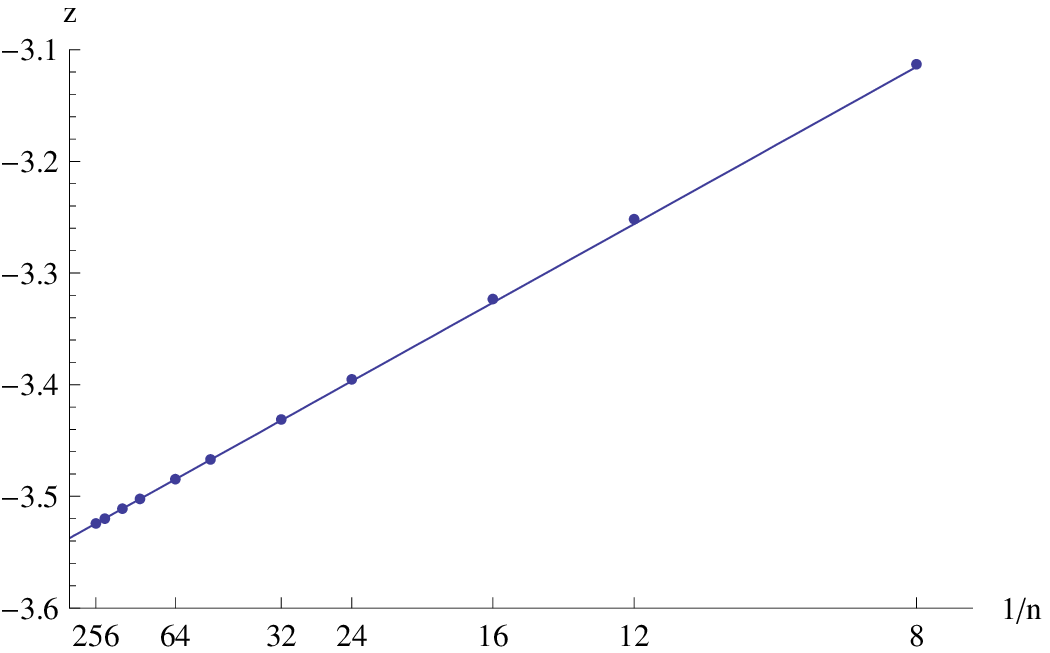}
    \end{center}
  \end{minipage}
  \caption{Left: Maximum value of $y=n\,(p-1)$ versus $1/n$ for $C_n$.
    Right: value of $z\,(q-1)$ versus $1/n$ for $C_n$ when $y$ is at
    its maximum. In both cases $n=8$,\,12, 16, 24, 32, 48, 64, 96,
    128, 192, 256.}
    \label{fig:1d2}
\end{figure}
Though we can not exactly solve what $y$ and $z$ should be at $K=2/15$
we can at least see how $y$ and $z$ relate at this point. For an
infinite 1-dimensional lattice we have that $\chi=e^{2\,K}$, see
e.g. \cite{baxter:82}. The second moment then should behave as
\begin{equation}
  \kmom{2} \sim \frac{n\,\chi}{4} = \frac{n\,e^{2\,K}}{4}
\end{equation}
Let $\ell=k-n/2$ and $\sigma=\sqrt{\kmom{2}}$. For high temperatures
we expect $\ell/\sigma$ to be normally distributed and thus
\begin{equation}
  \pr{\ell} \sim 
  \frac{\exp\left(-(\ell/\sigma)^2/2\right)}{\sigma\,\sqrt{2\,\pi}}
\end{equation}
The probability ratio is then
\begin{equation}
  \ratio{n}{n/2}{\ell}=\frac{\pr{\ell}}{\pr{0}} =
  \exp\left(-\ell^2/2\,\sigma^2\right)
\end{equation}
and for $\ell=1$ this simplifies to
\begin{equation}
  \ratio{n}{n/2}{1} = \exp\left(-1/2\,\sigma^2\right) =
  \exp\left(-2\,e^{-2\,K}/n\right) \sim 1 - \frac{2\,e^{-2\,K}}{n}
\end{equation}
Compare this with \eqref{pqratio8}. We thus have $a=-2\,e^{-2\,K}$.
Now $y$ and $z$ are related as $y=2\,w-a$ where $w$ is defined by
\eqref{wexpr8}. If we set $K=2/15$ then $a=-1.531857$, and choosing
$z=-3.537$ indeed gives us $y=1.01002$. To actually solve $z$ as a
function of $K$ seems harder though. However, numerical experimentation 
suggests that $y$ and $z$ for small $K$ behave as
\begin{gather}
  y(K) \approx c_1\,\sqrt{K} + c_2\,K \\
  z(K) \approx -c_1\,\sqrt{K} + c_2\,K
\end{gather}
where $c_1\approx 6.164$ and $c_2\approx -10.33$.

Strangely, when the $p,q$-distribution fit the magnetisation
distribution so well one might think that the free energy would be
well approximated by \eqref{free}. This is not so. The
$p,q$-approximation differs clearly from the asymptotic free energy,
given by $\log\left(2\,\cosh K\right)$.

\subsection{2D-lattices}
For the 2-dimensional lattices we can rely on exact data only for up
to $L=16$ and they were computed according to the method in
\cite{lundow:02}. We have sampled data for $L=32, 64, 128, 256, 512$,
collected with the methods described in \cite{sampart} and
\cite{reconart}. These methods gave us the energy distribution and
then it is just a matter of combining this with the distribution of
magnetisations for each given energy as described in
\cite{cubeart}. Figure~\ref{fig:2d0} shows an example of some
distributions for the $128\times 128$-lattice together with their
fitted $p,q$-binomial distributions. The fit is fairly good, but
hardly excellent near $K^*$. However, as the figure shows, at
$K=0.4388$ (i.e. for $L=128$) the fit is practically spot on. For the
lattices we have studied there is always one such temperature where
the $p,q$-distribution fit particularly well. This point is located
between $K^*$ and $K_c$ and is very close to, but not exactly equal
to, the point where the susceptibility is at its maximum.

Of course, for high temperatures (small $K$) and low temperatures
(high $K$) the fit is typically very good but in the high-temperature
region the measured $y$ and $z$ are unfortunately extremely sensitive
to noise. As we get closer to the critical region where the
distribution becomes bimodal this problem goes away, even though the
sampled distributions are more noisy there.  Regarding the free energy
it is well-fitted by \eqref{free} for low temperatures $K>K^*$ though
less well for high temperatures $K<K^*$.

\begin{figure}[!ht]
  \begin{center}
    \includegraphics[width=0.99\textwidth]{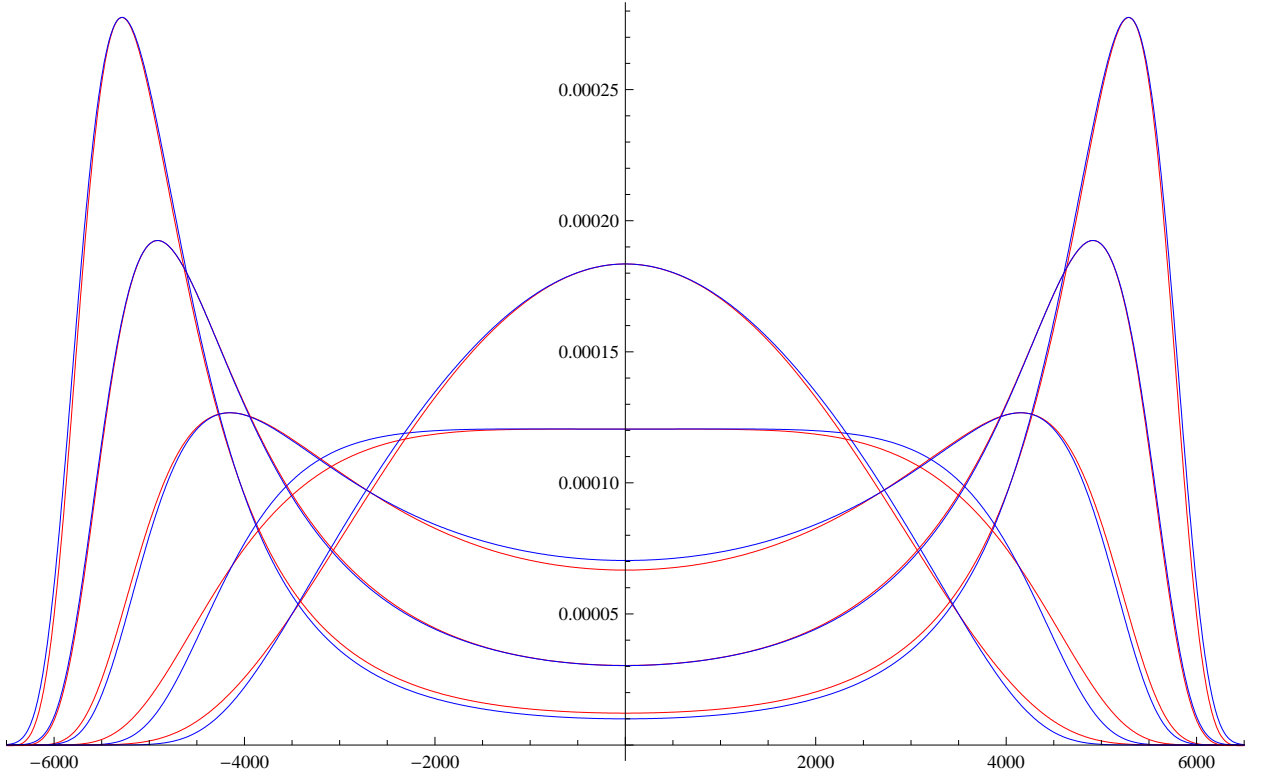}
  \end{center}
  \caption{Magnetisation distributions for the $128\times 128$-lattice
    (red) together with fitted $p,q$-binomial distributions
    $\pqpr{n}{k}{p}{q}$ (blue) vs $k-n/2$ at $K=0.432$, $K^*=0.43467$,
    $K=0.437$, $K=0.4388$ and $K_c=0.44068$ (downwards at $y$-axis). }
    \label{fig:2d0}
\end{figure}

Recall from section~\ref{sec:exponents} how the exponents of the
moment growth rates could be computed if we allow $z$ to depend on
$n$. For the 2D-lattices it is known that $\beta=1/8$, $\gamma=7/4$
and $\nu=1$, see~\cite{onsager:44}, \cite{yang:52} and
\cite{abraham:78}.  Thus the first moment $\kmom{1}$ should scale as
$n^{15/16}$ and the second moment $\kmom{2}$ as $n^{15/8}$. From
equation~\eqref{logmom1} and \eqref{logmom2} this would be achieved by
choosing $\lambda_1=-3/2$, $\lambda_2=-6$ and $\lambda_3=0$.  The left
plot of Figure~\ref{fig:2d1} shows $z$ versus $\log n$ at $K^*$
together with the curve $3-1.5\,\log n-6\log\log n$. The constant
$\lambda_0$ is chosen only to make the curve look plausibly near the
points. The point for $L=512$ deviate slightly but we suspect that
noise in the sampled data explains this. With $\lambda_0=3$ the
coefficient of $n^{15/8}$ obtained from \eqref{logmom2} would be
$0.301$ though the measured $\kmom{2}$ divided by $n^{15/8}$ are
closer $0.08$. To get this we have to choose $\lambda_0\approx
8.3$. In that case the convergence is extremely slow. Note also that
the fitted $p,q$-distribution is far from perfect which would
contribute some amount of error as well.

\begin{figure}[!ht]
  \begin{minipage}{0.5\textwidth}
    \begin{center}
      \includegraphics[width=0.99\textwidth]{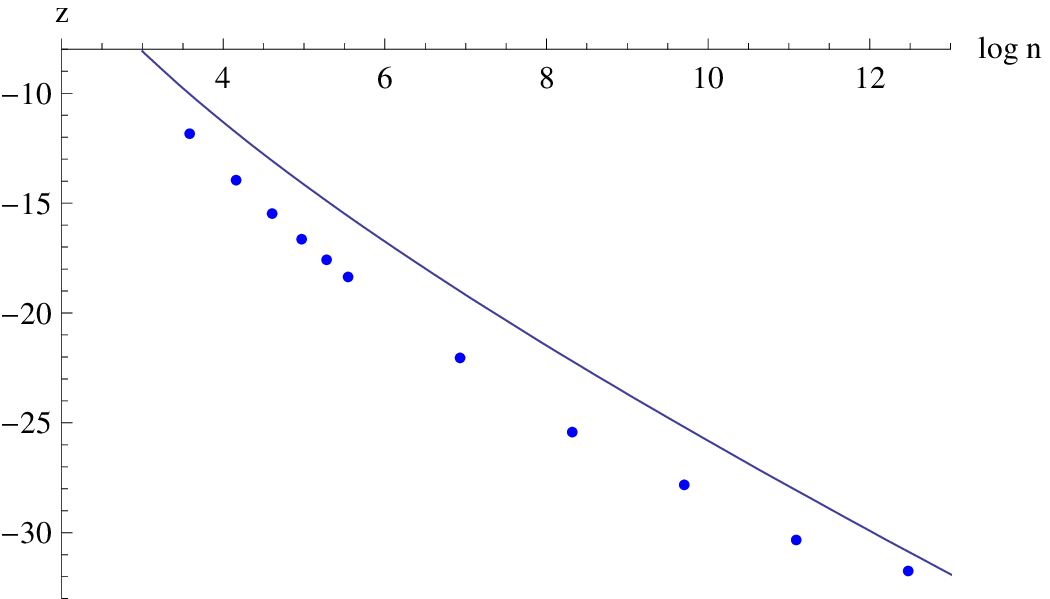}
    \end{center}
  \end{minipage}%
  \begin{minipage}{0.5\textwidth}
    \begin{center}
      \includegraphics[width=0.99\textwidth]{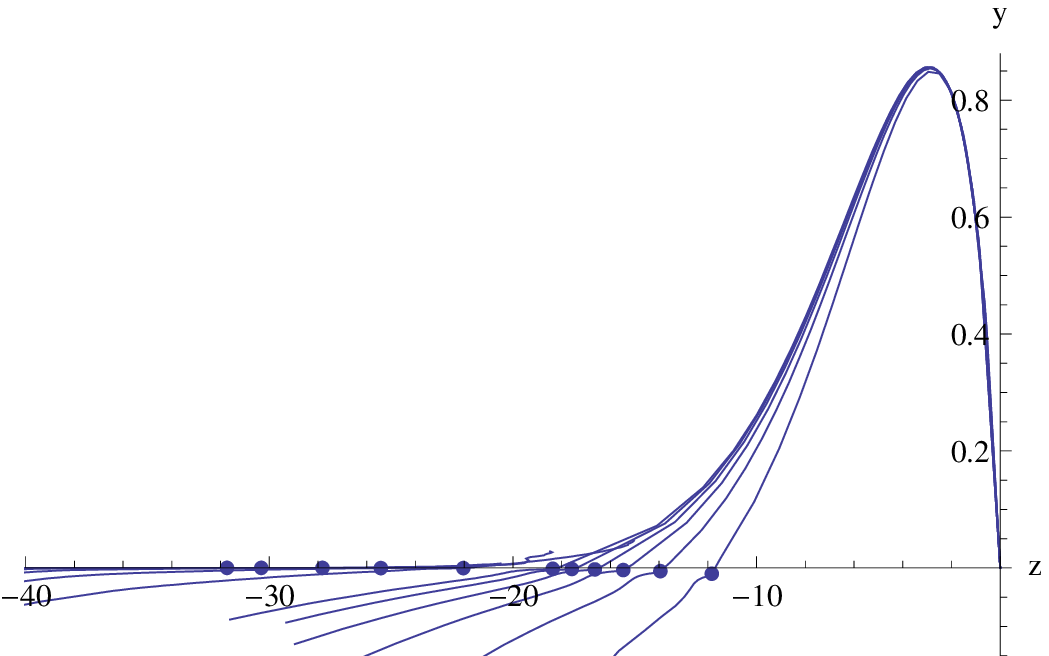}
    \end{center}
  \end{minipage}
  \caption{Left: $z=n\,(q-1)$ vs $\log n$ at $K^*$ for the $L\times
    L$-lattice, $L=6$,\, 8, 10, 12, 14, 16, 32, 64, 128, 256,
    512. The curve is $3-1.5\,\log n-6\,\log\log n$. Right:
    $y=n\,(p-1)$ vs $z=n\,(q-1)$ for the $L\times L$-lattice,
    $L=6$,\,8, 10, 12, 14, 16, 32, 64, 128, 256, 512 (512 barely
    visible near the $z$-axis). The points represent $K^*$.}
    \label{fig:2d1}
\end{figure}

The right plot of figure~\ref{fig:2d1} shows $y$ vs $z$ for a range of
temperatures. The points representing $K^*$ may appear to lie on the
$z$-axis but they are are slightly below it.  In the 1D-case we
suspected that there is a limit curve for the high-temperature region,
but we suspect that the exact data that produced this part of the plot
rely on far too small lattices to give any conclusive evidence. Also,
the $p,q$-find algorithm is rather sensitive to noise in this region
to be useful for sampled data. However, as we said before, this
problem goes away once $K\ge K^*$. Figure~\ref{fig:2d2} shows $y$ and
$z$ versus $K$ for all the lattices though for the sampled data we
only show low-temperature data. The red line is located at
$K_c=\atanh{\left(\sqrt{2}-1\right)}\approx 0.44068$.
\begin{figure}[!ht]
  \begin{minipage}{0.5\textwidth}
    \begin{center}
      \includegraphics[width=0.99\textwidth]{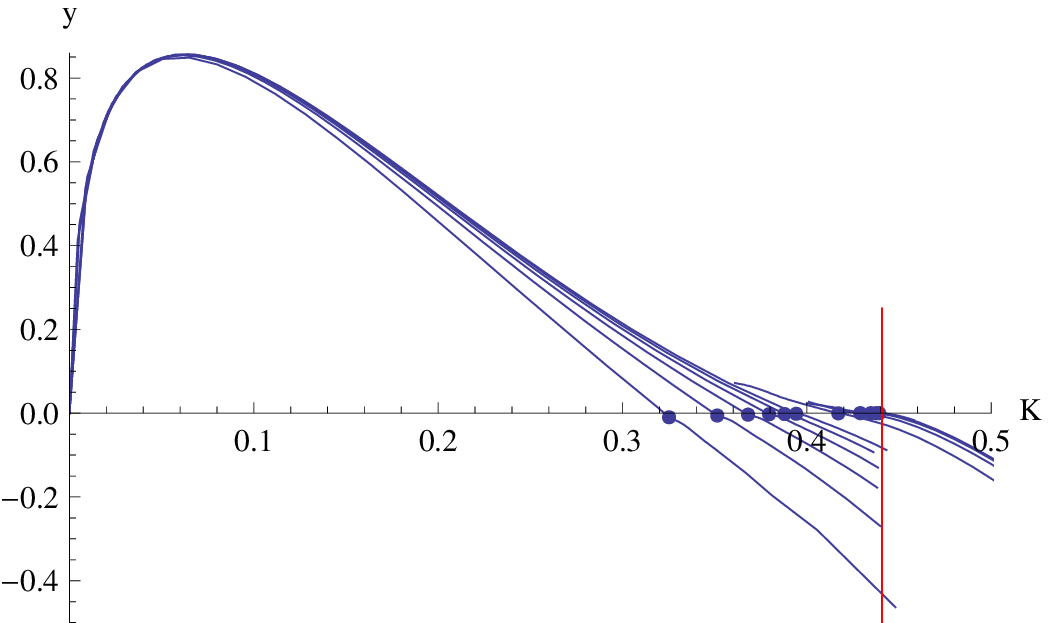}
    \end{center}
  \end{minipage}%
  \begin{minipage}{0.5\textwidth}
    \begin{center}
      \includegraphics[width=0.99\textwidth]{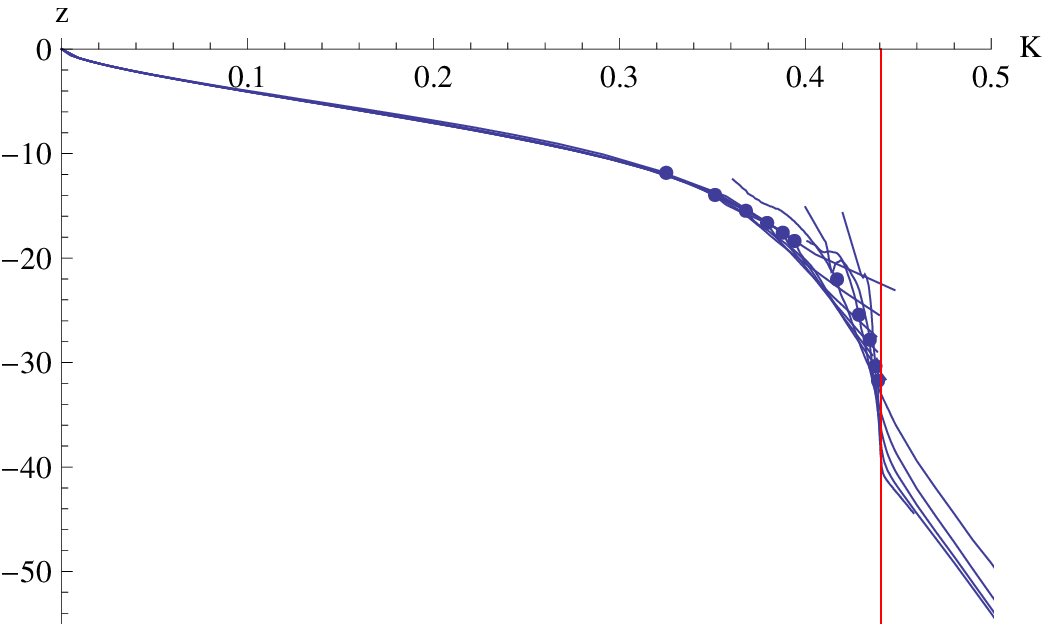}
    \end{center}
  \end{minipage}
  \caption{$y=n\,(p-1)$ vs $K$ (left) and $z=n\,(q-1)$ vs $K$ (right)
    for the $L\times L$-lattice, $L=6$,\,8, 10, 12, 14, 16, 32, 64, 128,
    256, 512. The points represent $K^*$ and the red line is at
    $K_c$. The larger lattices have their points farther to the right
    in the plots.}
    \label{fig:2d2}
\end{figure}

\subsection{3D-lattices}
For these lattices we only have exact data for $L=4$ and sampled data
for $L=6,8,12,16,32,64$.  The situation is actually somewhat better
for 3D-lattices. Figure~\ref{fig:3d0} shows some distributions in the
vicinity of $K^*$ for $L=32$ together with the fitted
$p,q$-distributions.  For $K\ge K^*$, just when the distributions
become bimodal, the fit is certainly less than perfect, but near $K^*$
the $p,q$-approximation is actually rather good.
\begin{figure}[!ht]
  \begin{center}
    \includegraphics[width=0.99\textwidth]{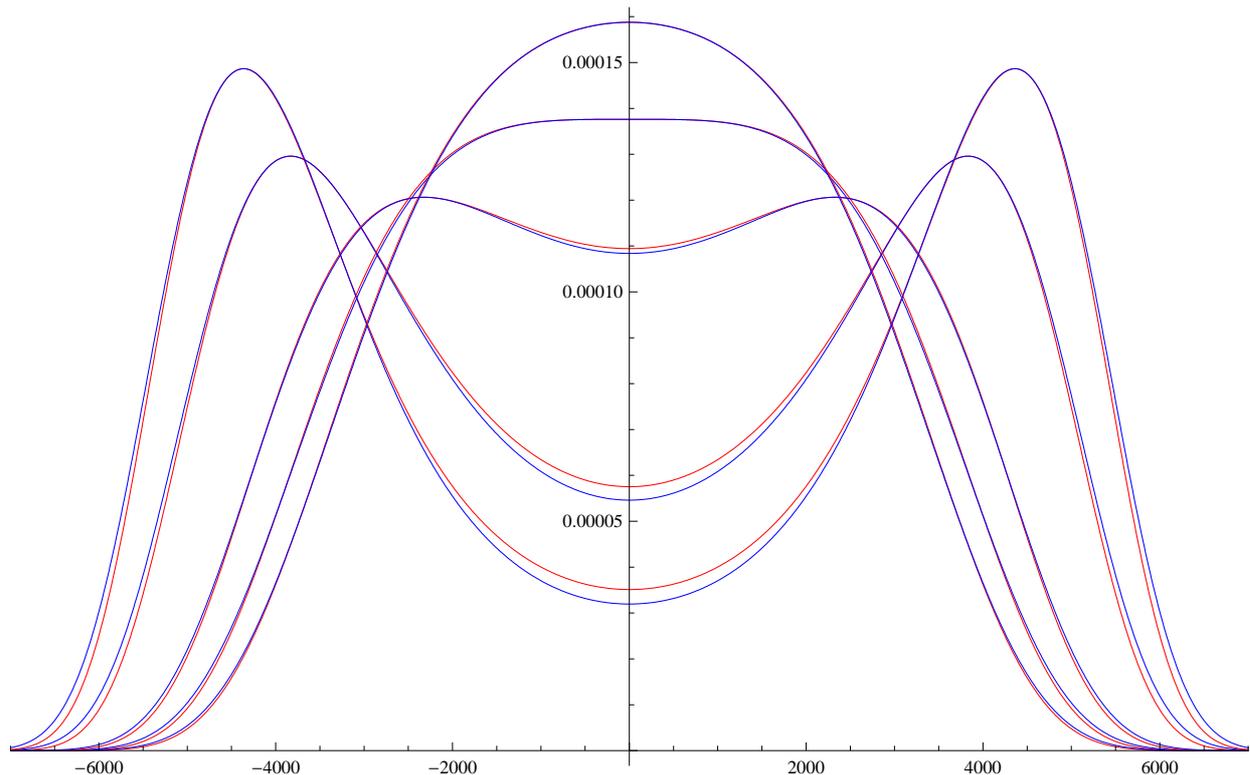}
  \end{center}
  \caption{Magnetisation distributions for the $32\times 32\times
    32$-lattice (red) and the fitted $\pqpr{n}{k}{p}{q}$ (blue) vs
    $k-n/2$ for $K=0.2204$, $K^*=0.22066$, $K=0.2210$, $K_c=0.2216546$ and
    $0.2220$ (downwards at the y-axis).}
    \label{fig:3d0}
\end{figure}

In the left plot of figure~\ref{fig:3d1} we show $z$ versus $\log n$
at $K^*$.  The fitted line through the points corresponds to
$z=-5.3-\log n$ and is not too bad an approximation. However, in
\cite{cubeart} it was estimated that the growth rate exponent at $K_c$
of the susceptibility is $\gamma/\nu=1.978\pm0.009$ (assuming
$\gamma=\gamma'$ and $\nu=\nu'$). For the magnetisation it was
estimated $\beta/\nu=0.5147\pm 0.0007$. Translated into exponents of
$n$ this means $1.657\le 1+\gamma/3\nu \le 1.663$ and $0.8282\le
1-\beta/3\nu\le 0.8287$. If we choose $\lambda_1=-5/8$ in
\eqref{logmom1} and \eqref{logmom2} the first moment exponent would be
$53/64=0.828125$ and $53/32=1.65625$ for the second moment, slightly
below the lower bound of the estimate intervals. Choosing
$\lambda_1=-2/3$ would give exponents $5/6=0.8333\ldots$ and
$5/3=1.666\ldots$ respectively, slightly above the upper bound of the
intervals.  Let us suggest, as an example, that $\lambda_0=6.8$,
$\lambda_1=-2/3$, $\lambda_2=-6$ and $\lambda_3=0$ in the expression
\eqref{logmom2}. In figure~\ref{fig:3d1} the curve use these
parameters for $z$ at $K^*$, i.e. $z=6.8-(2/3)\,\log n-6\,\log\log n$.
Will the points eventually converge to the curve? It would take
considerably larger lattices to shed any light on this. We also have
the problem what $\lambda_0$ should be. Using $\lambda_0=6.8$
means that the coefficient in \eqref{logmom2} is about
$0.393$. Comparing the measured $\kmom{2}$ with $n^{5/3}$ gives a
factor of roughly $0.16$ though the data are certainly far from
conclusive. Since the distribution fit is not perfect a different
constant is perhaps to be expected. Also, slow convergence is to be
expected here.

\begin{figure}[!ht]
  \begin{minipage}{0.5\textwidth}
    \begin{center}
      \includegraphics[width=0.99\textwidth]{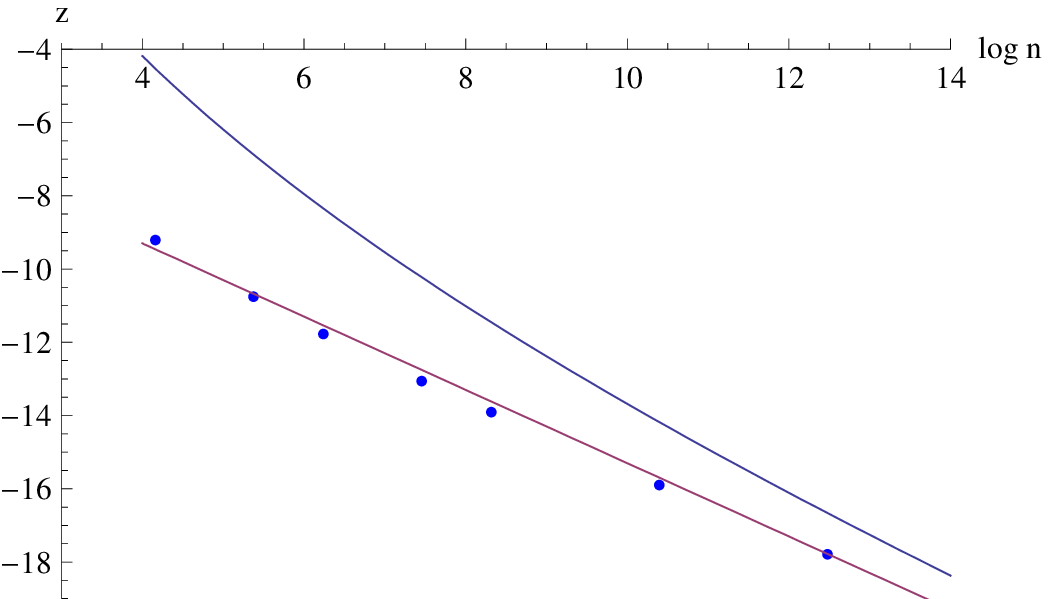}
    \end{center}
  \end{minipage}%
  \begin{minipage}{0.5\textwidth}
    \begin{center}
      \includegraphics[width=0.99\textwidth]{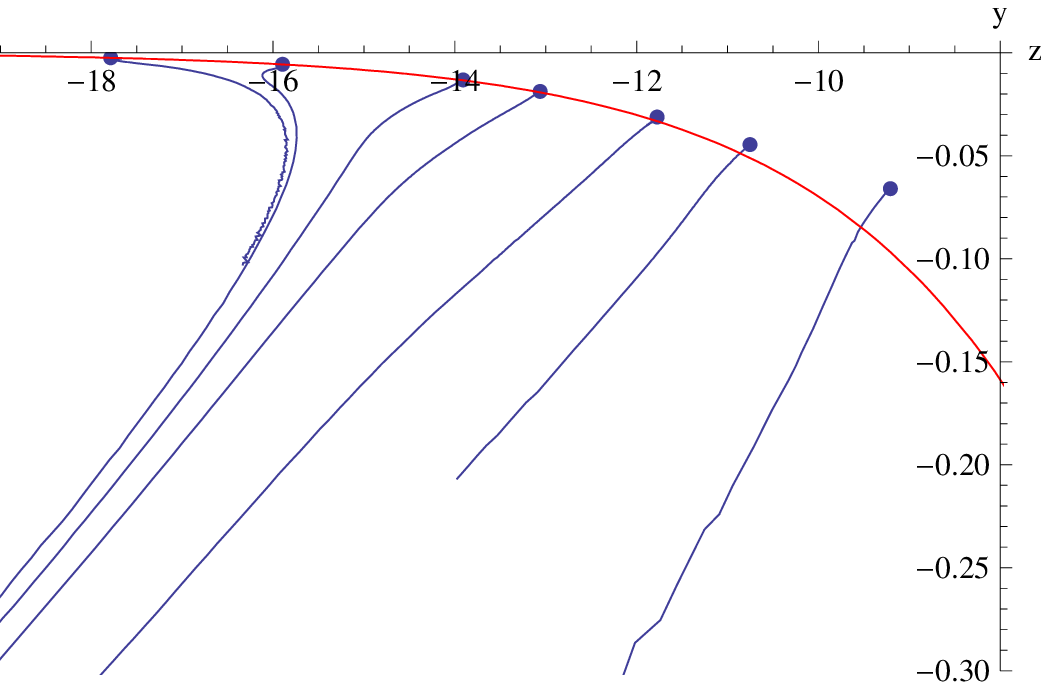}
    \end{center}
  \end{minipage}
  \caption{Left: $z=n\,(q-1)$ vs $\log n$ at $K^*$ for the $L\times
    L\times L$-lattice, $L=4,6,8,12,16,32,64$. The line though the
    points is $-5.3-\log n$ and the curve is $6.8-(2/3)\,\log
    n-6\,\log\log n$. Right: $y=n\,(p-1)$ vs $z$ for the $L\times
    L\times L$-lattice, $L=4,6,8,12,16,32,64$ (leftwards) for
    $K>K^*$. Higher values of $K$ when we move downwards left. The red
    curve is $y=2\,w$ with $w$ defined by \eqref{wdef}.}
  \label{fig:3d1}
\end{figure}

The right plot of figure~\ref{fig:3d1} shows $y$ versus $z$ for
$K>K^*$. Note the peculiar backwards movement of $z$ getting more and
more pronounced for larger $L$. The curves for $16$, $32$ and $64$
show signs of approaching some limit curve. We don't have data for
very low temperatures for the smaller lattices though, except for
$L=4$. The plots in figure~\ref{fig:3d2} shows $y$ and $z$ versus $K$
for $K>K^*$. The red lines show location of $K_c\approx 0.2216546$,
found in \cite{cubeart}, but see also \cite{rosengren:86} for a
theoretical estimate of $K_c$.

\begin{figure}[!ht]
  \begin{minipage}{0.5\textwidth}
    \begin{center}
      \includegraphics[width=0.99\textwidth]{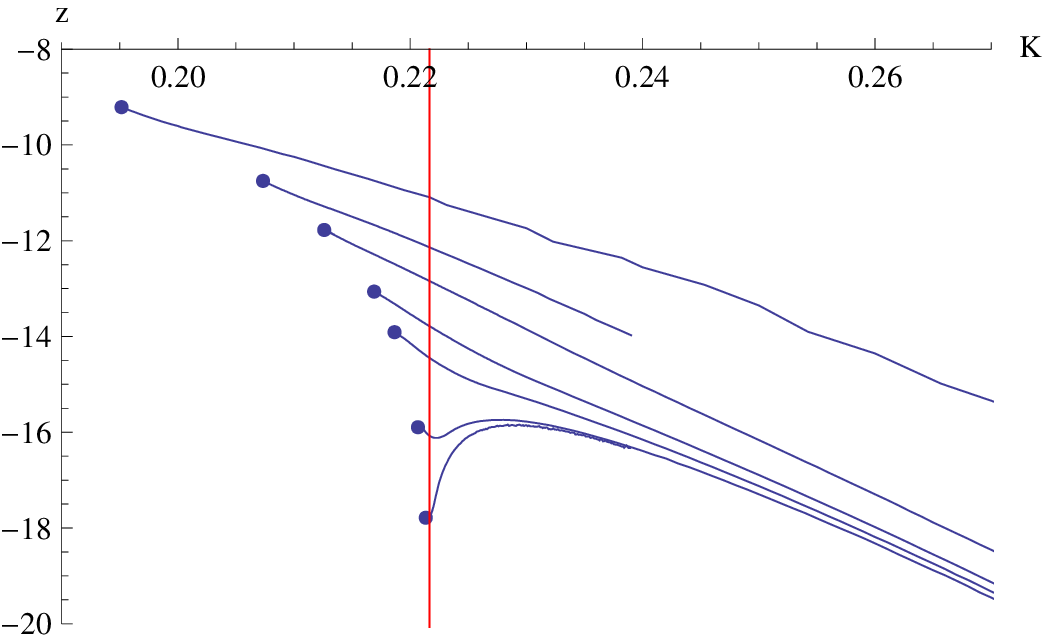}
    \end{center}
  \end{minipage}%
  \begin{minipage}{0.5\textwidth}
    \begin{center}
      \includegraphics[width=0.99\textwidth]{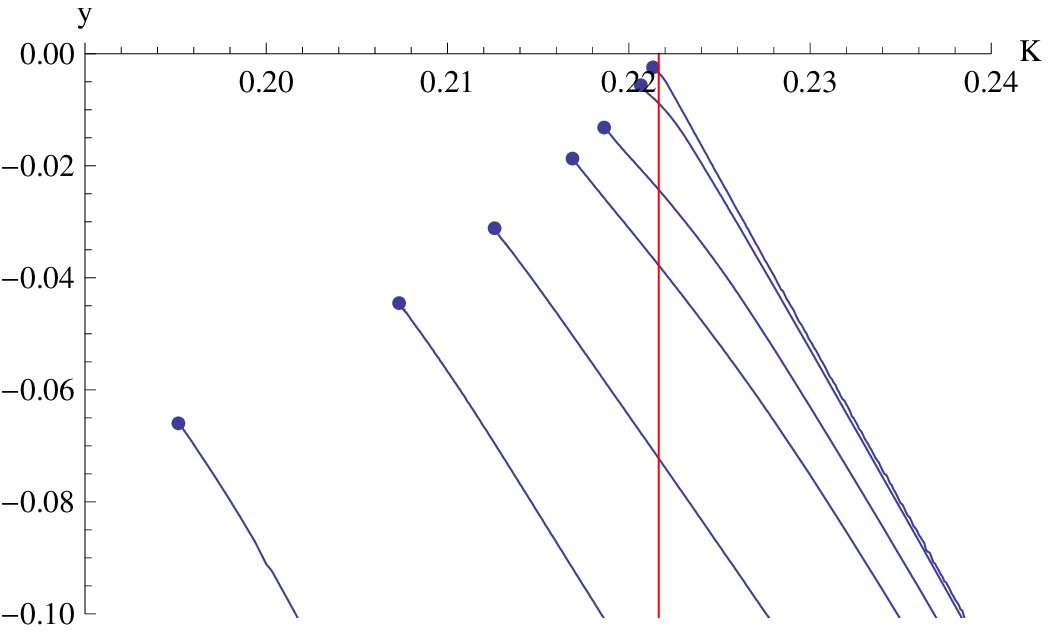}
    \end{center}
  \end{minipage}
  \caption{Left: $z=n\,(q-1)$ vs $K$ with $K>K^*$ for the $L\times
    L\times L$-lattice, $L=4,6,8,12,16,32,64$ (downwards).  Right:
    $y=n\,(p-1)$ vs $K$ with $K>K^*$ for the $L\times L\times
    L$-lattice, $L=4,6,8,12,16,32,64$ (upwards). In both plots the red
    line indicates location of $K_c$ and the points are the locations
    of $K^*$.}
  \label{fig:3d2}
\end{figure}

\subsection{4D-lattices}
In the case of 4-dimensional lattices we have sampled data of
magnetisation distributions for
$L=4,6,8,10,12,16$. Figure~\ref{fig:4d0} shows some of these
magnetisation distributions for $L=12$ near $K^*$ together with fitted
$p,q$-binomial distributions. The fit is quite good, considerably
better than for 2D and 3D, in the whole range of selected
temperatures. Though it is hard to distinguish the fitted curves from
the magnetisation curves, there is a small deviation near the middle.
\begin{figure}[!ht]
  \begin{center}
    \includegraphics[width=0.99\textwidth]{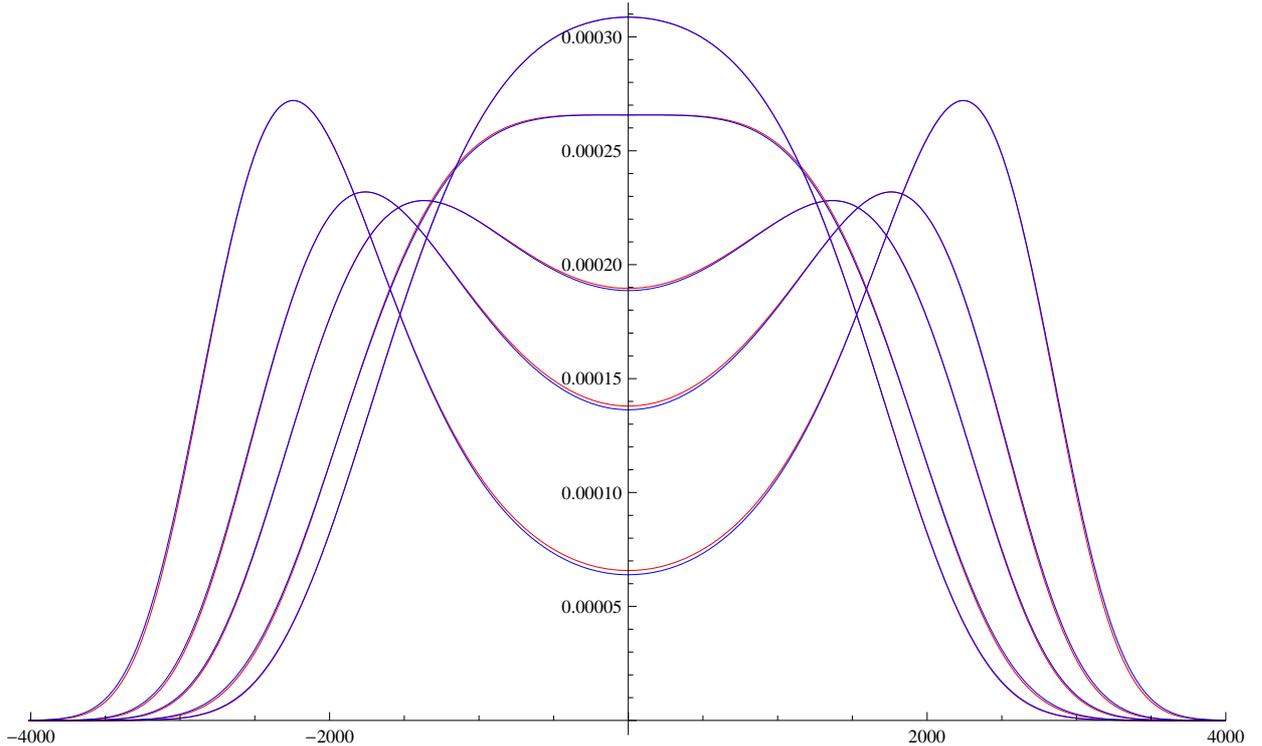}
  \end{center}
  \caption{Magnetisation distributions for the $12\times 12\times
    12\times 12$-lattice (red) and the fitted $\pqpr{n}{k}{p}{q}$
    (blue) vs $k-n/2$ for $K=0.1490$, $K^*=0.149255$, $K_c=0.149695$,
    $K=0.1500$ and $K=0.1505$ (downwards at the y-axis).}
    \label{fig:4d0}
\end{figure}
How should $z$ at $K^*$ depend on $n$? Actually, taking the data at
face-value they are rather well-fitted to the simple formula
$z=-6.5-0.45\,\log n$. However, for the 4D-lattice we have
$\gamma=\gamma'=1$, $\beta=1/2$ and $\nu=\nu'=1/2$. This gives that
$1+\gamma/d\,\nu = 3/2$ and $1-\beta/d\,\nu=3/4$. Moreover, according
to \cite{mon:90} there should be a correction to this. They
calculated, using renormalization group techniques, that the
susceptibility should scale as $L^2\,\sqrt{\log L}$ near $K_c$.  This
means that $\kmom{2}$ should scale as $n^{3/2}\,\sqrt{\log n}$. From
\eqref{loglogmom2} we see that we have to choose $\lambda_2=-2$, with
$\lambda_1=0$ and $\lambda_3=-6$, to obtain this. In the left plot of
figure~\ref{fig:4d1} we have set $z=-1.2-2\,\log\log n-6\log\log\log
n$ and plotted it versus $\log\log n$. The curve would then behave as
a limit curve rather than as a fitted curve. The choice of coefficient
$\lambda_0=-1.2$ is only supported by the human eye as a guide rather
than any theory and herein lies a problem. With this choice the
coefficient of \eqref{loglogmom2} is about $0.558$. However, dividing
the measured $\kmom{2}$ at the different $K^*$ with
$n^{3/2}\,\sqrt{\log n}$ gives values close to $0.15$.  This
discrepancy could be due to several sources; e.g. the expression in
\eqref{loglogmom2} could be incorrect or our data could be suffering
from very slow convergence.
\begin{figure}[!ht]
  \begin{minipage}{0.5\textwidth}
    \begin{center}
      \includegraphics[width=0.99\textwidth]{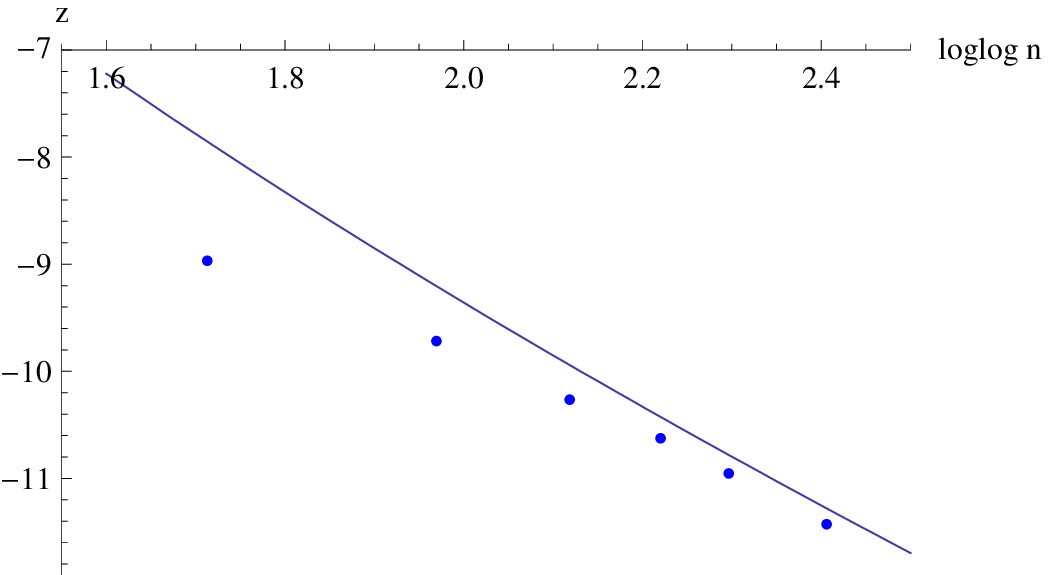}
    \end{center}
  \end{minipage}%
  \begin{minipage}{0.5\textwidth}
    \begin{center}
      \includegraphics[width=0.99\textwidth]{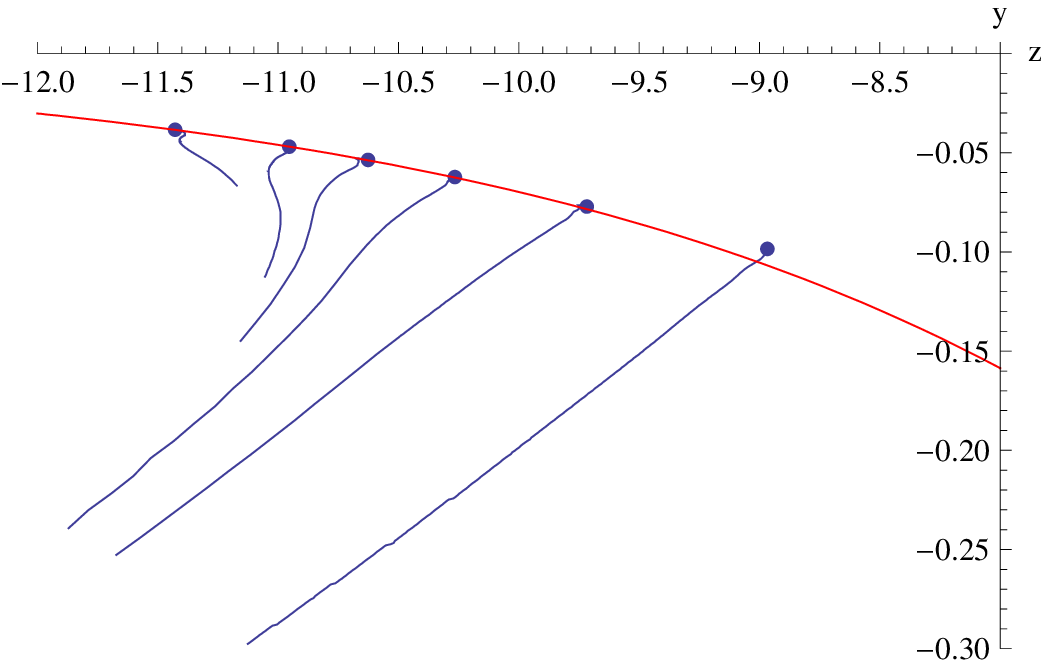}
    \end{center}
  \end{minipage}
  \caption{Left: $z=n\,(q-1)$ vs $\log\log n$ at $K^*$ for the
    $L\times L\times L\times L$-lattice, $L=4,6,8,10,12,16$. The curve
    is $-1.2-2\,\log\log n-6\,\log\log\log n$. Right: $y=n\,(p-1)$ vs
    $z=n\,(q-1)$ for the $L\times L\times L\times L$-lattice,
    $L=4,6,8,10,12,16$ (leftwards).  Higher temperatures (low $K$)
    begin at the upper right part of the plot and with lower
    temperatures we move down to the left. The red curve is $y=2\,w$
    with $w$ defined by \eqref{wdef}.}
  \label{fig:4d1}
\end{figure}
In the right plot of figure~\ref{fig:4d1} we show $y$ versus $z$ for
$K>K^*$ together with the curve $y=2\,w$ with $w$ defined by
\eqref{wdef}.  In figure~\ref{fig:4d2} we show $y$ and $z$ versus $K$
for $K>K^*$. The red line is located at $K_c\approx 0.1496497$,
estimated in~\cite{4dart}.
\begin{figure}[!ht]
  \begin{minipage}{0.5\textwidth}
    \begin{center}
      \includegraphics[width=0.99\textwidth]{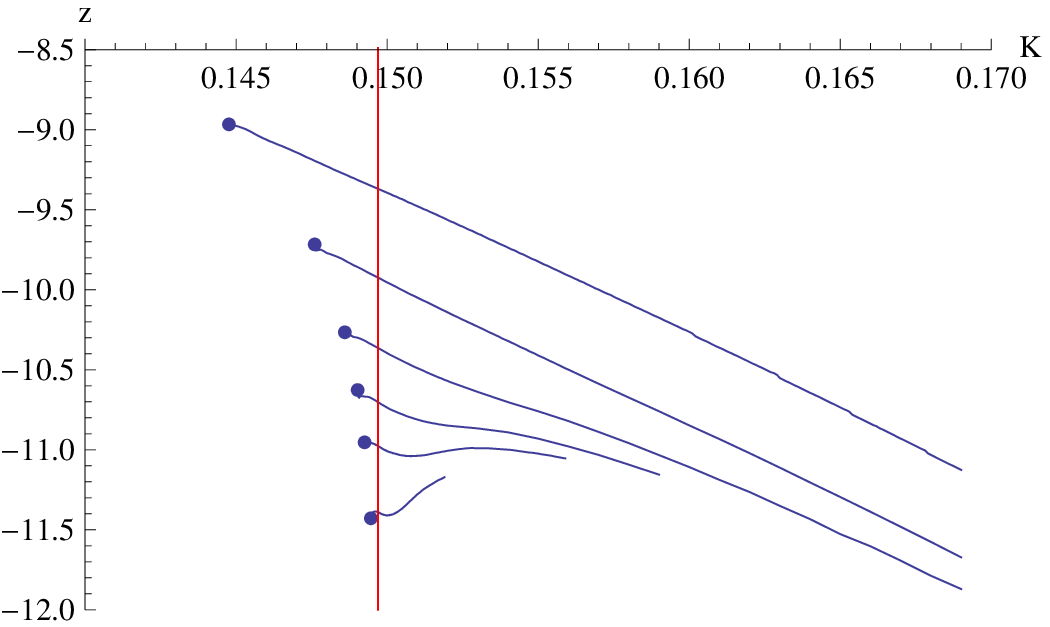}
    \end{center}
  \end{minipage}%
  \begin{minipage}{0.5\textwidth}
    \begin{center}
      \includegraphics[width=0.99\textwidth]{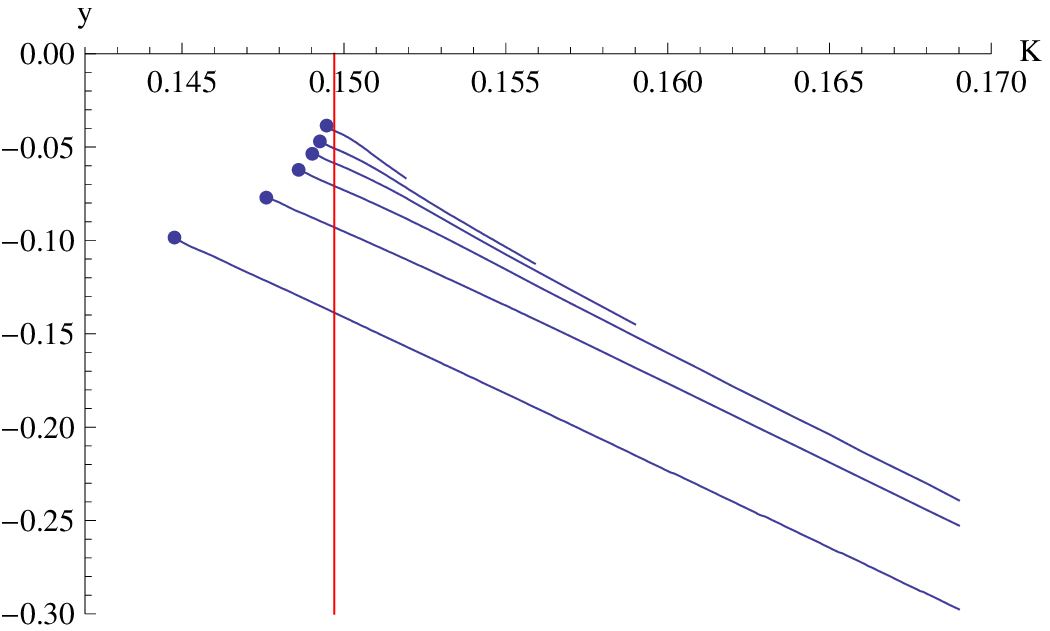}
    \end{center}
  \end{minipage}
  \caption{Left: $z=n\,(q-1)$ vs $K$ with $K>K^*$ for the $L\times
    L\times L\times L$-lattice, $L=4,6,8,10,12,16$ (downwards).
    Right: $y=n\,(p-1)$ vs $K$ with $K>K^*$ for the $L\times L\times
    L\times L$-lattice, $L=4,6,8,10,12,16$ (upwards). In both plots
    the red line indicates location of $K_c$ and the points are the
    locations of $K^*$.}
  \label{fig:4d2}
\end{figure}

\subsection{5D-lattices}
For the 5-dimensional lattices we have sampled data of magnetisation
distributions only for $L=4,6,8,10,12$. The distributions in
figure~\ref{fig:5d0} are extremely well fitted by $p,q$-binomial
distributions; it is almost impossible to tell them apart with the
naked eye.
\begin{figure}[!ht]
  \begin{center}
    \includegraphics[width=0.99\textwidth]{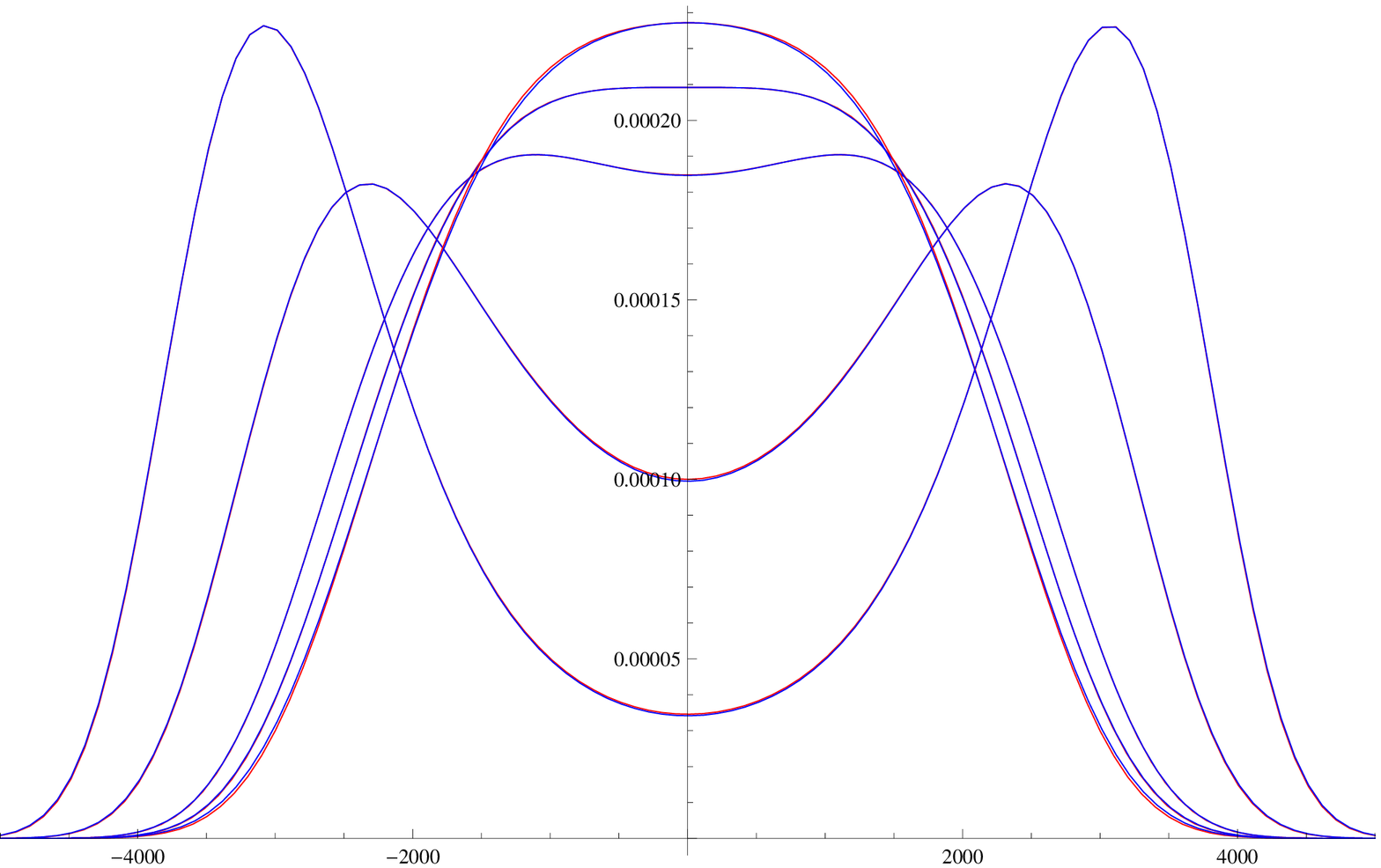}
  \end{center}
  \caption{Magnetisation distributions for $8\times 8\times 8\times
    8\times 8$-lattice (red) and the fitted $\pqpr{n}{k}{p}{q}$ (blue)
    vs $k-n/2$ for $K=0.1137$, $K^*=0.113786$, $K_c=0.113914$,
    $0.1143$ and $0.1147$.}
    \label{fig:5d0}
\end{figure}
In five dimensions the susceptibility near $K_c$ scales as $L^{5/2}$,
see \cite{cardy:96}. Thus $\kmom{2}$ should scale as $n^{3/2}$ which
is exactly what we receive when keeping $z$ fixed. So, for $z$
constant we obtain $\kmom{1} \propto n^{3/4}$ and $\kmom{2} \propto
n^{3/2}$.
\begin{figure}[!ht]
  \begin{minipage}{0.5\textwidth}
    \begin{center}
      \includegraphics[width=0.99\textwidth]{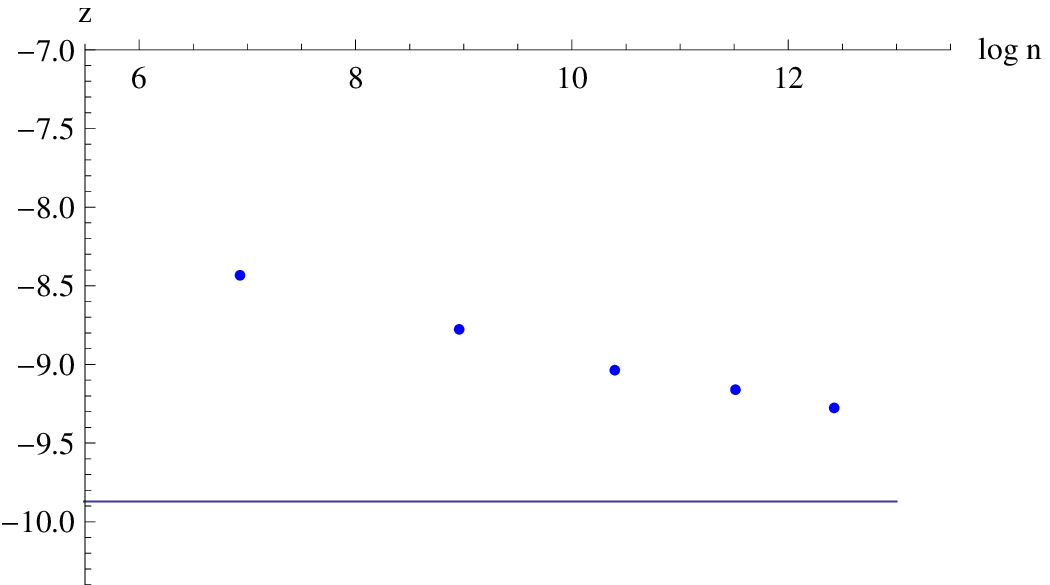}
    \end{center}
  \end{minipage}%
  \begin{minipage}{0.5\textwidth}
    \begin{center}
      \includegraphics[width=0.99\textwidth]{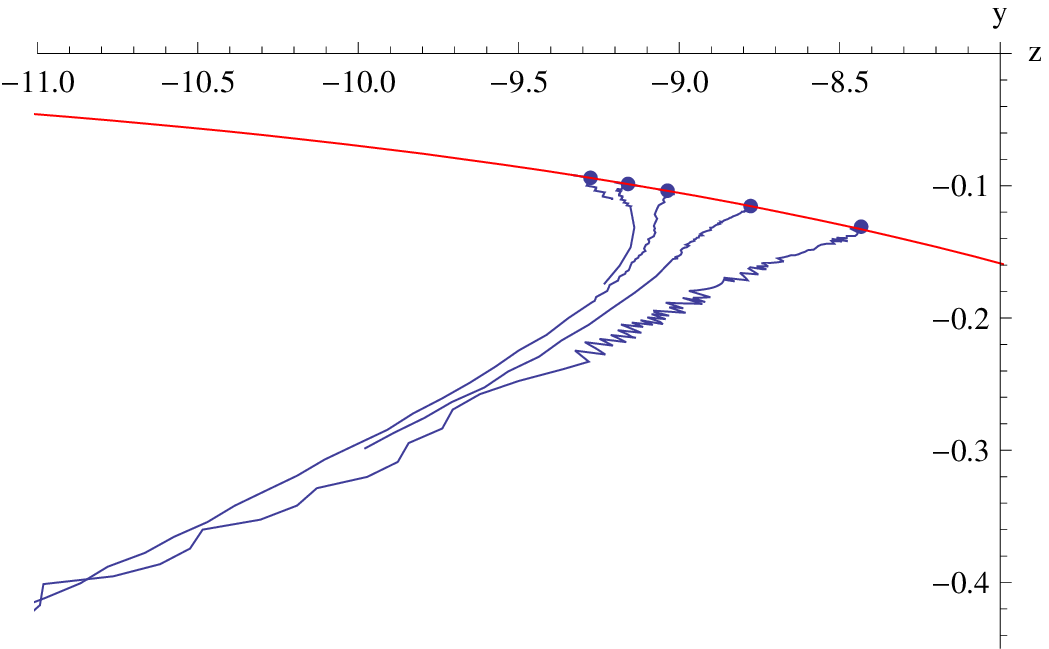}
    \end{center}
  \end{minipage}
  \caption{Left: $z=n\,(q-1)$ vs $\log n$ at $K^*$ for the $L\times
    L\times L\times L\times L$-lattice, $L=4,6,8,10,12$.
    The straight line is constant at $z=-9.87$. Right: $y=n\,(p-1)$ vs
    $z=n\,(q-1)$ for the $L\times L\times L\times L\times L$-lattice,
    $L=4,6,8,10,12$ (leftwards).  Higher temperatures (low $K$) begin
    at the upper right part of the plot and with lower temperatures we
    move down to the left. The red curve is $y=2\,w$ with $w$ defined
    by \eqref{wdef}.}
  \label{fig:5d1}
\end{figure}
The left plot of figure~\ref{fig:5d1} shows $z$ at $K^*$ for $L = 4,
6, 8, 10, 12$. If $z$ approaches a constant then what is the limit
value?  Extracting the limit $z$ from this plot is futile of course.
The right plot of figure~\ref{fig:5d1} shows $y$ vs $z$ for the
different lattices together with the points $K^*$ and the curve
$y=2\,w$. 
\begin{figure}[!ht]
  \begin{minipage}{0.5\textwidth}
    \begin{center}
      \includegraphics[width=0.99\textwidth]{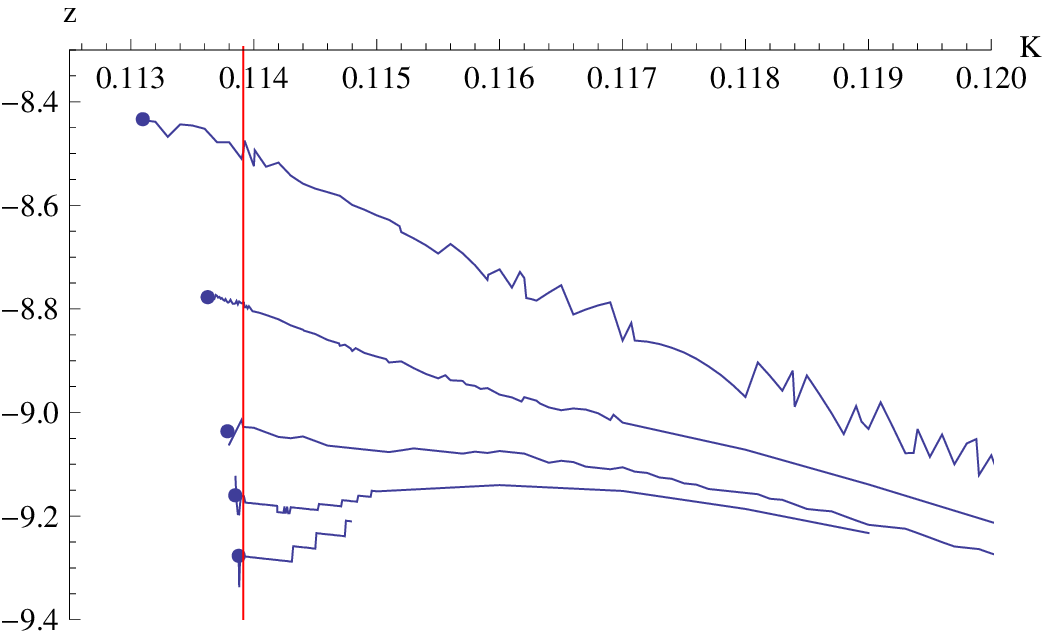}
    \end{center}
  \end{minipage}%
  \begin{minipage}{0.5\textwidth}
    \begin{center}
      \includegraphics[width=0.99\textwidth]{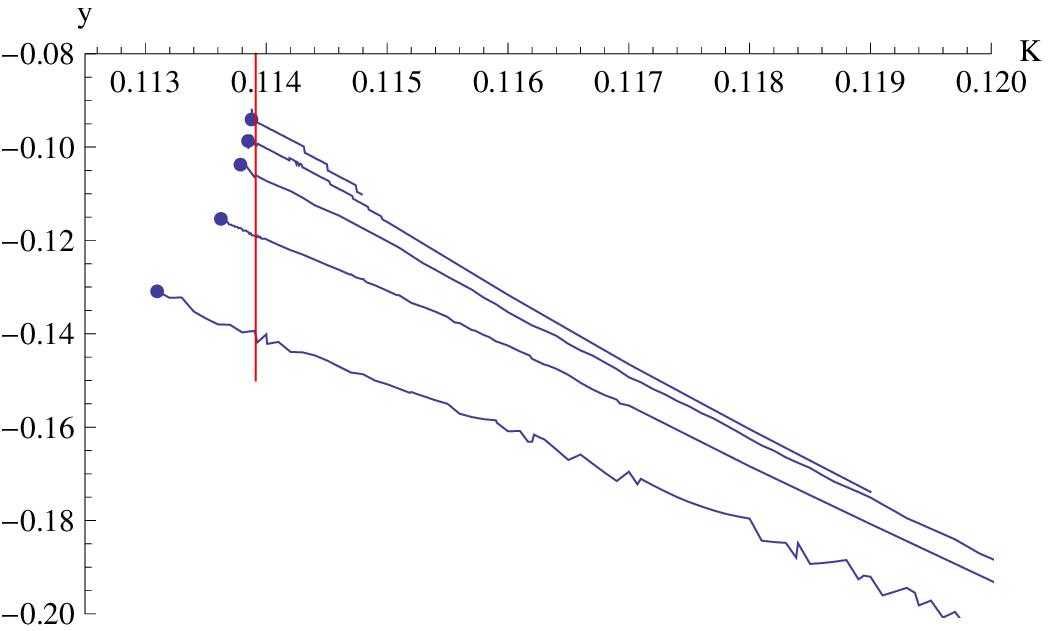}
    \end{center}
  \end{minipage}
  \caption{Left: $z=n\,(q-1)$ vs $K$ with $K>K^*$ for the $L\times
    L\times L\times L\times L$-lattice, $L=4,6,8,10,12$ (downwards).
    Right: $y=n\,(p-1)$ vs $K$ with $K>K^*$ for the $L\times L\times
    L\times L\times L$-lattice, $L=4,6,8,10,12$ (upwards). In both
    plots the red line indicates location of $K_c$ and the points are
    the locations of $K^*$.}
  \label{fig:5d2}
\end{figure}
In figure~\ref{fig:5d2} we show $y$ and $z$ versus $K$ for $K\ge K^*$
with an estimated $K_c$ marked as a red line. Despite the noise in the
plots it seems plausible that $z$ stays essentially constant very
close to $K^*$ (and $K_c$) and that only $y$ moves. Let us assume this
and see where this leads us. We employ the moment expressions in
section~\ref{sec:moments} in terms of the parameter $a$ to model the
behaviour near $K^*$. A normalised first cumulant of the absolute
magnetisation $\mean{\abs{M}}/2\,n^{3/4}=\kmom{1}/n^{3/4}$ should
approach $\xmom{1}/\xmom{0}$ when plotted as a function of $a$ for a
fixed $z$.  Analogously, the second cumulant (normalised) should
behave as
\begin{equation}
  \frac{\kmom{2}-\kmom{1}^2}{n^{3/2}}\to \frac{\xmom{2}}{\xmom{0}}
\end{equation}
where the $\xmom{m}$ were defined in section~\ref{sec:moments}. Note
that for a fixed $z$ the $\xmom{m}$ now depend only on $a$. The third
and fourth cumulants of the absolute magnetisation, divided by
respectively $8\,n^{9/4}$ and $16\,n^3$, quite analogously approach
their corresponding limits
\begin{equation}
  \frac{\xmom{3}}{\xmom{0}}-3\,\frac{\xmom{1}\,\kmom{2}}{\xmom{0}^2} +
  2\,\frac{\xmom{1}^3}{\xmom{0}^3}
\end{equation}
and 
\begin{equation}
  \frac{\xmom{4}}{\xmom{0}} - 4\,\frac{\xmom{1}\,\kmom{3}}{\xmom{0}^2}
  - 3\,\frac{\kmom{2}^2}{\xmom{0}^2} +
  12\,\frac{\xmom{1}^2\,\xmom{2}}{\xmom{0}^3} -
  6\,\frac{\xmom{1}^4}{\xmom{0}^4}
\end{equation}
Through a simple scaling analysis based on our sampled data we have
found that the normalised third cumulant has a limit maximum of about
$0.0205$ and a minimum of $-0.0500$. The fourth normalised cumulant
has a limit maximum of $0.0229$ and a minimum of $-0.0278$, based upon
our sampled data. Choosing $z=-9.87$ puts the maximums and minimums of
the limit curves at appropriate values. Now we identify the coupling
$K$ where the minimum of the fourth cumulant occurs with the point $a$
where the minimum of the corresponding limit curve occurs and likewise
for the maximum, thus providing us with a rescaling translating $K$
into $a$.  In figure~\ref{fig:5d-cum12} and \ref{fig:5d-cum34} the
first four cumulants are shown together with their estimated limit
curves for $z=-9.87$. Indeed the red curve may provide us with a
limit.

\begin{figure}[!ht]
  \begin{minipage}{0.5\textwidth}
    \begin{center}
      \includegraphics[width=0.99\textwidth]{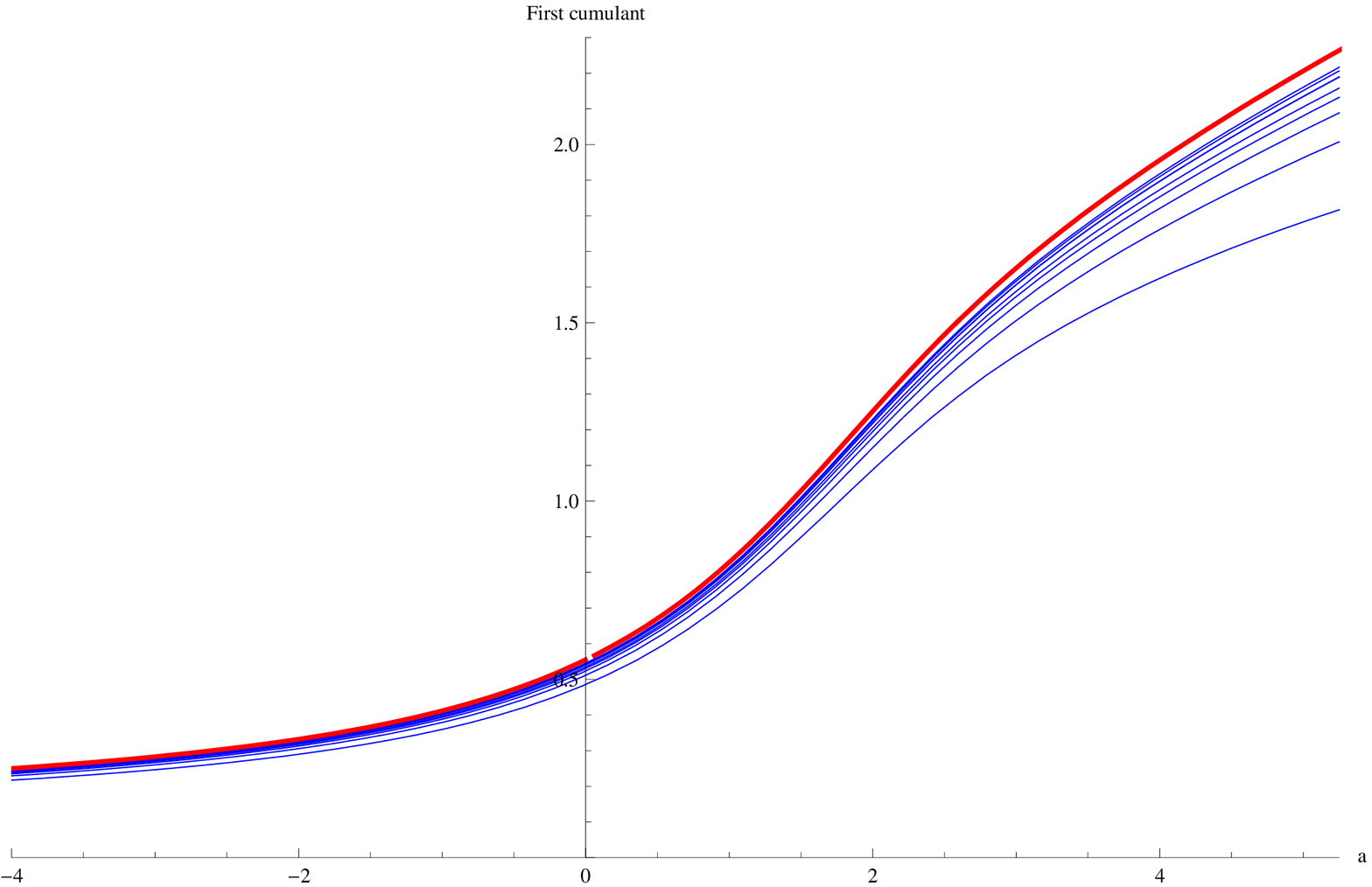}
    \end{center}
  \end{minipage}%
  \begin{minipage}{0.5\textwidth}
    \begin{center}
      \includegraphics[width=0.99\textwidth]{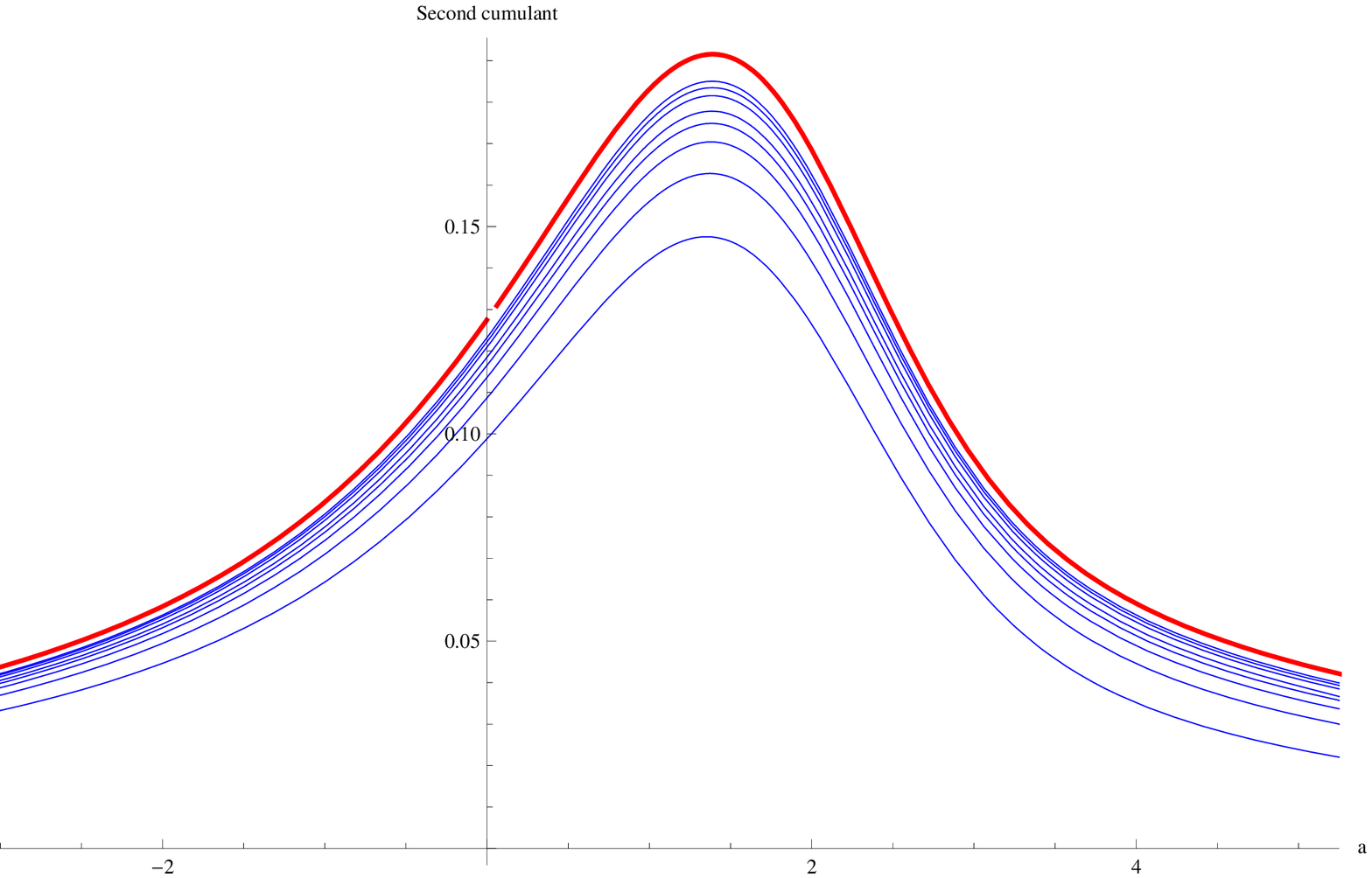}
    \end{center}
  \end{minipage}
  \caption{Normalised first (left) and second (right) cumulants for
    the $L\times L\times L\times L\times L$-lattice,
    $L=4,6,8,10,12,16,20,24$ (blue) versus $a$ for $z=-9.87$ together
    with the limit curve (red).}
  \label{fig:5d-cum12}
\end{figure}

\begin{figure}[!ht]
  \begin{minipage}{0.5\textwidth}
    \begin{center}
      \includegraphics[width=0.99\textwidth]{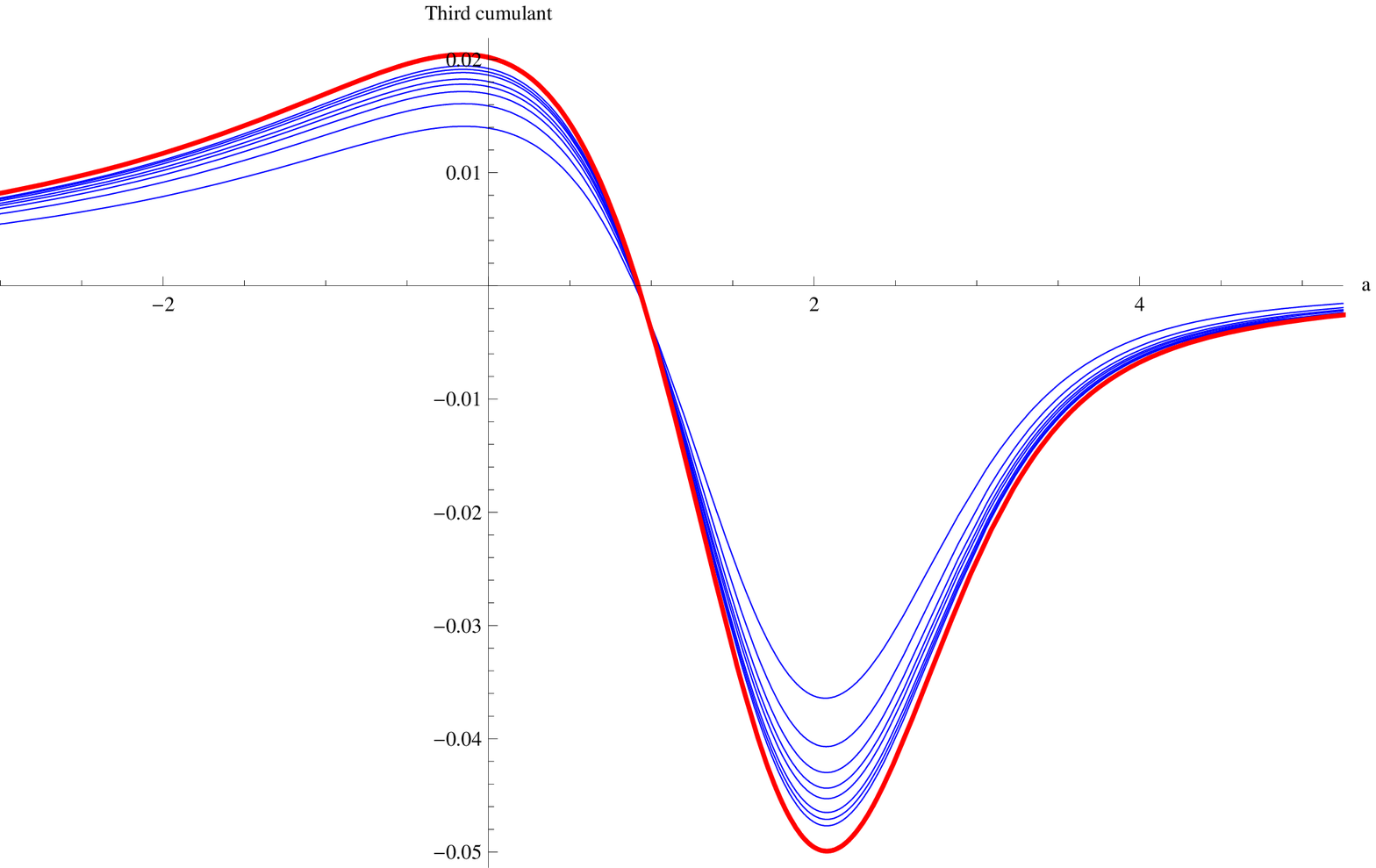}
    \end{center}
  \end{minipage}%
  \begin{minipage}{0.5\textwidth}
    \begin{center}
      \includegraphics[width=0.99\textwidth]{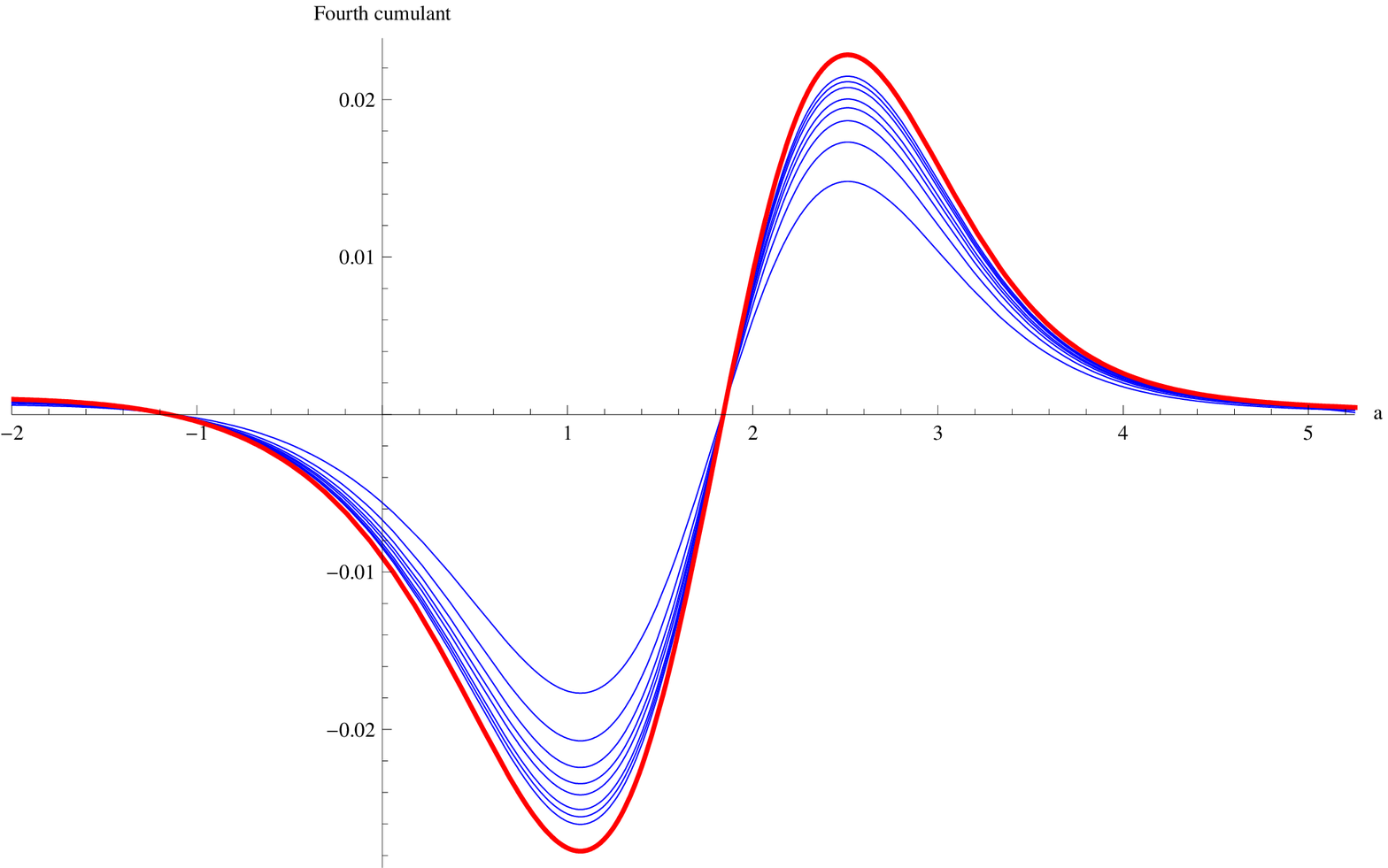}
    \end{center}
  \end{minipage}
  \caption{Normalised third (left) and fourth (right) cumulants for
    the $L\times L\times L\times L\times L$-lattice,
    $L=4,6,8,10,12,16,20,24$ (blue) versus $a$ for $z=-9.87$ together
    with the limit curve (red).}
  \label{fig:5d-cum34}
\end{figure}

Given a lattice size $L$ we denote by $K_{\min}(L)$ the location of
the minimum fourth cumulant and by $K_{\max}(L)$ the location of the
maximum.  Analogously for the limit curve, given a $z$ we denote by
$a_{\min}(z)$ and $a_{\max}(z)$ the location of the minimum and
maximum fourth cumulant.  For $z=-9.87$ we have $a_{\min}\approx
1.06965$ and $a_{\max}\approx 2.51275$. A simple scaling projection
gives that roughly $K_{\max}(L) \approx K_c + 0.22/L^{5/2}$ and
$K_{\max}(L)-K_{\min}(L)\approx 0.093/L^{5/2}$. Also $K_c\approx
0.113915$, see \cite{luijten:99}.  Thus, in principle at least, the
rescaling between $a$ and $K$ is
\begin{equation}
  K(a) \sim \frac{K_{\max}(L)-K_{\min}(L)}{a_{\max}(z)-a_{\min}(z)}\,
  \left(a-a_{\max}(z)\right) + K_{\max}(L)
\end{equation}
However, this kind of expression is somewhat too simplistic to get
figure~\ref{fig:5d-cum34}. It would take higher-order
corrections to scaling to produce it but this would probably take a
more involved numerical study of the 5D-model.  Other investigations
of the 5D-lattice includes e.g. \cite{brezin:85}, \cite{mon:96} and
\cite{luijten:99}.

\section{Conclusions}
The magnetisation distribution for the complete graph is exactly
described by the $p,q$-binomial distribution, corresponding to the
special (or limit) case of $p=q$.  For balanced complete bipartite
graphs this is most likely also true in some limit sense, yet to be
made precise. Actually, it appears that for most graphs, at least
those which are more or less regular, the magnetisations are
well-fitted by a $p,q$-binomial distribution for some choice of $p$
and $q$. The exact extent to which the $p,q$-binomial approximation is
\emph{good} we do not yet know (e.g. convergence in moment) nor the
exact class of graphs that would satisfy this. We have investigated
the matter more closely for lattices of dimension one through five. In
general they are always well-fitted by $p,q$-binomial distributions
for high- and low-temperatures but the problems arise near $K_c$, or
rather $K^*$ where the distribution changes from unimodal to bimodal.

For the 1-dimensional lattices (having no such bounded $K^*$) the
situation is basically always that of high temperatures. It seems
possible to give expressions for $p$ and $q$ in terms of $K$ in this
case though we have not done so. For 2-dimensional lattices the
distributions near $K^*$ are least well-fitted by the $p,q$-binomials
but slightly better fitted in the 3-dimensional case. We made
theory-based predictions of how $z$ should scale with $n$ near $K^*$.
Unfortunately, scaling is probably very slow, involving logarithms and
double logarithms, making it near impossible to test the
prediction. For 4-dimensional lattices the distributions are clearly
much better fitted by $p,q$-binomials, though some discrepancy still
remains just above $K^*$. For 5-dimensional lattices even this small
discrepancy is gone, leaving us perfectly fitted (that is, to the
human eye) $p,q$-binomial distributions. In this case the values of
$z$ at $K^*$ should approach a limit value. We estimated this limit
and, using this limit value, compared the first four normalised
cumulants for finite lattices with the (possible) limit curves.

We described and used a rather simple method to determine $p$ and $q$
given a distribution. Possibly this method is not optimal since it
simply forces the distribution to be correct at a single point rather
than providing a good overall-fit. It is also sensitive to noise when
the distributions are unimodal, thus making it difficult to determine
$p$ and $q$. On the other hand it works extremely well for bimodal
distributions where the noise sensitivity problem vanishes.

The $p,q$-binomial coefficients are just a tweaked form of
$q$-binomials, i.e. they are multiplied by a power of $p$. It is
possible that a different choice of factor would produce better
results in the case of 2- and 3-dimensional lattices.

We believe that what is said here for the Ising model also goes for
other models, i.e. the magnetisation distribution for quantum spin
models or for spin-glass models can be modeled by $p,q$-binomial
distributions.

\begin{acknowledgments}
  One of the authors (AR) wishes to thank the Swedish Research Council
  (VR) for financing this work.  This research was conducted using the
  resources of High Performance Computing Center North (HPC2N).
\end{acknowledgments}

\bibliographystyle{apsrev}

\end{document}